\documentclass{article}

\usepackage{amssymb,amsmath,amsfonts,amsthm}
\usepackage{cite}
\usepackage{graphicx}
\usepackage{bbold}
\usepackage{amsmath}
\usepackage{bbm}
\usepackage{graphicx}
\usepackage{pst-all}
\usepackage{graphicx}
\usepackage{graphics}
\usepackage{epsfig}
\usepackage{multicol}
\usepackage{color}
\usepackage{epsf}
\usepackage{changepage}

\newcommand{\1}{{\rm 1\hspace*{-0.4ex}%
\rule{0.1ex}{1.52ex}\hspace*{0.2ex}}}

\usepackage[T1]{fontenc}
\usepackage[english]{babel}
\usepackage[latin2]{inputenc}

\newtheorem{lemma}{Lemma}

\newtheorem{theorem}{Theorem}
\newtheorem{conjecture}{Conjecture}
\newtheorem{proposition}{Proposition}

\renewcommand{\b}[1]{\mathbf{#1}}

\renewcommand{\c}[1]{\mathcal{#1}}

\renewcommand{\r}[1]{\mathrm{#1}}
\newcommand{\s}[1]{\mathsf{#1}}

\newcommand{\Rl}{\mathbbm{R}}

\newcommand{\idty}{\1}
\DeclareMathOperator{\id}{id}

\DeclareMathOperator*{\tr}{Tr}
\DeclareMathOperator*{\map}{map}
\newcommand{\<}{\langle}
\renewcommand{\>}{\rangle}
\providecommand{\abs}[1]{|#1|}
\providecommand{\norm}[1]{\Vert #1 \Vert}

\title{Entropy of quantum channel in the theory of quantum information}
\author{Wojciech Roga}
\date{\today}

\begin{document}

\begin{center}
{\Large\textbf{Entropy of quantum channel in the theory of quantum information}}\\[10pt]
 (PhD Thesis)\\[10pt]

{Wojciech Roga}\\
Instytut Fizyki im.~Smoluchowskiego, Uniwersytet Jagiello{\'n}ski, PL-30-059 Krak{\'o}w, Poland\\
wojciech.roga@uj.edu.pl\\[20pt]

\end{center}

\abstract{
Quantum channels, also called  quantum operations, 
are linear, trace preserving and completely positive transformations 
in the space of quantum states. 
Such operations describe discrete time evolution of an open 
quantum system interacting with an  environment.
The thesis contains an analysis of properties of quantum channels 
and different entropies used to quantify the decoherence introduced into the system by a given operation.

Part \ref{partone} of the thesis provides a general introduction to the subject. 
In Part \ref{parttwo}, the action of a quantum channel is treated
as a process of preparation of a quantum ensemble. 
The Holevo information associated with this ensemble is shown to be bounded by
the entropy exchanged during the preparation process between the initial state and the environment.
A relation between the Holevo information and the entropy of  
  an auxiliary  matrix consisting of square root fidelities between the elements of the ensemble is proved in some special cases. 
Weaker bounds on the Holevo information are also established.

The entropy of a channel, also called the map entropy, is defined as the entropy of the state 
corresponding to the channel by the Jamio{\l}kowski isomorphism.
In Part \ref{partthree} of the thesis, the additivity of the entropy of a channel is proved.
The minimal output entropy, which is difficult  to compute, is estimated by
an entropy of a channel which is much easier to obtain. 
A class of quantum channels is specified, for which  additivity of channel capacity
is conjectured.

The last part of the thesis contains characterization of 
Davies channels, which correspond to an interaction of a state
with a thermal reservoir in the week coupling limit, under the condition
of quantum detailed balance and independence of rotational and dissipative evolutions.
The Davies channels are characterized for one--qubit and one--qutrit systems.
}

\newpage
 \medskip\medskip\medskip
 {\Large \textbf{Acknowledgements}}
 \medskip\medskip\medskip
 
 I would sincerely like to thank my supervisor
 Professor Karol \.Zyczkowski for motivation and support 
 in all the time of research and writing of this thesis.
 I would like to express my gratitude to Professor Mark Fannes
 for working together on diverse exciting projects. 
 Special thanks to my fellow-worker and friend Fernando de Melo.
 It is also pleasure to thank Professor Ryszard Horodecki, Professor Pawe\l{} Horodecki,
 Professor Micha\l{} Horodecki and Professor Robert Alicki for many
 opportunities to visit National Quantum Information Centre of Gda\'{n}sk 
 and helpful discussions.
 I would like to show my special gratitude to Professor Tomasz Dohnalik and Professor Jakub Zakrzewski from
 the Atomic Optics Department for the support and trust in me.
 I would like to thank my colleagues Piotr Gawron, Zbigniew Pucha\l{}a, Jaros\l{}aw Miszczak,
 Wojciech Bruzda, \L{}ukasz Skowronek and Marek Smaczy\'{n}ski for fruitful collaboration. 
 
 \newpage

\newpage

\tableofcontents

\newpage

\part{Introduction}\label{partone}
\section{Preliminary information}
\subsection{Preface}
It is not easy to give a satisfactory definition 
of information in sense in which this word is used in everyday life. 
For instance one could ask, how much information is contained in an allegorical baroque painting of Vermeer. 
There exist, of course, many interpretations and therefore, many kinds of information concerning this picture.
However, nowadays we are willing to distinguish some sort of information 
necessary to communicate a message 
independently on the interpretation. 
Due to our experience with computers we are used to problems 
how to encode the information into a string of digital symbols, 
transmit it and decode  it in order to obtain  the original message in another place. 
Imagine that we need to send the information contained in the Vermeer's picture. 
We have to encode it into digital data, transmit it 
to the other place and recover the picture on the screen of the receiver's computer. 
In a sense we send almost all the information without  knowing what interpretations it may carry. 

The problem rises what is the minimal amount of information 
measured in binary digits that enable the receiver to reliably recover the original message.  
In considered example we can divide the image of the Vermeer's picture into small pieces, 
decode colours of each piece into digital strings and transmit the description of colours one after another. 
However, we can also save some amount of digits when we menage  
to describe shapes of regions of the same colours in the picture 
and send only information about colours, shapes and patterns.  
How to do that in the most efficient way? This is a major  problem for experts 
 working on the information theory and computer graphics.
Some rules of the optimal coding were used intuitively during construction of the Morse alphabet. 
The letters which occur in the  English language more frequently are encoded by a smaller amount of symbols. 

In communication and computer sciences the problem 
of data compression is a subject of a great importance. 
To what extend the data can be compressed to still remain useful? 
Claude Shannon worked 
on the problem of transmission of messages through telecommunication channels. 
In 1958 he published his famous paper \cite{Shannon} opening the new 
branch of knowledge known as the theory of information. 
In this theory a message is composed of letters occurring with specified frequencies related to
probabilities. Every letter of a message can be encoded as a string of 
digital units. Every digital unit can appear in one of $r$ possible configurations.
Shannon found what is the minimal average 
amount of digital units per symbol which encodes a given message. 
This smallest average number of digital units is related to the information contained in the message 
and is characterized by a function of the probability distribution $P=\{p_1,...,p_n\}$ of letters, now called
the {\it Shannon entropy},
\begin{equation}
H(P)=-\sum_{i=1}^n p_i\log_r p_i,
 \label{shannon}
\end{equation}
where $0\log_r0\equiv 0$, $n$ is a number of letters and
the base of the logarithm $r$ characterizing the amount of configurations of
a chosen digital unit can be chosen arbitrary.
If the base is equal to $2$, the unit of entropy is called {\it binary unit} or {\it bit}.

%
The idea of efficient coding concerns in replacing more frequent letters
by means of a smaller amount of bits.
Shannon treated the message as a sequence of letters generated independently 
according to the probability distribution $P$ specified for a given language.
The original reasoning of Shannon proceeds as follows.
There are so many possible messages as the amount of
 typical sequences of letters with a  given probability distribution in the string of length $k\rightarrow\infty$.
Atypical sequences such as strings of letters $a$ repeated $k$
times are unlikely and are not taken into account.
The amount of possible messages is given
by the amount of typical sequences, which is of order of  $2^{kH(P)}$
if the base of the logarithm is equal to 2. 
This number is justified by methods of combinatorics.
Hence, every typical message of length $k$ can be represented by a
string of bits of size $kH(P)$. Therefore, the entropy
$H(P)$ can be interpreted as the smallest average amount
of bits per letter needed to reliably encode each typical message.

The information theory treats, as well, the information as a measure of uncertainty about the outcome of a random experiment.
Looking for a function which is suitable as a measure 
of the uncertainty about the concrete result of experiment, provided the 
probabilities of all experimental outcomes are given, Shannon formulated
a few postulates for the information measure \cite{Shannon}:
\begin{itemize}
\item It is a continues function of the probability distribution.
\item If all events are equally likely the function of uncertainty is an increasing function of their number.
\item If one of the events is split into two, the new function of uncertainty is equal to the sum of the original uncertainty
and the uncertainty of the new division weighted by the probability of the divided event.
\end{itemize}
The only function which satisfies these postulates is the Shannon entropy $H(P)$.
Therefore, the uncertainty or lack of information on the outcome of 
an experiment is the second interpretation of the entropy $H(P)$.

Taking a  weaker set of axioms allows one to generalize the definition of the measure of uncertainty and 
to find other functions of probability vector $P$, which in special case converge to
the Shannon entropy (\ref{shannon}). 
For instance, R\'{e}nyi introduced one parameter family of generalized entropy functions. 
Since,  
the information of an experiment consisting of two independent experiments should
be given by the sum of the information gained in both experiments,
the measure of information should be additive. 
The Shannon entropy of the joint probability distribution of two independent variables 
is additive. 
R\'{e}nyi noticed \cite{Renyi102} that the additivity of information 
is not equivalent to the third postulate of Shannon. However, 
if one replaces the third postulate by additivity of information of independent events, yet
another axiom should be postulated to obtain back the Shannon's formula (\ref{shannon}).
This additional postulate specifies the way of calculating the mean values.
If one considers the linear mean, the Shannon entropy is singled out by this 
set of postulates.
However, other definition of the mean value also can be taken.
In consequence, the new set of postulates implies an one parameter family of generalized entropy functions known as the {\it R\'{e}nyi entropy} of order $q$:
\begin{equation}
H_q(P)=\frac{1}{1-q}\log\sum_{i=1}^np_i^q.
\end{equation}
Here, $q$ denotes the free parameter depending on the definition of the average.
Another generalization of entropy function was analysed by Tsallis \cite{tsallis102, tsallis103}. 
The Tsallis entropy of order $q$ is defined as follows,
\begin{equation}
T_q=\frac{1}{q-1}(1-\sum_i^n p_i^q).
\end{equation}


Hence the information theory concerns entropies, however,
it also investigates  communication sources and communication channels
which can introduce some errors to messages.
Information theory defines such quantities as the {\it relative entropy} and the {\it mutual information} \cite{Shannon}. 
Using these concepts the {\it channel capacity} is defined. It is the maximal rate of 
 information which can be reliably decoded after passing through the
channel. The channel capacity is measured in bits per a unit of time.

The theory of quantum information, 
which considers quantum systems as carriers of information,
 should enable one to generalize the notions 
of classical information theory such as the channel capacity.  
To  formulate a quantum counterpart of the Shannon 
concepts such as the relative entropy or channel capacity
the theory of open quantum systems,  
quantum statistical processes, statistical operators, density matrices, 
partial traces and generalized measurements should be applied. 
In the early stage of the stochastic theory of open quantum systems, it  was developed by 
Davies \cite{daviesstoch},  and Kossakowski \cite{kossakowski}. 
Moreover, other important results on 
accessible information transmitted 
through a noisy quantum channel were obtained by Holevo  \cite{holevo}.

There are many advantages of using quantum resources to transform and 
transmit the information \cite{nielsen}. 
In particular, there exist a famous protocols  of superdense coding \cite{bennett} of information into a quantum carrier. 
Furthermore, some computational problems can be solved in framework of the quantum information processing faster than classically 
\cite{deutsh,grover,shorfactor}. 
Quantum world gives also new communication protocols like 
quantum teleportation \cite{bennett,hortele} which is possible due to quantum entanglement   \cite{horo4,KZ}. 
In quantum case, entangled states 
can increase the joint capacity of two channels 
with respect
to the sum of the two capacities \cite{hastings,horo,kingostatni}.
Also a new branch of cryptography was developed due to the quantum theory \cite{bennettbrassard}.
Although, these new possibilities are promising, manipulation of quantum resources is difficult in practice. 
In particular, the quantum states carrying the information are very sensible to noise, which can completely destroy 
 quantum information. 
Moreover, probabilistic nature of quantum theory does not 
allow us to extract uniquely the information from quantum sources.
Many restrictions and laws of quantum
information theory
are formulated in terms of inequalities of quantum entropies.
The most significant quantum entropy is the one defined by von Neumann \cite{wehrl,ruskaijakies},
which is a counterpart of the Shannon entropy.
However, the other quantum entropies such like the R\'{e}nyi entropy \cite{Renyi102}
or Tsallis entropy are also considered \cite{tsallis102,tsallis103,Renyitsallis}.



The issue of transmitting  a classical information
through a noisy channel is an important issue in the information
theory \cite{Shannon,shannon2,hartley}.
Among problems concerning the information channels
one can specify the following questions: 
How to encode the
information  in order to transmit it reliably through the channel in the most efficient way \cite{Shannon,schumacher}?
What is the maximal rate of
the information transmission? What is the capacity
of a given communication channel \cite{desurvire,holevocapa,schumwest,szuma}? Which states are the
most resistant to errors occurring in the a channel \cite{fuchs,hayashi}?
What are the efficient strategies of the error correction \cite{knill}?

Similar questions can also be formulated  in the framework of quantum 
 information theory.
The quantum channels, 
also called {\it quantum operations},
are transformations 
in the set of states \cite{choi,jam,kraus2,depillis}.
They describe evolution 
of an open quantum system interacting with
an environment in discrete time.


The set of all quantum channels 
is still not completely understood. 
Merely the set of one--qubit channels
is satisfactory explored \cite{fujivara,ruskaiszarek}. 
However, even in this simplest case
some interesting problems are open.
For instance, it is not clear, whether the capacity of one--qubit channels is additive \cite{kingostatni}.
Another approach to  quantum channels
suggests to analyse only certain special classes of them, motivated by some  
physical models \cite{shor,amosov,king2,ruskai,lindbl}.
 
Quantum channels are also formally related to 
measurement processes in quantum theory \cite{kraus2,breuer}. 
As a measurement process changes the quantum
state and in general cannot perfectly
distinguish measured states, there is a fundamental
restriction on the information which can be obtained from
the message encoded into quantum states \cite{holevo}. 
These restrictions are also formulated in terms
of entropies. 

The different aspects of quantum channels
mentioned above suggest that entropies which characterize the channels
play an important role in the information theory.
This thesis is devoted to investigation of
quantum channels and some entropies used to characterize them:
the minimal output entropy \cite{shor,kingostatni}, the map entropy \cite{roga0,verstaete,ziman}
and the exchange entropy \cite{szuma}.

\subsection{Structure of the thesis}

The thesis is partially
based on results already published in articles \cite{roga0,roga,roga2,roga3,roga5,roga4}, which are
enclosed at the end of the thesis.
In a few cases some issues from these papers are discussed here in a more detailed manner. 
The thesis contains also some new, unpublished results and technical considerations not included in the published articles.

The structure of the thesis is the following. 
The thesis is
divided into three parts. The first part is
mostly introductory and contains a short review of the literature. 
This part provides basic 
information useful in the other parts of the thesis and fixes 
notation used in the entire work.
In part \ref{partone} only the result from 
Section \ref{sectioncomp} concerning the Kraus representation
of a  complementary channels 
and Section \ref{caneleone} on the Kraus operators constructed for
an ensemble of states
are obtained by the author.

Part \ref{parttwo} contains results 
based on papers \cite{roga,roga5,roga0},
not known before in the literature.
However, some results not published previously are also analysed there.

Chapter \ref{sec:prl} contains the most important result of the
thesis --  the inequality between  the Holevo information 
related to an ensemble of quantum states and the entropy of
the state of environment taking part  in preparation of the ensemble.
As the entropy of the environment can be 
treated equivalently as the entropy of an output of the complementary channel,
or the entropy of a correlation matrix, or the entropy of a Gram matrix 
of purifications of mixed states, or as the entropy exchange, this relation
might be considered as a new and universal
result in the theory of quantum information. 
Consequences of this inequality  
have not been analysed so far.
Chapter \ref{sec:prl} contains also the discussion of the particular cases for which
the inequality is saturated. This result has not 
been published before. 
Section \ref{los} describes proofs of known entropic inequalities
which are related to the bound on the Holevo quantity. Some new and unpublished consequences of 
these inequalities are presented in Section \ref{consek}.
Original, new results are also contained in Sections \ref{linneusz} and
\ref{relat}.

Part \ref{parttwo} contains, moreover, 
the conjecture on the inequality between the Holevo information
of a quantum ensemble and the entropy of the matrix of square root of fidelities.
Several weaker inequalities are analysed here in a greater detail than it was done in \cite{roga5}. 
Section \ref{sadelko} presents a confirmation of the conjecture for a special
class of ensembles of quantum states.

Part \ref{partthree} of the thesis is based on the
results presented in \cite{roga3,roga2}.
Article \cite{roga3} described partially in Chapter \ref{sec:ent}
considers the relation between minimal output entropy and
the map entropy.
Section \ref{depol} contains a proof of additivity of 
the map entropy with respect to the tensor product of two maps, already published in our work \cite{roga3}.
These results allow us to specify a class of quantum channels for which
 additivity of the minimal output entropy is conjectured.

The Davies maps acting on one--qubit and one--qutrit quantum systems
are analysed in Chapter \ref{bida}.
Conditions for the matrix entries of a quantum operation
representing a Davies map are given along the lines formulated in our work \cite{roga2}.
Multiplicativity of the maximal output norm of one--qubit Davies maps, 
entirely based on the analogical proof for 
bistochastic maps \cite{king}, is presented in Section \ref{multipy}.
However, this result cannot be treated as a new one, since
multiplicativity of the maximal output two norm was proved earlier for all one--qubit quantum channels \cite{kingostatni}. 
Section \ref{qut} contains graphical representations 
of stochastic matrices of order three which correspond to quantum Davies maps,
which has not been 
published yet.



\subsection{A short introduction to quantum mechanics}

The formalism of quantum mechanics
can be derived from a few postulates (axioms) which
are justified by experiments. The set of axioms defining the
quantum theory differs depending on the author \cite{agarval}.
However, some features occur common in
every formulation, either as axioms or as their consequences.
One of such key features is the {\it superposition principle}.
It is justified by several experimental data as 
interference pattern in double slit experiment with 
 electrons or interference of a single photon in 
the Mach--Zender interferometer \cite{hong}.
The superposition principle states that
the state of a quantum system, 
which is denoted in Dirac notation by $|\psi\>$,
 can be represented by a coherent combination
of several states $|\psi_i\>$ with complex coefficients $a_i$,
\begin{equation}
|\psi\>=\sum_ia_i|\psi_i\>.
\label{lada}
\end{equation}
The quantum state $|\psi\>$ of an $N$ level system is represented by
a vector from the complex Hilbert space $\c H_N$. The inner product
$\<\psi_i|\psi\>$ defines the coefficients $a_i$ in (\ref{lada}).
The square norm of $a_i$ is interpreted as the probability
that the system described by $|\psi\>$ is in the state $|\psi_i\>$.
To provide a proper probabilistic interpretation a 
vector used in quantum mechanics is normalized by the condition $\<\psi|\psi\>=||\psi||^2=\sum_i|a_i|^2=1$.

Quantum mechanics is a probabilistic theory.
One single measurement does not provide much
information on the prepared system.
However, several measurements on identically prepared
quantum systems allow one to characterize
the quantum state. 

A physical quantity is represented
by a linear operator called an {\it observable}. 
An observable $A$ is a Hermitian operator, $A=A^{\dagger}$, 
which can be constructed by a set of 
real numbers $\lambda_i$ 
(allowed values of the physical quantity) 
and a set of states $|\varphi_i\>$ determined by the measurement, $A=\sum_i\lambda_i|\varphi_i\>\<\varphi_i|$.
The physical value corresponds to the average of the observable in
the state $|\psi\>$,
\begin{equation}
\<A\>_{\psi}=\sum_i\lambda_i|\<\psi|\varphi_i\>|^2=\<\psi|A|\psi\>.
\end{equation}
One can consider
the situation in which a state $|\psi\>$ is not known exactly. Only
 a statistical mixture of several quantum states $|\phi_i\>$
which occur with probabilities $p_i$ is given.
In this case the average value of an observable has the form
\begin{equation}
\<A\>_{\{p_i,\phi_i\}}=\sum_ip_i\<\phi_i|A|\phi_i\>,
\end{equation}
which can be written in terms of an operator on $\c H_N$ called a {\it density matrix} 
$\rho=\sum_ip_i|\phi_i\>\<\phi_i|$
as
\begin{equation}
\<A\>_{\{p_i,\phi_i\}}=\tr \rho A.
\end{equation}
A density matrix describes a so called {\it mixed state}.
In a specific basis the density matrices characterizing an $N$ level quantum system
 are represented by 
$N\times N$ matrices $\rho$ which are Hermitian, have trace equal to unity and are positive.
Let us denote the set of all such matrices  by $\c M_N$, 
\begin{equation}
\c M_N=\{\rho: {\rm dim}\rho=N, \rho=\rho^{\dagger}, \rho\geq 0, \tr\rho=1 \}.
\end{equation}
This set is convex. Extremal points of this set are formed by 
projectors of the form $|\psi\>\<\psi|$ called {\it pure states}, 
which correspond to vectors $|\psi\>$ of the Hilbert space. 

The state of composed quantum system which consists of 
one $N_1$--level system and one $N_2$--level system
is represented by a vector of size $N_1N_2$ from 
the Hilbert space which has a tensor product structure,
$\c H_{N_1N_2}=\c H_{N_1}\otimes\c H_{N_2}$.
Such a space contains also states which cannot be 
written as tensor products 
of vectors from separate spaces,
\begin{equation}
|\psi_{12}\>\neq|\psi_1\>\otimes|\psi_2\>.
\label{spaltane}
\end{equation}
and are called {\it entangled states}.
States with a tensor product structure are called {\it product states}.
If the state of only one subsystem is considered one
has to take an average over the second subsystem.
Such an operation is realized by taking the partial trace over 
the second subsystem and leads to 
a reduced density matrix,
\begin{equation}
\rho_1={\rm Tr}_{2}\rho_{12}.
\end{equation}
A density matrix describes therefore the state of an {\it open quantum system}.

The evolution of a normalized vector in the Hilbert space is
determined by a unitary operator $|\psi'\>=U|\psi\>$. 
The transformation $U$ is related to Hamiltonian evolution
due to the Schr\"{o}dinger equation,
\begin{equation}
i\hbar\frac{d}{dt}|\psi\>=H|\psi\>,
\label{zbik}
\end{equation}
where $H$ denotes the Hamiltonian operator of the system, while $t$ 
represents time and $2\pi\hbar$ is the Planck constant.
A discrete time evolution of an open quantum system characterized by a density operator $\rho$ is described 
by a {\it quantum operation} which will be considered in Chapter \ref{trombalski}.


According to a general approach to quantum measurement \cite{kraus2,daviesmierzy},
it can be defined by a set of $k$ operators $\{E^i\}_{i=1}^k$ forming a {\it positive operator valued measure} (POVM).
The index $i$ is related to a possible measurement result, for instance 
the value of the measured quantity. The operators $E^i$ are positive
and satisfy the identity resolution,
\begin{equation}
\sum_{i=1}^kE^i=\idty.
\label{kroliczek}
\end{equation}
The quantum state is changing during the measurement process.
After the measurement process that gives the outcome $i$ as a result, 
the quantum state $\rho$ is transformed into
\begin{equation}
\rho'_i=K^i\rho K^{i\dagger}/\tr(K^i\rho K^{i\dagger}),
\end{equation}
where $K^{i\dagger}K^i=E^i\geq 0$.
The probability $p_i$ of the outcome $i$ is given by
$p_i=\tr(K^i\rho K^{i\dagger})$.
Due to relation (\ref{kroliczek}), the probabilities of all outcomes
sum up to unity.

A quantum state characterizing a $2$--level system
is called {\it qubit} and its properties are discussed in more
detail in Section \ref{richibruk}.

\subsection{Schmidt decomposition}

The theorem known as {\it Schmidt decomposition} \cite{schmidt}
provides a useful representation of a pure state of a bi--partite quantum system.
\begin{theorem}[Schmidt]
Any quantum state $|\psi_{12}\>$ from the Hilbert space composed 
of the tensor product of two Hilbert spaces $\c H_{1}\otimes \c H_2$ of dimensions 
$d_1$ and $d_2$, respectively, can be represented
as 
\begin{equation}
|\psi_{12}\>=\sum_{i=1}^{d}\lambda_i|i_1\>\otimes|i_2\>,
\end{equation}
where $\{|i_1\>\}_{i=1}^{d_1}$ and $\{|i_2\>\}_{i=1}^{d_2}$ are orthogonal basis of the Hilbert spaces
$ \c H_{1}$ and $ \c H_{2}$ respectively, and $d=\min\{d_1,d_2\}$.
\end{theorem}
\begin{proof}
Choose any orthogonal basis $\{|\phi^k_1\>\}_{k=1}^{d_1}$ of  $ \c H_{1}$ 
and any orthogonal basis  $\{|\phi^j_2\>\}_{j=1}^{d_2}$ of  $ \c H_{2}$.
In this product basis, the bi--partite state  $|\psi_{12}\>$ reads
\begin{equation}
|\psi_{12}\>=\sum_{0\leq k\leq d_1,\ 0\leq j\leq d_2}a_{kj}|\phi^k_1\>\otimes|\phi^j_2\>.
\end{equation}
Singular value decomposition of a matrix $A$ of size $d_1\times d_2$ with entries $a_{kj}$ 
gives $a_{kj}=\sum_i u_{ki}\lambda_i v_{ij}$. Here $u_{ki}$ and $v_{ij}$ are
entries of two unitary matrices,
while $\lambda_i$ are singular values of $A$.
Summation over indexes $k$ and $j$ cause changes of two orthogonal bases
into
\begin{eqnarray}
|i_1\>=\sum_k u_{ki}|\phi^k_1\>,\\
|i_2\>=\sum_j v_{ij}|\phi^j_2\>.
\end{eqnarray}
The number o nonzero singular values is not larger than the smaller one  of
the numbers $(d_1,d_2)$. 
\end{proof}
The Schmidt decomposition 
implies that both partial traces of any bi--partite pure state
have the same nonzero part of the spectrum:
\begin{eqnarray}
{\rm Tr}_1|\psi_{12}\>\<\psi_{12}|=\sum_{i=1}^d\lambda_i^2|i_2\>\<i_{2}|,\\
{\rm Tr}_2|\psi_{12}\>\<\psi_{12}|=\sum_{i=1}^d\lambda_i^2|i_1\>\<i_{1}|.
\end{eqnarray}
The Schmidt coefficients $\lambda_i$ are invariant under
local unitary transformations $U_1\otimes U_2$ applied to $|\psi_{12}\>$.
The number of non--zero coefficients $\lambda_i$ is called the {\it Schmidt number}.
Any pure state which has the Schmidt number greater than 1 is called {\it entangled state}.
A pure state for which all Schmidt coefficients $\lambda_i$ are equal to $1/\sqrt{d}$ is
called a {\it maximally entangled state}.
 
Another important consequence of the Schmidt decomposition is
that for any mixed state $\rho$ there is a pure state $|\psi\>$
of a higher dimensional Hilbert space such that 
 $\rho$ can be obtained by taking the partial trace,
\begin{equation}
\rho={\rm Tr}_1|\psi\>\<\psi|.
\end{equation}
Such a state $|\psi\>$ is called a {\it purification} of $\rho$.
The Schmidt decomposition gives the recipe for the purification procedure.
It is enough to take square roots of eigenvalues of $\rho$
 in place of $\lambda_i$ and its
eigenvectors in place of $|i_1\>$. Any orthogonal basis in $\c H_2$ provides
a purification of $\rho$, which can be written as
\begin{equation}
|\psi\>=\sum_i(U_1\otimes\sqrt{\rho})|i_1\>\otimes|i_2\>,
\label{paproc}
\end{equation}
where $U_1$ is an arbitrary unitary transformation and $\sqrt{\rho}|i_2\>=\lambda_i|i_2\>$.


\subsection{Von Neumann entropy and its properties}

Many theorems concerning the theory of quantum information can be
formulated in terms of the
 {\it von Neumann entropy} \cite{vonneumann} of a quantum state,
\begin{equation}
S(\rho)=-\tr\rho\log\rho,
\end{equation}
which is equivalent to the Shannon entropy (\ref{shannon})
of the spectrum of $\rho$.
The entropy characterizes the 
degree of mixing of a quantum state.
Assume that $\rho$ is a density matrix of size $N$.
The value of $S(\rho)$ is equal to zero
if and only if the state $\rho$ is pure. It gains its
maximal value $\log N$ for the maximally mixed
state $\rho_*=\frac{1}{N}\idty$ only. 

Von Neumann entropy has  also an important interpretation in
quantum information theory, as it plays the role similar to the Shannon entropy in the classical theory
of optimal compression of a message \cite{schumacher}. 
Let the letters $i$ of the message, which occur with probabilities $p_i$,  
be encoded into 
pure quantum states $|\psi_i\>$ from the Hilbert space $\c H_N$. 
Sequences of $k$ letters are encoded into a Hilbert space of dimension $N^k$.
A long message can be divided into sequences of size $k\rightarrow\infty$.
Among them one can distinguish sequences in typical 
subspaces and such which occur with negligible probability.
A unitary transformation applied to the sequence of quantum systems can
transmit almost all the information into a typical subspace.
The space of a typical sequence
has the smallest dimensionality required to 
encode the message reliably with negligible probability of an error. 
This smallest dimensionality 
per symbol is shown \cite{schumacher} to be equal to the von Neumann entropy 
of the state $\rho=\sum_ip_i|\psi_i\>\<\psi_i|$. Therefore,
quantum coding consists in taking states from the smaller subspace of dimension $2^{kS(\rho)}$
instead of a space of dimension $N^k$ to encode the same message.
If the state $\rho$ represents completely random set of states
there is no possibility to compress the message,
since $S(\rho)=S(\rho_*)=\log_2 N$, where logarithm is of base 2.
The entropy, therefore, describes the capability of compression
of the message encoded in a given set of states,
or the smallest amount of qubits needed to 
transmit a given message.

The von Neumann entropy, as the entropy of eigenvalues
of a density matrix, describes also the uncertainty of 
measuring a specific state from the set of the eigenvectors.
The most important properties of the von Neumann entropy
are \cite{wehrl}:
\begin{itemize}
\item The von Neumann entropy is a non negative function of any $\rho$.
\item It is invariant under unitary transformations, $S(\rho)=S(U\rho U^{\dagger})$.
\item It is a concave function of its argument, $\sum_{i=1}^kp_iS(\rho_i)\leq S(\sum_{i=1}^kp_i\rho_i)$, where
$p_i\geq 0$ for any $i$ and $\sum_i^{k}p_i=1$.
\item It is subadditive 
\begin{equation}
S(\rho_{12})\leq S(\rho_1)+S(\rho_2),
\end{equation}
where $\rho_{12}$ is a bi--partite state of a composite system and 
the partial traces read
$\rho_1=\tr_2\rho_{12}$ and $\rho_2=\tr_1\rho_{12}$.
\item The von Neumann entropy satisfies the relation of strong subadditivity \cite{liebruskai},     
\begin{equation}
S(\rho_{123})+S(\rho_2)\leq S(\rho_{12})+S(\rho_{23}),
\end{equation}
where the state $\rho_{123}$ is a composite state of three subsystems $(1,2,3)$ 
and the other states are obtained by its partial traces.
\end{itemize}

%
%
%
%
%
%

\subsection{Quantum channels and their representations}\label{trombalski}

One distinguishes two approaches 
to describe time evolution of an open
quantum system. One of them starts
from a concrete physical model defined by a given 
Hamiltonian
which determines the Shr\"{o}dinger equation (\ref{zbik}) 
or the master equation, \cite{breuer}. Solving them one may find the state of the 
quantum system at any moment at time. 
An alternative approach to the dynamics of an open quantum system 
relies on a stroboscopic picture and a discrete time evolution. It
starts from a mathematical construction
of a quantum map, $\rho'=\Phi(\rho)$, allowed by the general
laws of quantum mechanics. 
This approach is often 
used in cases in which the physical 
model of the time evolution is unknown. 
This fact justifies the name "black box" model to describe the
evolution characterized by a quantum map $\Phi$. 
Such a model is also considered  
if one wants to investigate
the set of 
all possible operations independently on whether the physical context is specified. 
Main features and some representations 
 of the map $\Phi$, which describes a "black box" model of non--unitary
 quantum evolution, are given below.

The quantum map $\Phi$ describes 
the dynamics of a quantum system $\rho$ which interacts with 
an environment.  It is given by a nonunitary  
quantum map $\Phi:\rho\rightarrow \rho'$.  
Any such map is completely positive, and trace preserving \cite{choi,jam,kraus2,depillis}.
"Complete positivity" means that an extended map $\Phi\otimes\idty_M$, which is a
trivial extension of $\Phi$ on the space of any dimension $M$, transforms the set of 
positive operators into itself. 
A completely 
positive and trace preserving quantum map is called {\it quantum operation} 
or {\it quantum channel}. 

Due to the theorem of Jamio{\l}kowski \cite{jam} and Choi \cite{choi}
the complete positivity of a map is equivalent
to positivity of a state corresponding to the 
map by the {\it Jamio{\l}kowski isomorphism}. 
This isomorphism determines the correspondence 
between a quantum operation $\Phi$ acting on $N$ dimensional 
matrices and density matrix $D_{\Phi}/N$ of dimension $N^2$ which 
is called Choi matrix or the Jamio{\l}kowski state
\begin{equation}
\frac{1}{N}D_{\Phi}=[\id_N\otimes\Phi]\big(|\phi^+\left.\right\rangle\left\langle\right. \phi^+|\big),
\label{miopio}
\end{equation}
where $|\phi^+\left.\right\rangle=\frac{1}{\sqrt{N}}\sum_{i=1}^N |i\left.\right\rangle\otimes|i\left.\right\rangle$ 
is a maximally entangled state.
The dynamical matrix $D_{\Phi}$ corresponding to a trace preserving 
operation satisfies the partial trace condition
\begin{equation}
{\rm Tr}_{2} D_\Phi=\idty.
\label{partialtrace}
\end{equation}

The quantum operation $\Phi$ can be represented as
 {\it superoperator matrix}. It is a matrix which acts on  
the vector of length $N^2$, which contains the entries $\rho_{ij}$ of the density matrix ordered
lexicographically. Thus the
superoperator $\Phi$ is represented by a square matrix of size $N^2$.
The superoperator in some orthogonal product basis $\{|i\>\otimes|j\>\}$
 is represented by a matrix indexed by  four indexes,
\begin{equation}
\Phi_{\!\!\!\!\footnotesize{
\begin{array}{cc}
i&\!\!\!\!\!\! j\\
k&\!\!\!\!\!\! l
\end{array}}
}=\<i|\otimes\<j|\Phi|k\>\otimes|l\>.
\end{equation} 
The matrix representation of the dynamical matrix is related to 
the superoperator matrix 
by the reshuffling formula \cite{KZ} as follows
\begin{equation}
\<i|\otimes\<j|D_{\Phi}|k\>\otimes|l\>=\<i|\otimes\<k|\Phi|j\>\otimes|l\>.
\label{zyrafa}
\end{equation}

To describe a quantum operation, one may use the Stinespring's dilation theorem \cite{sting}. 
Consider a quantum system, described by the state $\rho$ on $\c H_N$,  
interacting with its environment characterized by a state on $\c H_M$. 
The joint evolution of the two states is 
described by a unitary operation $U$.
Usually it is assumed that the joint state of the 
system and the environment is initially not entangled. Moreover, due
to the possibility to purification the environment, its 
initial state is given by a pure one.
The evolving joint state is therefore:
\begin{equation}
\omega=U\Big(\left|1\right\rangle\left\langle 1\right|\otimes \rho \Big)U^{\dagger},
\label{omom}
\end{equation}
where $|1\>\in \c H_M$ and $U$ is a unitary matrix of size $NM$.
The state of the system  after the operation is obtained by tracing out the environment, 
\begin{equation}
\rho'=\Phi(\rho)={\rm Tr}_M \Big[U\big(\left|1\right\rangle\left\langle 1\right|\otimes \rho\big) U^{\dagger}\Big]=\sum_{i=1}^M K^i\rho K^{i\dagger},
\label{quantop}
\end{equation}
where the Kraus operators read, $K^i=\left\langle i\right| U\left|1\right\rangle$. 
In matrix representation the Kraus operators
are formed by successive blocks of the first block--column of the unitary evolution matrix $U$. 
Here the state $\omega$ can be equivalently given as
\begin{equation}
\omega=\sum_{i,j=1}^MK^i\rho K^{j\dagger}\otimes\left|i\right\rangle\left\langle j\right|.
\label{stephany}
\end{equation}
A transformation $\rho\rightarrow\omega$ is obtained by an isometry $F:\c H_N\rightarrow\c H_{NM}$, where
\begin{equation}
F\left|\phi\right\rangle=\sum_i (K^i\left|\phi\right\rangle)\otimes\left|i\right\rangle.
\end{equation}
Due to the Kraus theorem \cite{kraus2}
any completely positive map $\Phi$
can be written in the Kraus form,
\begin{equation}
\rho'=\Phi(\rho)=\sum_{i=1}^M K^i\rho K^{i\dagger}.
\label{mrowkojad}
\end{equation}
The opposite relation is also true, any map of the Kraus form (\ref{mrowkojad}) is completely positive.


%

\subsubsection{Representation of a complementary channel}\label{sectioncomp}

Consider a quantum channel $\Phi$ described by the Kraus operators $K^i$, 
\begin{equation}
\Phi(\rho)={\rm Tr}_M\omega=\sum_{i=1}^{M}K^i\rho K^{i\dagger},
\end{equation}
where notation from Section \ref{trombalski} is used.
The channel $\tilde{\Phi}$ {\it complementary} to $\Phi$ is defined by
\begin{equation}
\tilde{\Phi}(\rho)={\rm Tr}_N\omega=\sum_{i=1}^{N}\tilde{K}^i\rho \tilde{K}^{i\dagger}
\end{equation}
and it describes the state of the $M$--dimensional environment
after the interaction with the principal system $\rho$.
One can derive the relation between operators $\{\tilde{K}^j\}_{j=1}^N$ 
and $\{K^i\}_{i=1}^M$ from the last equation by substituting $\omega$ as in (\ref{stephany}). This relation can be rewritten as
\begin{equation}
\sum_{i,j=1}^M(\tr K^i\rho K^{j\dagger})\left|i\right\rangle\left\langle j\right|=\sum_{i=1}^{N}\tilde{K}^i\rho \tilde{K}^{i\dagger}.
\label{sigmal}
\end{equation}
Comparison of the matrix elements of both sides gives
\begin{equation}
 \sum_{\alpha=1}^N\tilde{K}^{\alpha}_{im}\rho_{mn}\tilde{K}^{\alpha\dagger}_{nj}=\sum_{\alpha=1}^N K^{i}_{\alpha m}\rho_{mn}K^{j\dagger}_{n\alpha},
\end{equation}
where matrix elements are indicated by lower indexes
and the Einstein summation convention is applied.
Hence, for any quantum channel $\Phi$ given by a set of Kraus operators 
$K^i$, one can define the Kraus operators $\tilde{K}^{\alpha}$  representing the complementary channel $\tilde{\Phi}$ as
\begin{equation}
\tilde{K}^{\alpha}_{ij}=K^{i}_{\alpha j},\qquad i=1,...,M,\qquad j,\alpha=1,...,N.
\label{pantalon}
\end{equation}

\subsection{One--qubit channels}\label{richibruk}

One--qubit channels acting on density matrices of size $2$ 
have many special features which cause that the set of these
channels is well understood \cite{fujivara,ruskaiszarek,king}.
However, many properties of one--qubit maps are not shared
with the quantum maps acting on higher dimensional systems.
Since one--qubit quantum channels are often considered 
in this thesis, the following section presents a brief review of their basic properties.

A quantum two level state is called {\it quantum bit} or {\it qubit}.
It is represented by a $2\times 2$ density matrix.
Any Hermitian matrix of size two can be represented in the basis of 
identity matrix and the three Pauli matrices $\vec{\sigma}=\{\sigma_1,\sigma_2,\sigma_3\}$,
\begin{equation}
 \sigma_1=\begin{pmatrix}
           0&1\\
1&0
          \end{pmatrix},\qquad
\sigma_2=\begin{pmatrix}
           0&-i\\
i&0
          \end{pmatrix},\qquad
\sigma_3=\begin{pmatrix}
           1&0\\
0&-1
          \end{pmatrix}.
\label{malpa}
\end{equation}
One qubit state $\rho$ decomposed in the mentioned basis is given by the formula
\begin{equation}
\rho=\frac{1}{2}(\id+{\vec r}\cdot\vec{\sigma}),\quad {\vec r}\in \Rl^3.
\label{bobr}
\end{equation}
Positivity condition, $\rho\geq 0$, implies that $|\vec{r}|\leq 1$.
The vector $\vec{r}$ is called the {\it Bloch vector}. All possible Bloch vectors 
representing quantum states form the {\it Bloch ball}. Pure one--qubit states form a sphere of radius $|\vec{r}|=1$. 

Any linear one--qubit quantum operation $\Phi$ transforms the Bloch ball 
into the ball or into an ellipsoid inside the ball.
The channel $\Phi$ transforms the 
Bloch vector $\vec{r}$ representing the state $\rho$ into $\vec{r}\,'$ which corresponds to $\rho'$. 
This transformation
 is described by
\begin{equation}
\vec{r}\,'=W\vec{r}+\vec{\kappa}.
\end{equation}
Here the matrix $W$ is a square real matrix of size $3$.
A procedure analogous to the singular value decomposition 
of the matrix $W$ gives $W=O_1DO_2$,
where $O_i$ represents an orthogonal rotation and $D$ is diagonal.
Up to two orthogonal rotations, one before the transformation $\Phi$ and
one after it, the one--qubit map $\Phi$ can be represented by the following matrix
\begin{equation}
\Phi=\bordermatrix{        
 &{}\cr
& 1  & 0 & 0 & 0\cr 
& \kappa_{1}  & \eta_{1} & 0 & 0\cr
& \kappa_{2}  & 0 & \eta_{2} & 0\cr
& \kappa_{3} & 0 & 0 & \eta_{3}\cr}.
\label{bloch}
\end{equation}
The absolute values of the parameters $\eta_{i}$ are 
interpreted as the lengths of the axes of the ellipsoid which is the image 
of the Bloch ball transformed by the map. 
The parameters $\kappa_{i}$ form the vector $\vec{\kappa}$ of translation 
of the center of the ellipsoid with respect to the center of the Bloch ball. 

Due to complete positivity of the map $\Phi$ and the trace preserving property,
the vectors $\vec{\eta}$ and $\vec{\kappa}$ are subjected to several constraints.
They can be derived from the positivity condition of a dynamical matrix 
given by \cite{fujivara,KZ}:
\begin{equation}
D_{\Phi}=\frac{1}{2}\bordermatrix{         &{}\cr
 & 1+\eta_{3}+\kappa_{3}  & 0 &\kappa_{1}+i*\kappa_{2} & \eta_{1}+\eta_{2}\cr 
& 0  & 1-\eta_{3}+\kappa_{3} & \eta_{1}-\eta_{2} & \kappa_{1}+i*\kappa_{2}\cr
& \kappa_{1}-i*\kappa_{2}  & \eta_{1}-\eta_{2} & 1-\eta_{3}-\kappa_{3} & 0\cr
& \eta_{1}+\eta_{2} & t_{1}-i*\kappa_{2} & 0 & 1+\eta_{3}-\kappa_{3}\cr}.
\label{dyn}          
\end{equation}    

The channels which preserve the maximally mixed state
are called {\it bistochastic} channels.
The structure of one--qubit bistochastic channels is
discussed in more detail in Section \ref{eryk}.

\subsection{Correlation matrices}\label{mrowka}

%
%
%
%

A general measurement process is described in quantum mechanics 
  by operators  forming a 
{\it positive operator valued measure} (POVM).
Products of matrices $K^{i\dagger}K^i$ representing the POVM 
are positive and determine the identity resolution, $\idty=\sum_{i=1}^k K^{i\dagger}K^i$.
During the measurement of a quantum state $\rho$
the output $\rho_i=\frac{K^i\rho K^{i\dagger}}{\tr K^i\rho K^{i\dagger}}$
occurs with probabilities $p_i=\tr K^i\rho K^{i\dagger}$.
The identity resolution guarantees that $\sum_{i=1}^kp_1=1$.

The outcomes of a quantum measurement
are not perfectly distinguishable, unless 
different POVM operators project on orthogonal subspaces,
$K^{i\dagger}K^iK^{j\dagger}K^j=\delta_{ij}K^{i\dagger}K^i$. 
Probability distribution of the outcome states
does not contain any information on indistinguishability of outcomes.
Therefore, a better characterization of the measurement process
is given by the following {\it correlation matrix}
 $\sigma$  with entries
\begin{equation}
\sigma_{ij}=\tr K^i\rho K^{j\dagger}, \qquad i,j=1,...,k.
\label{mackor}
\end{equation}
Its diagonal contains the probabilities of
measurement outputs, while the off--diagonal
entries are related to probabilities
that the state $i$
has been determined by the measurement as the state $j$.
The correlation matrix depends on both,
the measured state and the 
measurement process.

The operators $K^i$, satisfying $\sum_{i=1}^k K^{i\dagger}K^i=\idty$,
can also be treated as Kraus operators (\ref{quantop})
characterizing the
quantum channel, $\Phi(\rho)=\sum_{i=1}^k K^i\rho K^{i\dagger}$.
In such an interpretation of operators $K^i$,
the correlation matrix (\ref{mackor}) is equivalent to
the state of environment given by  the
output of the complementary channel $\tilde{\Phi}(\rho)$ specified in Eq. (\ref{sigmal}).

The entropy of the state $\sigma$ produced by a 
complementary channel $\tilde{\Phi}$ is called the {\it exchange entropy},
since, if the initial states of the system and the environment are pure, then $S(\sigma)$ 
is equal to the entropy which is gained by both the state and the 
environment \cite{szuma}. If the initial state is maximally mixed,
$\rho=\rho_*=\frac{1}{N}\idty$, where $N$ is the dimensionality of $\rho$,
 the entropy of 
the output of the complementary channel is equal to the {\it map entropy} $S^{\map}(\Phi)$ \cite{roga0} (see also discussion in Section \ref{consek}),
\begin{equation}
S^{\map}(\Phi)=-\frac{1}{N}D_{\Phi}\log\Big(\frac{1}{N}D_{\Phi}\Big),
\label{pusia2000}
\end{equation}
where the dynamical matrix $D_{\Phi}$ is given by Eq. (\ref{miopio}).
This entropy is equal to zero if $\Phi$ represents any unitary transformation.
It attains the largest value $\log{2N}$ for completely depolarizing channel which 
transform any state into the maximally mixed state. Therefore the map entropy can 
characterize the decoherence caused by the channel.


Due to the polar decomposition of an arbitrary non normal operator
$X=HU$, we can write $K^i\rho^{1/2}=h_iU_i$, where $h_i$ 
is a Hermitian matrix and $U_i$ is  unitary. 
One can observe that 
$h_i^2=K^i\rho K^{i\dagger}=p_i\rho_i$. 
Therefore the entries of the correlation matrix
 (\ref{mackor}) can be written as:
\begin{equation}
\sigma_{ij}= \tr K^i\rho K^{j\dagger}=p_i^{\frac{1}{2}}p_j^{\frac{1}{2}}\tr\rho_i^{\frac{1}{2}}U_iU_j^{\dagger}\rho_j^{\frac{1}{2}}.
\label{krowiatko}
\end{equation}

 
%

As noticed above, the correlation matrix characterizing the
quantum measurement can be equivalently treated as 
the state of an environment after evolution given by 
a quantum channel.
The following section indicates a third possible interpretation
of the correlation matrix $\sigma$.
It can be formally treated as 
a Gram matrix of purifications of mixed states $\rho_i$. 

{\it Purification} of a given state $\rho_i\in\c M_N$  is given by a pure state 
$|\Psi_i\>$ (see Eq. (\ref{paproc})),
\begin{equation}
 {\rm Tr}_1\left|\Psi_i\right\rangle\left\langle\Psi_i\right|=\rho_i.
\end{equation}
The purification $|\Psi_i\>$ of given state $\rho_i$ 
can be written explicitly, 
\begin{equation}
\left|\Psi_i\right\rangle=\sum_{r=1}^N\Big(U_i\otimes\sqrt{\rho_i}\Big)\left|r\right\rangle\otimes\left|\phi^i_r\right\rangle,
\label{def:fid:purif}
\end{equation}
where $\{\left|\phi^i_r\right\rangle\}_{r=1}^N$ are eigenvectors of $\rho_i$. 
Notice that a purification of a given state $\rho_i$ is not unique. 
The degree of freedom is introduced by the unitary transformation $U_i$.
Moreover, any purification of given state $\rho_i$ 
can be given by such a form. 
Since eigenvectors of $\rho_i$ denoted by $\left|\phi^i_r\right\rangle$ 
form an orthonormal basis in the Hilbert space, a unitary transformation $V_i$ 
can transform it into the canonical basis $\{\left|r\right\rangle\}_{i=1}^N$. 
The purification (\ref{def:fid:purif}) can be described as
\begin{equation}
\left|\Psi_i\right\rangle=\sum_{r=1}^N\Big(U_i\otimes\sqrt{\rho_i}V_i\Big)\left|r\right\rangle\otimes\left|r\right\rangle.
\label{eq:fid:purif}
\end{equation}
The overlap between two purifications of states  $\rho_i$ and $\rho_j$ emerging from a POVM measurement  is given by
\begin{equation}
\abs{\left\langle \Psi_j\right|    \left.\Psi_i\right\rangle}^2=\abs{\left\langle m\right|    (U_j^{\dagger}U_{i}\otimes V_j^{\dagger}\sqrt{\rho_j}\sqrt{\rho_i}V_i)\left|m\right\rangle}^2,
\label{eq:fid:overlap}
\end{equation}
where $|m\>=\sum_r|r\>\otimes|r\>$.
For any two operators
 $A$ and $B$ the following relation holds, $\left\langle m\right|A\otimes B\left| m\right\rangle=\tr A^{\dagger}B$ \cite{uhlmann76}. 
Hence the overlap (\ref{eq:fid:overlap}) reads
\begin{equation}
\abs{\left\langle \Psi_j\right|    \left.\Psi_i\right\rangle}^2=\abs{\tr W\sqrt{\rho_j}\sqrt{\rho_i}}^2,
\label{eq:fid:overlap2}
\end{equation}
where the unitary matrix $W=V_iU_i^{\dagger}U_{j}V_j^{\dagger}$. 
Therefore the
matrix elements of $\sigma$ (\ref{krowiatko}) 
are equal to the scalar product of purifications  
of respective mixed states $\rho_i$ and $\rho_j$ as follows $\sigma_{ij}=\sqrt{p_ip_j}\<\Psi_j|\Psi_i\>$.


\subsubsection{Gram matrices and correlation matrices}\label{koza}

In previous chapter it was shown that the correlation
matrix can by defined by the set of purifications of states
emerging from the quantum measurement. 
Therefore, the correlation matrix can be
identified with the normalized Gram matrix of the purifications.

The Gram matrix is an useful tool in many fields.
It can receive a geometrical interpretation,
as it consists of the overlaps of normalized vectors.
If vectors are real the determinant of their Gram
matrix defines the volume of the parallelogram 
spanned by the vectors \cite{kokk,bahr}.
The Gram matrix of the evolving pure state is analyzed
in \cite{AlickiFannes}. The spectrum of this matrix can 
determine whether the evolution is regular or chaotic. 

The Gram matrix $\sigma$, 
\begin{equation}
\sigma_{ij}=\sqrt{p_ip_j}\<\psi_i|\psi_j\>
\label{mysz}
\end{equation}
has the same eigenvalues as 
\begin{equation}
\rho=\sum_ip_i|\psi_i\>\<\psi_i|.
\label{maus}
\end{equation}
The proof of this fact \cite{szli} 
uses the pure state, 
\begin{equation}
 |\phi\>=\sum_i\sqrt{p_i}|\psi_i\>\otimes|e_i\>,
\label{jelonek}
\end{equation}
where states $|e_i\>$ form the set 
of orthogonal vectors.
Since the state (\ref{jelonek}) is pure,
its complementary partial traces equal to 
(\ref{mysz}) and (\ref{maus}) have the same entropy
\begin{equation}
S\left([\sqrt{p_ip_j}\<\psi_i|\psi_j\>]_{ij}\right)=S\left(\sum_ip_i|\psi_i\>\<\psi_i|\right).
\label{slon}
\end{equation}
The entropy of the Gram matrix (\ref{mysz}) can be used
in quantum information theory to describe
the ability of compression of quantum information \cite{mitch}.
The authors of \cite{mitch} describe the fact that 
it is possible to enlarge the information transmitted by means
of set of states which are pairwise less orthogonal and thus  more
indistinguishable. 
This fact encourages us to consider global properties of 
quantum ensemble which, sometimes, are not 
reduced to joint effects of each pair considered
separately.
In Chapter \ref{dusiu} some efforts will be made to
define the quantity characterizing fidelity between three
states.

\subsection{Kraus operators constructed for an ensemble of states}\label{caneleone}

The previous section concerns the ensembles $\c E=\{p_i,\rho_i\}_{i=1}^k$ formed by
the outputs of a given quantum channel and a given input state. 
In the following section it will be shown that for any ensemble
$\c E$ the suitable Kraus operators $K^i$ can be constructed and the corresponding initial state $\rho$ can be found. 

Initial state is constructed from the states of the ensemble by taking
\begin{equation}
\rho=\sum_{i=1}^k p_iU_i^{\dagger}\rho_iU_i,
\label{initial}
\end{equation}
where the unitary matrices $U_i$ are arbitrary.
The Kraus operators constructed for ensemble $\c E$ and unitaries $U_i$ are defined by
\begin{equation}
K^i=\sqrt{p_i\rho_i}U_i\frac{1}{\sqrt{\rho}}.
\end{equation}
Notice that $K^i\rho K^{i\dagger}=p_i\rho_i$ 
and the Hermitian conjugation, $K^{i\dagger}=\frac{1}{\sqrt{\rho}}U_i^{\dagger}\sqrt{p_i\rho_i}$. 
Due to the choice of $\rho$ in (\ref{initial}) the identity resolution holds,
\begin{equation}
\sum_{i=1}^{k}K^{i\dagger}K^i=\sum_{i=1}^k p_i\frac{1}{\sqrt{\rho}}U_i^{\dagger}\rho_iU_i\frac{1}{\sqrt{\rho}}=\idty.
\end{equation}

In the special case of $k=2$ states in an ensemble,
by choosing
\begin{equation}
U_2=U_1\frac{1}{\sqrt{\sqrt{\rho_1}\rho_2\sqrt{\rho_1}}}\sqrt{\rho_1}\sqrt{\rho_2},
\label{unita}
\end{equation}
one obtains $\sigma_{12}$ equal to square root fidelity between states $\rho_1$ and $\rho_2$, as follows  $\sqrt{F(\rho_1,\rho_2)}=\tr\sqrt{\sqrt{\rho_1}\rho_2\sqrt{\rho_1}}$.
 
In consequence of the above considerations one can say that
 the ensemble emerging from POVM measurement can be arbitrary and
for any ensemble $\c E$ we can construct the set of operators $K^i$ and the corresponding  initial state $\rho$.

%

\subsection{ Quantum fidelity }

An important problem in the theory of probability is how to 
distinguish between two probability distributions. 
The so called {\it fidelity} is a quantity used for this purpose.
Assume that $P=(p_1,p_2,...,p_N)$ and $Q=(q_1,q_2,...,q_N)$ 
are two probability distributions. The fidelity between $\bold{p}$ and $\bold{q}$ 
is defined as,
\begin{equation}
F(P,Q)=\left(\sum_{i=1}^N \sqrt{p_iq_i}\right)^2.
\label{def:fid:clas}
\end{equation}
This function has several properties:
\begin{itemize}
\item it is real, 
\item positive, $F(P,Q)\geq 0$, 
\item symmetric, $F(P,Q)=F(Q,P)$, 
\item smaller or equal to unity, $F(P,Q)\leq 1$. 
\item equal to one if and only if two distributions are the same,\\ 
\mbox{$\left(F(P,Q)=1\right)\Leftrightarrow(P=Q)$}. 
\end{itemize}
These properties are shared by fidelities defined for quantum states given below.

Quantum counterpart of the fidelity for the pure states 
$\left|\phi_1\right\rangle\in \c H_N$ and $\left|\phi_2\right\rangle\in \c H_N$ 
is given by the overlap
\begin{equation}
F(\left|\phi_1\right\rangle,\left|\phi_2\right\rangle)=\left|\left\langle\phi_1\left|\right.\phi_2\right\rangle\right|^2.
\label{def:fid:pure}
\end{equation}
A probability distribution can be considered as a  diagonal density matrix. 
Generalization of two formulas (\ref{def:fid:clas}) and (\ref{def:fid:pure}) 
for arbitrary mixed states $\rho_1\in \c M_N$ and $\rho_2\in\c M_N$ is given by
\begin{equation}
 F(\rho_1,\rho_2)=\Big(\tr\sqrt{\sqrt{\rho_1}\rho_2\sqrt{\rho_1}}\,\Big)^2.
\label{def:fid:mixed}
\end{equation}
To show a relation to previous definitions of fidelity consider two commuting quantum states. 
They can be given, in the same basis, as 
$\rho_1=\sum_i^N r_i\left|i\right\rangle\left\langle i\right|$, and $\rho_1=\sum_i^N s_i\left|i\right\rangle\left\langle i\right|$. 
Hence the fidelity between them reads
\begin{equation}
\Big(\tr\sqrt{\sqrt{\rho_1}\rho_2\sqrt{\rho_1}}\,\Big)^2=\left(\tr \sqrt{\sum_{i=1}^N r_is_i\left|i\right\rangle\left\langle i\right|}\,\right)^2=\left(\sum_{i=1}^N\sqrt{r_is_i}\right)^2.
\end{equation}
This gives a relation between fidelity between mixed  quantum states 
(\ref{def:fid:mixed}) and fidelity of probability distributions 
which are composed by the eigenvalues of the states (\ref{def:fid:clas}). 
Consider now pure states,  
$\left|\Psi_1\right\rangle, \left|\Psi_2\right\rangle\in\c H_{N}\otimes\c H_N$ such that 
the partial trace over the first subspace reads,
 ${\rm Tr}_1\left|\Psi_i\right\rangle\left\langle\Psi_i\right|=\rho_i$.
There exists a relation between formula (\ref{def:fid:mixed}) 
for fidelity between two mixed states and overlaps of their 
purifications.
\begin{theorem}[Uhlmann \cite{uhlmann76}]
Consider two quantum states  $\rho_1$ and $\rho_2$ 
and their purifications $\left|\Psi_1\right\rangle$ and $\left|\Psi_2\right\rangle$. 
Then 
\begin{equation}
\Big(\tr\sqrt{\sqrt{\rho_1}\rho_2\sqrt{\rho_1}}\,\Big)^2=\max_{\left|\Psi_1\right\rangle}\left|\left\langle\Psi_1\left|\right.\Psi_2\right\rangle\right|^2,
\end{equation}
where the maximization is taken over all purifications $\left|\Psi_1\right\rangle$ of the state $\rho_1$.
\end{theorem}
\begin{proof}
The proof starts from purification formula (\ref{eq:fid:purif}), 
\begin{equation}
\left|\Psi_i\right\rangle=(U_{i}\otimes\sqrt{\rho_i}V_i)\left|m\right\rangle,
\end{equation}
where $|m\>$ is an unnormalized vector, 
$\left|m\right\rangle=\sum_{i=1}^N\left|r\right\rangle\otimes\left|r\right\rangle$. 
The overlap of two purifications  (\ref{eq:fid:overlap}) is given by
\begin{equation}
\abs{\left\langle \Psi_j\right|    \left.\Psi_i\right\rangle}^2=\abs{\tr W\sqrt{\rho_j}\sqrt{\rho_i}}^2,
\end{equation}
where the unitary matrix $W=V_iU_i^{\dagger}U_{j}V_j^{\dagger}$. 
The maximization over purifications is equivalent 
to maximization over the unitary matrix $W$. 
An inequality  $\abs{\tr A\,B} \le \norm A\, \tr \abs B$ provides the required lower bound
\begin{equation}
\abs{\tr W\sqrt{\rho_j}\sqrt{\rho_i}}^2\leq \left(\tr\abs{\sqrt{\rho_j}\sqrt{\rho_i}}\right)^2.
\label{krasnal}
\end{equation}
The upper bound is attained by the unitary matrix $W^{\dagger}$ 
equal to the unitary part of the polar decomposition of $\sqrt{\rho_j}\sqrt{\rho_i}$. 
This finishes the proof.
\end{proof}

\subsubsection{Geometrical interpretation of  fidelity}

Consider two one--qubit states in the Bloch representation (\ref{bobr}),
\begin{eqnarray}
\rho_x=\frac{1}{2}(\id+{\vec x}\cdot\vec{\sigma}),\\
\rho_y=\frac{1}{2}(\id+{\vec y}\cdot\vec{\sigma}),
\end{eqnarray}
where $\vec{\sigma}$ is the vector of Pauli matrices (\ref{malpa}).
Fidelity of the pair of states  $\rho_x$ and $\rho_y$ reads
\begin{equation}
F(\rho_x,\rho_y)=\frac{1}{2}(1+{\vec x}\cdot{\vec y}+\sqrt{1-\lVert{\vec x}\rVert^2}\sqrt{1-\lVert{\vec y}\rVert^2}).
\end{equation}
If the states $\rho_x$ and $\rho_y$ are both pure then $\lVert{\vec x}\rVert=\lVert{\vec y}\rVert=1$ and the fidelity can be given by
\begin{equation}
F(\rho_x,\rho_y)=\cos^2{\frac{\alpha}{2}},
\end{equation}
where the angle $\alpha$ is formed by two Bloch vectors which represent the pure states $\rho_x$ and $\rho_y$ 
at the Bloch sphere. One can use this statement to define the angle between two states as a function of the fidelity. 
The generalization of such an angle for arbitrary two mixed states is given by
\begin{equation}
A(\rho_1,\rho_2):=\arccos \sqrt{F(\rho_1,\rho_2)}.
\end{equation}
It was proved \cite{wooters} that such an angle satisfies the axioms of a distance and leads to a metric.

\subsection{Mutual information}\label{bolek}

The goal of quantum information is to efficiently 
apply quantum resources for information processing. 
Consider the following situation.
A sender transmits the letters of the message from the set $X=\{a_1,a_2,...,a_k\}$. The letters occur with probabilities $p_i$, 
where $i=1,...,k$. 
The message is transmitted by a communication
channel, which can be noisy and can change
some of the letters.
The receiver performs a measurement and obtains outputs $Y$ 
with a possibly different probability distribution. 
According 
to the Shannon information theory \cite{Shannon}
the amount of information 
contained in the message characterized by
probability distribution $p_i$
 is given by the
 entropy $H(X)=-\sum_i p_i \log p_i$.
Entropy describes the average amount
of digits per letter necessary to transmit the message characterized
by this probability distribution in an optimal 
encoding scheme.  

The receiver knowing the letters $Y$ has
only a part of information contained in the 
original message $X$.
The information which $Y$ and $X$ have in common 
is characterized by the {\sl mutual information} $H(X:Y)$ defined by
\begin{equation}
H(X:Y) = H(X) + H(Y) - H(X,Y),
\label{byczek}
\end{equation}
where $H(X,Y)$ is the Shannon entropy of the 
joint probability distribution of the pairs of letters, one from $X$ and one from $Y$.


 The 
errors caused by a channel can be perfectly corrected if the mutual information is equal to the entropy of the initial probability distribution.  
Otherwise the mutual information is bounded by 
the entropy of an initial distribution \cite{nielsen},
\begin{equation}
H(X:Y)\leq H(X).
\label{pajak}
\end{equation}
Following properties of the mutual information hold \cite{nielsen}: 
\begin{itemize}
\item Mutual information does not change 
$H(X:Y,Z)=H(X:Y)$ if the system $Z$ is uncorrelated with $Y$.
\item Mutual information does not increase 
if any process is made on each part, $H(X:Y)\geq H(X':Y')$, 
where prime denotes the states after the transformation.
\item If part of a system is discarded the mutual information decreases \newline\mbox{$H(X:Y,Z)\geq H(X:Z)$}. 
\end{itemize}

Mutual information can also be defined for quantum composite systems in terms of the
von Neumann entropy . The definition
is analogous to (\ref{byczek}):
\begin{equation}
S(\rho_P:\rho_Q)=S(\rho_P)+S(\rho_Q)-S(\rho_{PQ}),
\end{equation}
where states of subsystems are given by partial traces, for example, $\rho_P={\rm Tr}_Q\rho_{PQ}$.
Mutual information $S(\rho_P:\rho_Q)$ for quantum states
satisfies properties analogous to these listed above for the classical mutual information $H(X,Y)$.

\subsection{Holevo quantity}

{\sl Holevo $\chi$ quantity} (Holevo information) of the ensemble 
$\c E=\{q_i,\rho_i\}_{i=1}^k$ is defined by the formula
\begin{equation}
\chi( \{q_i,\rho_i\})\equiv S\left( \sum_{i=1}^k q_i \rho_i \right) - \sum_{i=1}^k q_i S(\rho_i).
\label{chichot}
\end{equation}
It  plays an important role
in quantum information theory. 
As the bound on the mutual information \cite{holevo},
Holevo quantity is related to fundamental restriction
on the information achievable from
measurement allowed by quantum mechanics.
It directly reflexes these features of quantum mechanics
which distinguishes this theory from classical physics.
In classical information theory 
the mutual information between the sender 
and the receiver is bounded only by the Shannon 
entropy of the probability distribution describing the original message.
In the case of an ideal channel between two parts
the mutual information is  equal to 
the upper bound.
In quantum case, even without any noise 
present during the transmission process, 
the mutual information is restricted
by the Holevo quantity which is smaller than
the entropy associated with the original message, unless
the states used to encode the message are
orthogonal. 



The theorem of Holevo \cite{holevo} is presented
below together with its proof.

\begin{theorem}[Holevo]
Let $\{\rho_i\}_{i=1}^k$ be a set of quantum states produced 
with probabilities $p_i$ from the distribution $P$. 
Outcomes of a POVM measurement performed on 
these states are encoded into symbols with probabilities $q_j$ from probability distribution $Q$.  
Whichever measurement is 
done, the accessible mutual information is bounded from above,
\begin{equation}
H(P:Q) \le S\left( \sum_{i=1}^k p_i \rho_i \right) - \sum_{i=1}^k p_i S(\rho_i).
\label{ineq:hol}
\end{equation}
\end{theorem}


\begin{proof}
Consider a three partite state, where its parts 
are denoted by the letters $P, Q$ and $M$
\begin{equation}
\omega_{PQM}=\sum_ip_i\left|i\right\rangle\left\langle i\right|\otimes\rho_i\otimes\left|0\right\rangle\left\langle 0\right|.
\label{state:holera}
\end{equation}
Three parts of the system $P$, $Q$ and $M$ 
can be associated with the preparation state, 
quantum systems, and the measurement apparatus respectively.
The state $\omega_{PQM}$ describes the 
quantum system before the measurement, since
the state of the apparatus  is independent on 
the quantum states.

Assume that the state $\omega_{PQM}$ is 
subjected to the quantum operation acting on the subsystem $QM$ as follows,
$\Phi(\rho\otimes\left|0\right\rangle\left\langle 0\right|)=\sum_j K^j\rho K^{j\dagger}\otimes\left|j\right\rangle\left\langle j\right|$. 
The Kraus operators of this quantum operation 
form a POVM measurement since $\sum_jK^{j\dagger}K^j=\idty$.
The state after this measurement is given by
\begin{equation}
\omega_{P'Q'M'}=\sum_{ij}p_i\left|i\right\rangle\left\langle i\right|\otimes K^j\rho_iK^{j\dagger}\otimes\left|j\right\rangle\left\langle j\right|.
\label{zucek}
\end{equation}

Properties of the mutual information 
listed in section \ref{bolek} imply
the key inequality of the proof:
\begin{equation}
S(\omega_P:\omega_Q)\geq S(\omega_{P'}:\omega{M'}).
\label{ineq:hol:key}
\end{equation}
To prove inequality (\ref{ineq:hol}) it is 
enough to calculate the quantities occurring in (\ref{ineq:hol:key}) 
for the state (\ref{state:holera}) and (\ref{zucek}) respectively. 
Since $\omega_{PQ}={\rm Tr}_M\omega_{PQM}=\sum_ip_i\left|i\right\rangle\left\langle i \right|\otimes\rho_i$,
the left hand side of (\ref{ineq:hol:key}) is given by
\begin{equation}
S(\omega_P:\omega_Q)=S(\omega_P)+S(\omega_Q)-S(\omega_{PQ})=S(\rho')-\sum_{i=1}^k p_iS(\rho_i),
\end{equation}
where $\rho'=\sum_ip_i\rho_i$.
This is the Holevo quantity which does not depend on the measurement operators $K^i$. 
To compute the right hand side of (\ref{ineq:hol:key}), $S(\omega_{P'}:\omega_{M'})$, consider a state (\ref{zucek}).
The observation that $p(x,y)=p_xp(y|x)=p_x\tr{K^{y\dagger}K^y\rho_x}$ leads to
\begin{equation}
S(\omega_{P'}:\omega_{M'})=H(P:Q),
\end{equation}
where $Q=\{q_y\}_y$ and $q_y=\tr K^y\rho'K^{y\dagger}$. 
This is the mutual information between the probability distributions describing the outcomes of the measurement and the original message.
That finishes the proof of the Holevo bound on 
the mutual information of message encoded into quantum systems.
\end{proof}

Above theorem is one of the most
important applications of the Holevo quantity.
Quantum information theory 
uses also the Holevo quantity $\chi$ 
to define channel capacity. 
There exist several definitions of quantum capacity of a channel depending on whether the entanglement between the
input states is allowed or not. In the case that quantum 
states in a message are not entangled 
the {\it Holevo capacity} of channel $\Phi$ is defined by
\begin{equation}
C_H(\Phi)=\max_{\c E=\{p_i,\rho_i\}_{i=1}^{k}} \left[S\left(\sum_{i=1}^kp_i\Phi(\rho_i)\right)-\sum_{i=1}^kp_iS\left(\Phi(\rho_i)\right)\right].
\label{kunka}
\end{equation}
The Holevo  quantity $\chi(\c E)$, which can be interpreted as the Holevo capacity of the
identity channel, bounds the capacity $C_H$ for any channel \cite{nielsen}: 
\begin{equation}
C_H\leq \chi(\c E).
\end{equation}

Yet another application of the Holevo quantity 
concerns the ensembles of quantum states.
Formula (\ref{chichot}) can be given by
the average relative entropy 
\begin{equation}
\sum_{i=1}^k p_iD\left(\rho_i,\sum_{j=1}^k p_j\rho_j\right)=S\left(\sum_{i=1}^k p_i\rho_i\right)-\sum_{i=1}^k p_iS(\rho_i),
\label{picolo}
\end{equation}
where the relative entropy is defined as $D(\rho_1,\rho_2)\equiv \tr\rho_1(\log\rho_1-\log\rho_2)$.
It defines an average divergence of every state 
from the average state.
Average (\ref{picolo}) is known as the quantum Jensen Shannon divergence \cite{topsoe}.
Its classical version, 
for probability measures, is considered in \cite{topsss}. 
From mathematical point of view, the
Holevo quantity can be treated as a
quantity which characterizes the concavity
of the entropy function.
 
The Holevo information will be the main object 
 considered in Part \ref{parttwo} of this thesis.

\part{Bounds on the Holevo quantity}\label{parttwo}

\section{Holevo quantity and the correlation matrix}\label{sec:prl}

In the following chapters several 
inequalities for the Holevo information (Holevo quantity) will be given.
It is well-known~\cite{nielsen} that the Shannon 
entropy of the probability vector $P=\{p_1,...,p_k\}$ 
is an upper bound for the Holevo quantity of an ensemble $\c E=\{p_i,\rho_i\}_{i=1}^k$:
\begin{equation*}
\chi\bigl( \c E \bigr) \ \le \  H(P).
\end{equation*}
Since the Holevo quantity forms a bound on
accessible mutual information,
the difference between entropy of probability vector $H(P)$
and the Holevo quantity specifies how the chosen set 
of density matrices 
differs from 
the ideal code, which can be decoded 
perfectly by the receiver. The upper bound on the
Holevo quantity can be used for estimating this
difference. One of the estimation for the 
Holevo quantity is presented in the following section.


As discussed in Section \ref{mrowka} the  
correlation matrix $\sigma$ can be equivalently interpreted in several  ways.
If the set of  the Kraus operators $K^i$ defines a quantum channel,
$\Phi(\rho)=\sum_{i=1}^k K^i\rho K^{i\dagger}$,
the correlation matrix $\sigma$ characterizes the output state
of the complementary channel, $\sigma=\tilde{\Phi}(\rho)$, 
or the state of the environment after the quantum operation.
As mentioned in Section \ref{koza}, $\sigma$
defines also  the Gram matrix of 
purifications of the states $\{\rho_i\}_{i=1}^k$. 
The entropy $S(\sigma)$ is related to the
exchange entropy or the entropy which the
environment gains during a quantum operation
provided the initial state of the environment is pure.
In the following analysis a quantum channel $\Phi(\rho)=\sum_iK^i\rho K^{i\dagger}$
is treated as a device preparing an ensemble of quantum states $\c E=\{p_i,\rho_i\}_{i=1}^k$, where
\begin{equation}
 p_i=\tr K^i\rho K^{i\dagger}, \qquad {\rm and}\qquad \rho_i=\frac{K^i\rho K^{i\dagger}}{\tr K^i\rho K^{i\dagger}}.
\end{equation}
The described situation is illustrated in Fig. \ref{fig:bara}.
\begin{figure}[ht]
\centering
\scalebox{.8}{\includegraphics{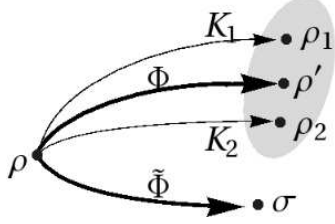}}
\caption{
A quantum channel $\Phi$ represents a device preparing the 
ensemble of quantum states $\c E=\{p_i,\rho_i\}_{i=1}^2$. 
The average of this ensemble is denoted as $\rho'=\Phi(\rho)=\sum_{i=1}^2K^i\rho K^{i\dagger}$.
The complementary channel $\tilde{\Phi}$ transforms an initial state $\rho$ into 
the state  $\sigma$ of the environment.
}
\label{fig:bara}
\end{figure}

Independently of the interpretation of the Kraus operators $K^i$ the
following theorem proved in \cite{roga} holds.

\begin{theorem}
\label{prop1}
Let $\sum_{i=1}^k K^{i\dagger}K^i=\idty$ be the identity decomposition and  
$\rho$ an arbitrary quantum state.
Define the probability 
distribution  $p_i=\tr K^i\rho K^{i\dagger}$ 
and a set of density matrices 
$\rho_i=\frac{K^i\rho K^{i\dagger}}{\tr K^i\rho K^{i\dagger}}$.
The Holevo quantity $\chi(\{\rho_i,p_i\})$ 
is bounded by the entropy of the correlation matrix, 
$\sigma=\sum_{i,j=1}^k \tr K^i\rho K^{j\dagger}|i\>\<j|$:
\begin{equation}
\chi(\{\rho_i,p_i\}) =S\big(\sum_{i=1}^k p_i\rho_i\big)-\sum_{i=1}^k p_i S(\rho_i)\leq S(\sigma)\leq H(P),
\label{propeq1}
\end{equation}
where $H(P)$ is the Shannon entropy of the probability distribution
$P=\{p_1,...,p_k\}$.
\end{theorem}

\begin{proof}
The right hand side of the inequality: 
$S(\sigma)\leq H(P)$, is a consequence of 
the majorization theorem, see e.g. \cite{KZ}. 
Since the probability vector $P$
forms a diagonal of a correlation matrix, we have  
$S(\sigma)\leq S({\rm diag}(\sigma))=H(P)$.
The left hand side of the inequality (\ref{propeq1})
is proved due to the strong subadditivity of the von Neumann entropy \cite{liebruskai}.
The multipartite state $\omega_{123}$ 
is constructed in such a way that entropies of its partial traces 
are related to specific terms of (\ref{propeq1}). 

The multipartite state $\omega_{123}$ is constructed by using
an isometry $F|\phi\rangle=\sum_{i=1}^k |i\rangle\otimes|i\rangle\otimes K^i|\phi\rangle$.
The state $\omega_{123}=F\rho F^{\dagger}$ 
is given explicitly by the formula
\begin{equation}
 \omega_{123}=F\rho F^{\dagger}=\sum_{i,j=1}^{k}|i\rangle\langle j|\otimes|i\rangle\langle j|\otimes K^i \rho K^{j\dagger}.
\label{om123}
\end{equation}
States of the subsystems $\omega_{i}$ are given by partial traces
over the remaining subsystems, for example,
 $\omega_1={\rm Tr}_{23}\omega_{123}$ and so on.
 
Let us introduce the following notation 
$A_{ij}=K^i \rho K^{j\dagger}$. 
In this notation the quantities from the Theorem 
\ref{prop1} take the form $p_i=\tr A_{ii}$ and $\rho_i=A_{ii}/p_i$. 
Notice that
\begin{align}
 &S(\omega_{12})=S(\sigma),
\label{om12}\\
&S(\omega_{3})=S\big(\sum_{i=1}^k p_i\rho_i\big).
\end{align}
Moreover
\begin{align}
-\sum_{i=1}^k p_i S(\rho_i)&=\sum_{i=1}^k \tr A_{ii}\log A_{ii} - \sum_{i=1}^k \tr(A_{ii})\log \tr(A_{ii})\nonumber\\
&=S(\omega_1)-S(\omega_{23}).
\label{average}
\end{align}
The strong subadditivity relation
in the form which is used most frequently
\begin{equation}
S(\omega_{123})+S(\omega_2)\leq S(\omega_{12})+S(\omega_{23})
\label{ssassa}
\end{equation}
does not lead to the desired form (\ref{propeq1}). However,
due to the purification procedure and the fact that
a partial trace of a pure state has the same entropy 
as the complementary partial trace, inequality 
(\ref{ssassa}) can be rewritten in an alternative form \cite{ruskaijakies}:
\begin{equation}
 S(\omega_{3})+S(\omega_1)\leq S(\omega_{12})+S(\omega_{23}).
\label{ssa2}
\end{equation}
This inequality applied to the partial traces of the state (\ref{om123})
proves \mbox{Theorem \ref{prop1}.}
\end{proof}

For an ensemble of pure states $\rho_i=|\psi_i\>\<\psi_i|$,
the left hand side of (\ref{propeq1}) 
consists of the term $S(\sum_ip_i|\psi_i\>\<\psi_i|)$ only.
The correlation matrix $\sigma$ in the case of pure
states is given by the Gram matrix.
Due to the simple observation (\ref{slon}), the left inequality (\ref{propeq1}) is saturated 
in case of any ensemble $\c E$  consisting of pure states only.

Using a different method an inequality analogous to 
Theorem \ref{prop1} has been recently proved in
\cite{holsir} for the case of infinite dimension. 
It can be also
found in \cite{private} in context 
of quantum cryptography.
The authors analyse there the security of a cryptographic key   
created by using so called 'private qubits'.
In such a setup an inequality analogous to 
 (\ref{propeq1}) appears as a bound
on the information of the eavesdropper.

\subsection{Other inequalities for the Holevo quantity}\label{los}

Methods similar to that used to prove Theorem \ref{prop1}
can be applied to prove other useful bounds.

\begin{proposition}
Consider a 
POVM measurement characterized by operators  
$\sum_{i=1}^k K^{i\dagger}K^i=\idty$ which define the
outcome states, 
 $\rho_i=\frac{K^i\rho K^{i\dagger}}{\tr K^i\rho K^{i\dagger}}$
 and their probabilities,
 $p_i=\tr K^i\rho K^{i\dagger}$.
The average entropy of the output states 
is smaller than entropy of the initial state,
\begin{equation}
 \sum_{i=1}^k p_i S(\rho_i)\leq S(\rho).
\end{equation}
\label{prop2}
\end{proposition}
\begin{proof}
Due to the fact that the transformation $F$ in Eq. (\ref{om123}) is 
an isometry, the three-partite state $\omega_{123}$ has 
 the same nonzero spectrum as the initial state $\rho$. 
Hence $\omega_{123}$ and $\rho$ have the same entropy. 
Due to  equality (\ref{average}) and the Araki--Lieb inequality \cite{arakilieb}:
\begin{equation}
 S(\omega_{1})-S(\omega_{23})\leq S(\omega_{123}),
\end{equation}
one completes the proof of Proposition \ref{prop2}.
\end{proof}
Note that  concavity of entropy implies also another inequality 
$\sum_{i=1}^k p_i S(\rho_i)\leq S(\rho')=S(\sum_{i=1}^k p_i \rho_i)$.
Proposition \ref{prop2} has been known before \cite{blady}
as the {\it quantum information gain}.

Definition of the channel capacity (\ref{kunka})
encourages one to consider bounds
on the Holevo quantity for the concatenation
of two quantum operations.
Treating the probabilities $p_i$ and
states $\rho_i$ as the outputs from the
first channel one can replace
maximization over $\c E=\{\rho_i,p_i\}_{i=1}^k$
in (\ref{kunka}) by maximization 
over the initial state $\rho$ and the quantum operation $\Phi_1$.
The strategy similar to that used in Theorem \ref{prop1}
allows us to prove the following relations.

\begin{proposition}{Consider two quantum operations:
$\Phi_1(\rho)=\sum_{i=1}^{k_1} K_1^i\rho K^{i \dagger}_1$ and $\Phi_2(\rho)=\sum_{i=1}^{k_2} K^i_2\rho K^{i \dagger}_2$. 
Define $p_i=\tr K^i_1\rho K_1^{i\dagger}$ and 
$\rho_i=\frac{K^i_1\rho K_1^{i\dagger}}{\tr K^i_1\rho K_1^{i\dagger}}$. 
The following inequality holds:
\begin{equation}
 S\big(\Phi_2\circ\Phi_1(\rho)\big)-\sum_{i=1}^{k_1} p_i S\big(\Phi_2(\rho_i)\big)\leq S(\Phi_1(\rho))-\sum_{i=1}^{k_1} p_i S(\rho_i).
\label{pierw1}
\end{equation}}
\label{prop3}
\end{proposition}

\begin{proof}
 Let us consider the four--partite state:
\begin{equation}
 \omega'_{1234}=\sum_{n,l=1}^{k_1}\sum_{i,j=1}^{k_2}|i\rangle\langle j|\otimes|nn\rangle\langle ll|\otimes K^i_2K^n_1\rho K^{l\dagger}_1 K^{j\dagger}_2,
\label{omega1234}
\end{equation}
where $|nn\>\equiv |n\>\otimes|n\>$,
and  the strong subadditivity relation in the form
\begin{equation}
 S(\omega'_{124})+S(\omega'_4)\leq S(\omega'_{14})+S(\omega'_{24}).
\label{ssa22}
\end{equation}
Notice that 
\begin{eqnarray*}
& S(\omega'_4)=S(\Phi_2\circ\Phi_1(\rho)),\\
& S(\omega'_3)-S(\omega'_{24})=-\sum_i p_i S\big(\Phi_2(\rho_i)\big),\\
& S(\omega'_{14})=S(\sum_{i,j=1}^{k_2}|i\rangle\langle j|\otimes K^i_2\Phi_{1}(\rho) K^{j\dagger}_2)=S(\Phi_{1}(\rho)).
\end{eqnarray*}
The third equality is due to the fact that an isometry, $F_2|\phi\rangle=\sum_{i=1}^{k_2} |i\rangle\otimes K^i_2|\phi\rangle$,
does not change the nonzero part of spectrum. This property
 is also used to justify the following equation
\begin{equation}
 S(\omega'_3)-S(\omega'_{124})=-\sum_{i=1}^{k_1} p_i S(\rho_i).
\end{equation}
Substituting these quantities to the strong subadditivity relation (\ref{ssa22})
we finish the proof.
\end{proof}

Inequality \ref{pierw1} is known \cite{nielsen} as
the property that the Holevo quantity 
decreases under a quantum operation
 $\chi(p_i,\rho_i)\geq\chi(p_i,\Phi(\rho_i))$.

Consider notation used in the proof of Proposition \ref{prop3}.
Concavity of the entropy gives 
\begin{equation}
\sum_{i=1}^{k_1}p_iS\Big(\Phi_2(\rho_i)\Big)=\sum_{i=1}^{k_1}p_iS\Big(\sum_{j=1}^{k_2}q_j\rho_{ij}\Big)
\geq\sum_{i=1}^{k_1}\sum_{j=1}^{k_2}p_iq_jS(\rho_{ij}).
\label{konik}
\end{equation}
where $\rho_{ij}=\frac{K^j_2K^i_1\rho K^{i\dagger}_1 K^{j\dagger}_2}{\tr K^j_2K^i_1\rho K^{i\dagger}_1 K^{j\dagger}_2}$ 
and probabilities $p_iq_j=\tr K^j_2K^i_1\rho K^{i\dagger}_1 K^{j\dagger}_2$.
Using Theorem \ref{prop1} and concavity of entropy
(\ref{konik}) one proves:

\begin{proposition}{Consider two quantum operations: 
$\Phi_1(\rho)=\sum_{i=1}^{k_1} K^i_1\rho K^{i \dagger}_1$ and $\Phi_2(\rho)=\sum_{i=1}^{k_2} K^i_2\rho K^{i \dagger}_2$. 
Define $p_i=\tr K^i_1\rho K_1^{i\dagger}$ and $\rho_i=\frac{K^i_1\rho K_1^{i\dagger}}{\tr K^i_1\rho K_1^{i\dagger}}$. 
The following inequality holds:
\begin{equation}
S\Big(\Phi_2\circ\Phi_1(\rho)\Big)-\sum_{i=1}^{k_1} p_i S\Big(\Phi_2(\rho_i)\Big)\leq S(\sigma_{II}),
\label{drug1}
\end{equation}
where the output of the complementary channel to $\Phi_2\otimes\Phi_1$ is denoted as $\sigma_{II}=\widetilde{\Phi_2\circ\Phi_1}(\rho)$.}
\label{prop4}
\end{proposition}



\subsubsection{Some consequences}\label{consek}

This section provides three applications 
of theorems proved in Sections \ref{sec:prl} and \ref{los}.
One of them concerns the {\it coherent information}.
This quantity is defined for a given quantum operation $\Phi$ 
and an initial state $\rho$ as follows  \cite{szum}
\begin{equation}
I_{coh}(\Phi,\rho)=S\big(\Phi(\rho)\big)-S\big(\tilde{\Phi}(\rho)\big),
\label{zwierzyna}
\end{equation}
where $\tilde{\Phi}(\rho)$ is the output state of the channel complementary
to $\Phi$. 
To some extent, coherent information
in quantum information theory plays a
similar role to mutual information
in classical information theory.
It is known \cite{nielsen} 
that $I_{coh}(\Phi,\rho)\leq S(\rho)$. 
That is a relation similar to (\ref{pajak}).
Moreover, it has been shown that only if $I_{coh}(\Phi,\rho)=S(\rho)$
the process $\Phi$ can be perfectly reversed.
In this case the perfect quantum error correction is possible \cite{szum}.  
The coherent information is also used to
define the {\it quantum capacity} of a quantum channel \cite{lloyd}
\begin{equation}
C_Q(\Phi)=\max_{\rho}I_{coh}(\Phi,\rho).
\label{katemidlton}
\end{equation}

The definition of the coherent information (\ref{zwierzyna})
can be formulated alternatively \cite{szum} 
by means of an extended quantum operation $\Phi\otimes\id$ acting on
a purification $|\psi\>\in\c H_2\otimes\c H_3$ of an initial state, $\rho={\rm Tr}_3|\psi\>\<\psi|$. 
This fact is justified as follows. The purification of $\rho$ determines as well
the purification $\Omega_{123}\in\c H_{1}\otimes\c H_{2}\otimes\c H_{3}$ of the state $\omega\in\c H_{1}\otimes\c H_{2}$ in (\ref{omom}),
\begin{equation}
\Omega_{123}=U_{12}\otimes\id_3\Big(|1\>\<1|_1\otimes|\psi\>\<\psi|_{23}\Big)U_{12}^{\dagger}\otimes\id_{3}.
\end{equation}
The partial trace over the environment (subspace $\c H_{1}$) reads
\begin{equation}
\Omega_{23}=[\Phi\otimes\id]\left(|\psi\>\<\psi|\right).
\end{equation}
It has the same entropy as the partial trace over the second and third subspace,
$
\Omega_{1}=\sigma,
$
which is a state of environment after evolution, 
\begin{equation}
S(\sigma)=S([\Phi\otimes\id]\left(|\psi\>\<\psi|\right)),
\label{smark}
\end{equation}
and $S(\sigma)=S(\tilde{\Phi}(\rho))$.

Coherent information (\ref{zwierzyna}) can be written as 
\begin{equation}
I_{coh}(\Phi,\rho)=S({\rm Tr}_{3}\Omega_{23})-S(\Omega_{23}). 
\label{motyl}
\end{equation}
The classical counterpart of the coherent information can be 
defined by using the Shannon entropy instead of the von Neumann entropy
and probability vectors instead of
density matrices in Eq. (\ref{motyl}).
The classical coherent information is always negative,
since the entropy of a  joint probability distribution 
cannot be smaller than its marginal distribution.

Inequalities proved in Theorem \ref{prop1} and 
Proposition \ref{prop2} together provide  the following bound on 
the coherent information,
\begin{equation}
 I_{coh}(\Phi,\rho)\leq \sum_{i=1}^k p_i S(\rho_i)\leq S(\rho),
\label{icoh}
\end{equation}
where $p_i=\tr K^i\rho K^{i\dagger}$ and $\rho_i=K^i\rho K^{i\dagger}/p_i$ are defined by  
Kraus representations  of the channel, $\Phi(\rho)=\sum_{i=1}^k K^i\rho K^{i\dagger}$.
The equality between coherent information and
the entropy of initial state $S(\rho)$ guarantees that 
$\Phi$ is reversible. 
Inequality (\ref{icoh}) implies a similar,  weaker statement: 
 only if the following equality holds $\sum_{i=1}^k p_iS(\rho_i)=S(\rho)$, 
the quantum operation $\Phi$ can be reversed.

Another consequence of inequalities proved in Section \ref{los}
concerns the so called {\it degradable channels}.
These channels are considered in quantum information theory in the context 
of their capacity \cite{ruskai}. A channel $\Phi_{deg}$ is called {\it degradable} if there 
exists a channel  $\Psi$ such that $\Psi\circ\Phi_{deg}=\tilde{\Phi}_{deg}$. 
Substituting the
degradable channel $\Phi_1=\Phi_{deg}$ and 
the  additional channel $\Phi_2=\Psi$ to inequality in Proposition \ref{prop3}  one 
 obtains a lower bound for the average entropy of $\Psi(\rho_i)$, where $\rho_i$ 
are output states from the channel $\Phi_{deg}$,
\begin{equation}
0\leq \sum_{i=1}^k p_i S(\rho_i)-I_{coh}(\Phi_{\deg},\rho)\leq \sum_{i=1}^k p_i S\big(\Psi(\rho_i)\big),
\end{equation}
where $I_{coh}(\Phi,\rho)=S\big(\Phi_{deg}(\rho)\big)-S\big(\tilde{\Phi}_{deg}(\rho)\big)$. 
The left   inequality is due to  inequality (\ref{icoh}). Therefore  
Proposition \ref{prop3} provides some characterization of the 
channel $\Psi$ which is associated with a degradable channel.

The third application of propositions from Section \ref{los}
is given as follows.
The Jamio{\l}kowski isomorphism \cite{jam} gives a representation of 
a quantum map $\Phi$ which acts on $N$ dimensional system by 
a density matrix on the extended space of size $N^2$. 
This state can be written as:
\begin{equation}
\sigma_{\Phi}=[\id\otimes\Phi]\big(\left|\phi^{+}\right\rangle\left\langle \phi^{+}\right|\big),
\label{dyn}
\end{equation}
where $\left|\phi^{+}\right\rangle=\frac{1}{\sqrt{N}}\sum_{i=1}^{N}\left|i\right\rangle\otimes\left|i\right\rangle$ 
is the maximally entangled state. A rescaled state $D_{\Phi}=N\sigma_{\Phi}$ is called the \textit{dynamical matrix}. 
In the special case, if the initial state is maximally mixed, $\rho=\frac{1}{N}\idty$, 
the entropy of the correlation matrix $\sigma$ written in (\ref{smark}) 
is equal to the entropy of the dynamical matrix. 

A quantum map $\Phi$ can by defined using its Kraus representation 
(\ref{quantop}). Since the Kraus representation is not unique \cite{KZ}, one can 
associate  many different 
correlation matrices with a given quantum operation $\Phi$ depending on both, the initial state and the set of Kraus operators. 
However the entropy of the dynamical matrix $D_{\Phi}$ is invariant under 
different decompositions. 
This entropy characterizes the quantum operation and 
is called the \textit{entropy of a map} \cite{roga}, 
denoted by $S^{\map}(\Phi)$ as defined in Eq. (\ref{pusia2000}).

Due to Theorem \ref{prop1} the entropy of a map has the following interpretation. 
It determines an upper bound 
on the Holevo quantity (\ref{chichot}) for a POVM measurement defined by the Kraus operators of $\Phi$  if
the initial state is
maximally mixed $\rho=\rho_{*}=\frac{1}{N}\idty$. 
Moreover, the entropy of a map 
is an upper bound for the Holevo quantity for POVM given by any 
set of Kraus operators $\{K^{i\dagger}K^i\}$ 
which realize the same quantum operation $\Phi$,
\begin{equation}
 \max_{\{K^i\}}\quad \chi\Big(\{p_i=\tr K^i\rho_{*} K^{i\dagger},\ \ \rho_i=\frac{K^i\rho_{*} K^{i\dagger}}{\tr K^i\rho_{*} K^{i\dagger}}\}\Big)\leq S(\Phi),
\label{entop}
\end{equation}
where $\rho'=\Phi(\rho)=\sum_{i=1}^k K^i\rho K^{i\dagger}$.

Proposition \ref{prop4} provides also an alternative lower bound for the entropy 
of composition of two quantum maps given by Theorem 3 in \cite{roga0}. 
The inequality for the entropy of composition of two maps can be now stated as
\begin{equation}
 0\leq{\rm MAX}\Big\{S(\Phi_2\circ\Phi_1(\rho_*))-\sum_{i=1}^k p_i S(\Phi_2(\rho_i)), \ S(\Phi_1)+\Delta\Big\}\leq S(\Phi_2\circ\Phi_1),
\end{equation}
where $\Delta=S\big(\Phi_2\circ\Phi_1(\rho_*)\big)-S\big(\Phi_1(\rho_*)\big)$ and $\Phi(\rho)=\sum_{i=1}^k p_i \rho_i$. 
The lower bound proved in our earlier paper  \cite{roga0} could be smaller than $0$. 
The improved bound is always greater than $0$ due to concavity of entropy.


\subsection{Discussion on the Lindblad inequality}\label{linneusz}

Lindblad \cite{lindblad} proved an inequality which relates 
the von Neumann entropy of a state $\rho$, its image $\rho'=\Phi(\rho)=\sum_{i=1}^k p_i\rho_i$ 
and the entropy of the correlation matrix $\sigma$ equal to the
output state of the complementary channel $\sigma=\tilde{\Phi}(\rho)$, 
\begin{equation}
 |S(\rho')-S(\rho)|\leq S(\sigma)\leq S(\rho')+S(\rho).
\label{lin}
\end{equation}
Another two Lindblad inequalities are obtained by permuting the states $\rho, \rho'$ and $\sigma$ in this formula.
The proof of Lindblad proceeds in a similar way to the
proof of Theorem \ref{prop1}.
It involves
a bi--partite auxiliary state 
$\omega''=\sum_{i,j=1}^k|i\rangle\langle j|\otimes K_i \rho K_j^{\dagger}$, where
the identity $S(\rho)=S(\omega'')$ is due to an isometry
similar to $F$ in (\ref{om123}). 
The Araki--Lieb inequality \cite{arakilieb}, $|S(\rho_1)-S(\rho_2)|\leq S(\rho_{12})$ 
applied to $\omega''$ proves the left hand side inequality of (\ref{lin}), 
while the 
subadditivity relation $S(\rho_{12})\leq S(\rho_1)+S(\rho_2)$
applied to $\omega''$ proves the right hand side inequality of (\ref{lin}).

Inequalities from Theorem \ref{prop1} and Proposition \ref{prop2}
\begin{eqnarray}
 S(\rho')-\sum_{i=1}^k p_i S(\rho_i)\leq S(\sigma),\label{wynik1}\\
\sum_{i=1}^k p_i S(\rho_i)\leq S(\rho)\label{wynik2}
\end{eqnarray}
use a three--partite auxiliary state 
$\omega=\sum_{i,j=1}^k|ii\rangle\langle jj|\otimes K_i \rho K_j^{\dagger}$. 
As in the case of the Lindblad inequality (\ref{lin}), 
the identity $S(\rho)=S(\omega)$ holds due to isometry.
The strong subadditivity relation applied to $\omega$ 
proves inequality (\ref{wynik1}), while 
the Araki--Lieb inequality applied for $\omega$ proves inequality (\ref{wynik2}).
Notice that an extension of the auxiliary state and application of
the strong subadditivity relation allows one to use the average entropy 
to  new inequalities for interesting quantities:
the entropy of the initial state, the entropy of the output state of 
a quantum channel $\rho'=\Phi(\rho)$ and the 
entropy of the output state of the complementary channel $\tilde{\Phi}(\rho)$.

 In the case $S(\rho')\geq S(\rho)$ (e.g. for any bistochastic operations) 
the result (\ref{wynik1}) gives a better lower constraints for $S(\sigma)$ 
than the Lindblad bound (\ref{lin}). In this case
\begin{equation}
S(\rho')-S(\rho)\leq S(\rho')-\sum_{i=1}^k p_i S(\rho_i)\leq S(\sigma),
\end{equation}
due to Prop. \ref{prop2}. However, if $S(\rho')\leq S(\rho)$ the result of 
Lindblad can be more precise depending on the values of $S(\rho)$, $S(\rho')$ 
and the average entropy $\sum_{i=1}^k p_i S(\rho_i)$. In consequence, due to Lindblad inequality 
(\ref{lin}) and the inequality (\ref{entop}) one obtains another lower bound for the entropy of a map:
{\small
\begin{equation}
 {\rm MAX} \Big\{\log(N)-S(\Phi(\rho_*)),\ \ \  \max_{\{K^i\}}\ \chi\Big(p_i=\tr K^i\rho_{*} K^{i\dagger},\ \ \rho_i=\frac{K^i\rho_{*} K^{i\dagger}}{\tr K^i\rho_{*} K^{i\dagger}}\Big)\Big\}\leq S^{\map}(\Phi),
\end{equation}}
where $\rho'=\Phi(\rho)=\sum_{i=1}^k K^i\rho K^{i\dagger}=\sum_{i=1}^k p_i\rho_i$.

\subsection{Inequalities for other entropies}\label{relat}

Inequality (\ref{propeq1}) uses the strong subadditivity relation
in the form (\ref{ssa2}) which is a specific feature of
the von Neumann entropy. Relation (\ref{ssa2})
can be equivalently formulated in terms of
relative von Neumann entropies.

The relative von Neumann entropy $D(\rho_1,\rho_2)$ 
is defined as follows
\begin{equation}
D(\rho_1,\rho_2)=\tr \rho_1\big[\log{\rho_1}-\log(\rho_2)\big]
\label{nogacz}
\end{equation}
and is finite for $\rho_2\in {\rm supp}(\rho_1)$, otherwise
it becomes infinite.

Monotonicity of relative entropy 
states that for any three--partite quantum state 
$\omega_{123}$ and its partial traces the following inequality holds:
\begin{equation}
D(\omega_{23},\omega_{2}\otimes\omega_3)\leq D(\omega_{123},\omega_{12}\otimes\omega_3).
\label{monotonicity}
\end{equation}
It is an important and 
nontrivial property of the von Neumann entropy \cite{liebruskai}, \cite{uhlma}. 
Monotonicity of the von Neumann entropy (\ref{monotonicity}) 
rewritten using the definition (\ref{nogacz}) leads to
the strong subadditivity relation:
\begin{equation}
S(\omega_{123})+S(\omega_{3})\leq S(\omega_{13})+S(\omega_{23}).
\label{gug}
\end{equation}
Complementary partial traces of any multipartite pure state
have the same entropy. This fact can be applied to
purifications of $\omega_{123}$. Therefore,
relation (\ref{gug}) is equivalent to (\ref{ssa2}) which 
can be applied to the specific three--partite state (\ref{om123})
\begin{equation}
\omega_{123}=\sum_{i,j=1}^k|i\>\<j|\otimes|i\>\<j|\otimes K_i\rho K_{j}^{\dagger}
\label{lucky}
\end{equation}
and used to prove the upper bound on the Holevo quantity in terms of a correlation matrix $\chi\leq S(\sigma)$. 
Hence, inequality (\ref{propeq1}) is 
a consequence of the monotonicity of the relative von Neumann entropy. 

Monotonicity of entropy holds also for some generalized entropies e.g. 
Tsallis entropies of order $0\leq \alpha<1$ \cite{monotsallis} 
or R\'{e}nyi entropies of order $0\leq q\leq 2$ \cite{monoRenyi}. 
Direct generalization of $\chi\leq S(\sigma)$ is not so easy, since 
the key step in the proof was the strong subadditivity form (\ref{ssa2}). 
In case of generalized entropies such a form cannot be obtained from the monotonicity of relative entropy.

The Holevo quantity can be expressed by the  
relative entropy. Consider the state (\ref{lucky}) 
and the notation: 
 $K^i\rho K^{i\dagger}=p_i\rho_i$, and $\sum_{i=1}^k p_i\rho_i=\rho'$.
The relative entropy reads:
\begin{eqnarray}
& &D(\omega_{23},\omega_{2}\otimes\omega_3)=\\
\!\!\!\!\!\!&=&\tr \omega_{23}\log{\omega_{23}}-\tr\omega_{23}\log{\omega_2}-\tr \omega_{23}\log\omega_3 \\
&=&\sum_{i=1}^k\tr p_i\rho_i\log{p_i\rho_i}-\sum_{i=1}^k p_i\log{p_i}-\tr \rho'\log{\rho'}\\
&=&\sum_{i=1}^k p_i\tr \rho_i\log{\rho_i}-\tr \rho'\log{\rho'}\\
&=&S(\rho')-\sum_{i=1}^k p_iS(\rho_i)
=\sum_{i=1}^k p_iD(\rho_i,\rho')=\chi.
\end{eqnarray}
The equality between the Holevo quantity and
relative entropy holds also for the Tsallis entropies of any order $q$
\begin{equation}
T_{\alpha}(\rho)=\frac{1}{1-\alpha}\Big[1-\tr \rho^{\alpha}\Big],
\end{equation} 
where the relative Tsallis entropy $D^T_{\alpha}$ of order ${\alpha}$ is defined as \cite{monotsallis} 
\begin{equation}
D^T_{\alpha}(\rho_1,\rho_2)=\frac{1}{{\alpha}-1}\Big[1-\tr \rho_1^{\alpha}\rho_2^{1-{\alpha}}\Big].
\end{equation}
It is now possible to compute the Tsallis--like generalized relative 
entropy $D^T_{\alpha}$ between a bipartite state $\omega_{23}$ and the 
product of its partial traces which leads to the {\it generalized Holevo quantity} $\chi^T_{\alpha}$.
If one considers the state (\ref{lucky})
\begin{eqnarray}
D^T_{\alpha}(\omega_{23},\omega_2\otimes\omega_3)
&=&\frac{1}{{\alpha}-1}\Big[1-\tr \omega_{23}^{\alpha}(\omega_2\otimes\omega_3)^{1-{\alpha}}\Big]\\
&=&\frac{1}{{\alpha}-1}\Big[1-\sum_{i=1}^k\tr (p_i\rho_i)^{\alpha}p_i^{1-{\alpha}}\rho'^{1-{\alpha}}\Big]\\
&=&\sum_{i=1}^k p_i\frac{1}{{\alpha}-1}(1-\tr\rho_i^{\alpha}\rho'^{1-{\alpha}})\\
&=&\sum_{i=1}^k p_iD^T_{\alpha}(\rho_i,\rho')\equiv\chi^T_{\alpha}.
\end{eqnarray}

In a similar way we can work with the R\'{e}nyi entropy $S_q^{R}(\rho)=\frac{1}{1-\alpha}\log[\tr\rho^{\alpha}]$. 
The corresponding relative R\'{e}nyi entropy reads \cite{stephanie}
\begin{equation}
D^R_q(\rho_1,\rho_2)=\frac{1}{q-1}\log\tr[\rho_1^{q}\rho_2^{1-q}]
\label{relak}
\end{equation} 
and the {\it R\'{e}nyi--Holevo quantity} is given by
\begin{equation}
\chi^R_{q}=\frac{1}{q-1}\log \tr (\sum_ip_i\rho_i^{q})^{1/q}.
\label{renif}
\end{equation}
Equality between the generalized R\'{e}nyi--Holevo quantity (\ref{renif}) and 
the R\'{e}nyi relative entropy (\ref{relak}) holds if relative entropy
concerns partial traces of (\ref{lucky}) and the state 
 $\rho''=(\sum_ip_i\rho_i^{q})^{1/q}$ as follows
\begin{equation}
\chi^R_{q}=D^R_{q}(\omega_{23},\omega_2\otimes\rho'').
\label{mezo}
\end{equation}
The Holevo quantity (\ref{mezo}) is smaller than $D^R_{q}(\omega_{23},\omega_2\otimes\omega_3)$ \cite{stephanie}.

The monotonicity of relative entropy 
for three considered types of generalized entropies: von Neumann entropy,
Tsallis entropy of order $0\leq \alpha<1$ and R\'{e}nyi entropy
of order $0\leq q\leq 2$ gives
\begin{eqnarray}
\chi\leq D(\omega_{123},\sigma\otimes\rho'),\label{chirelative}\\
\chi^T_{\alpha}\leq D^T_{\alpha}(\omega_{123},\sigma\otimes\rho'),\label{chirelative2}\\
\chi^R_q\leq D^R_q(\omega_{123},\sigma\otimes\rho').\label{chirelative3}
\end{eqnarray}
These relations state 
that the Holevo quantity is bounded by the relative entropy between 
the joint state of the quantum system and its environment and 
the states of these subsystems taken separately.  

In case of von Neumann entropy, inequality
(\ref{chirelative}) can be written explicitly as
\begin{equation}
\chi\leq S(\sigma)+S(\rho')-S(\rho).
\label{threestates}
\end{equation}
Notice that $\rho$ is an initial state and $S(\rho)=S(\omega_{123})$
due to isometry transformation, $F:\rho\rightarrow \omega_{123}$. 
Relation (\ref{threestates}) joints entropies of the initial state, the final state,
the state of the environment and the Holevo quantity in a single formula.
Inequality (\ref{threestates}) which can be rewritten as
\begin{equation}
S(\rho)\leq S(\sigma)+\sum_{i=1}^k p_iS(\rho_i)
\end{equation}
 gives a finer bound  than that provided by the Lindblad inequality: $S(\rho)\leq S(\sigma)+S(\rho')$. 
Inequality (\ref{threestates}) can be written as $\chi\leq S(\sigma)+Y$, 
where $|Y|=|S(\rho')-S(\rho)|\leq S(\sigma)$, due to one of the Lindblad inequalities.
In some cases this inequality confines the relation (\ref{propeq1}).

\subsection{Searching for the optimal bound}

The state $\sigma$ can be defined for a triple consisting of a
probability distribution, set of $k$ density matrices of size $N$ and a set of $k$ unitary
matrices, $\{p_i,\rho_i,U_i\}_{i=1}^{k}$.
Every triple $(p_i,\rho_i,U_i)$ defines uniquely
the pure state $|\psi_i\>$ which is the purification
of state $\rho_i$ as follows
\begin{equation}
|\psi_i\>=\sum_{r=1}^N(U_i\otimes\sqrt{\rho_i}V_i)|e_r\>\otimes|e_r\>
\label{smok}
\end{equation}
as shown in (\ref{eq:fid:purif}).
The Holevo quantity depends only 
on $\c E=\{p_i,\rho_i\}_{i=1}^{k}$.
Therefore,  Theorem \ref{prop1}
can be reformulated as follows:
\begin{theorem}\label{psyikoty}
For any ensemble $\{p_i,\rho_i,U_i\}_{i=1}^{k}$
the Holevo quantity is bounded by the entropy 
of the correlation matrix $\sigma$ minimized over all unitary matrices $U_i$
\begin{equation}
\chi(\{p_i,\rho_i\})=S(\sum_{i=1}^k p_i\rho_i)-\sum_{i=1}^k p_iS(\rho_i)\leq 
\min_{\{U_i\}}S(\sigma)=\min_{\{U_i\}} S(\sum_{i=1}^k p_i|\psi_i\>\<\psi_i|),
\label{zmaxem}
\end{equation}
where $|\psi_i\>=\sum_{r=1}^N (U_i\otimes\sqrt{\rho_i}V_i)|e_r\>\otimes|e_r\>$ and
 $\sigma_{ij}=\sqrt{p_ip_j}\tr\sqrt{\rho_i}\sqrt{\rho_j}U_j^{\dagger}U_i$.
\end{theorem}
The last equality of (\ref{zmaxem}) holds since
the correlation matrix $\sigma$ can be represented as  
the Gram matrix of purifications of $\rho_i$. 
It is known that for any Gram matrix equality (\ref{slon}) holds.

Finding minimization of $S(\sigma)$ over unitaries is
not an easy problem in general. In the following chapter
the problem will be solved for the ensemble of $k=2$ states,
and the solution is written in terms of square root of the fidelity between both states.
A conjecture that the matrix of the square roots of fidelities also
bounds the Holevo quantity for ensembles of $k=3$ states
will be formulated and some weaker bounds will be proved in the next section.

\subsubsection{Optimal bound for  two matrices}

The tightest upper bound on the Holevo quantity occurring in 
Theorem \ref{psyikoty} is obtained by taking minimum of $S(\sigma)$ over the set of unitaries. 
This is equivalent to the POVM which minimizes the correlation matrix 
among all POVM which give the same output states.
For two output states $\rho_1$ and $\rho_2$ occurring with 
probabilities $(\lambda, 1-\lambda)$ the correlation matrix
is given by
\begin{equation}
 \sigma=\begin{pmatrix}
\lambda & \sqrt{\lambda(1-\lambda)}\tr\sqrt{\rho_1}\sqrt{\rho_2}U_2^{\dagger}U_1  \\
\sqrt{\lambda(1-\lambda)}\tr\sqrt{\rho_2}\sqrt{\rho_1}U_1^{\dagger}U_2  & 1-\lambda
\end{pmatrix}.
\label{mamin}
\end{equation}
Its entropy is the lowest, if the absolute values of the off--diagonal elements are
the largest. As has been shown in Eq. (\ref{krasnal}) the expression
$\tr\sqrt{\rho_1}\sqrt{\rho_2}U_2^{\dagger}U_1$
attains its maximum over unitary matrices at the value 
\begin{equation}
\sqrt{F_{12}}=\tr{\sqrt{\sqrt{\rho_1}\rho_2\sqrt{\rho_1}}},
\end{equation}
where for brevity we use $F_{12}$ instead of $F(\rho_1,\rho_2)$.
This quantity is equal to the square root fidelity (\ref{def:fid:mixed}).
Therefore the correlation matrix of the smallest 
entropy can be rewritten in terms of the square root fidelity,
\begin{equation}
 \sigma_{min}=\begin{pmatrix}
\lambda & \sqrt{\lambda(1-\lambda)}\sqrt{F_{12}}  \\
\sqrt{\lambda(1-\lambda)}\sqrt{F_{12}} & 1-\lambda
\end{pmatrix}.
\label{sigmamin}
\end{equation}

\subsection{Jensen Shannon Divergence}

Minimal entropy of the correlation matrix characterizing 
an ensemble of two 
density matrices is related to the distance between them in the set of
density matrices. 
If the probability distribution in (\ref{sigmamin}) is uniform, $\lambda=1/2$,
the square root of the von Neumann entropy of $\sigma_{min}$ forms a metric \cite{roga4}.
It is called the {\it entropic distance} $D_E(\rho_1,\rho_2)$
\begin{equation}
 D_E(\rho_1,\rho_2)=\sqrt{S(\sigma_{min})},\qquad \sigma_{min}=\frac{1}{2}\begin{bmatrix}1&\sqrt{F(\rho_1,\rho_2)}\\ \sqrt{F(\rho_1,\rho_2)}&1\end{bmatrix}.
\label{myszy}
\end{equation}
Inequality (\ref{zmaxem}) provides the relation between this metric
and another one defined by means of the {\it Jensen--Shannon Divergence}.
The Jensen--Shannon Divergence $JSD(\{ \alpha_\nu P_\nu\})$ 
has been initially defined \cite{topsoe}, \cite{briet} as the divergence of classical probability distributions $P_\nu$ 
occurring with probabilities $\alpha_\nu$
\begin{equation}
JSD(\{ \alpha_\nu P_\nu\})= H\Big(\sum_{\nu} \alpha_\nu P_\nu\Big)-\sum_{\nu} \alpha_\nu H(P_\nu)=\sum_{\nu} \alpha_\nu H(P_\nu||\bar{P})
\end{equation}
where $H(P)$ denotes the Shannon entropy of the probability distribution $P$, 
$H(P_\nu||\bar{P})$ is the
relative entropy between $P_\nu$ and $\bar{P}$, while the average probability distribution reads
 $\bar{P}=\sum_{\nu} \alpha_\nu P_\nu$. 
\begin{figure}[ht]
\centering
\scalebox{.5}{\includegraphics{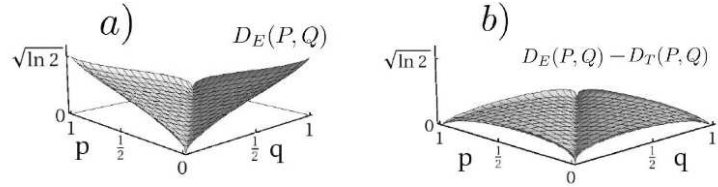}}
\caption{
$a)$ The entropic distance $D_E$ (\ref{myszy}) for two $2$--point probability distributions $P=(p,1-p)$
and $Q=(q,1-q)$. $b)$ The difference between the entropic distance $D_E$ and the transmission distance $D_T$
(\ref{spodnie}).
}
\label{fig:paraka}
\end{figure}

The square root of the Jensen-Shannon divergence 
between two probability distributions $P$ and $Q$,
\begin{equation}
 JSD(P||Q)=\frac{1}{2}H(P||M)+\frac{1}{2}H(P||M),
\end{equation}
where $M=\frac{1}{2}(P+Q)$,
forms a metric in the set of classical probability distributions \cite{briet}, \cite{endres}  
called the {\it transmission distance} $D_T(P,Q)$,
\begin{equation}
 D_T(P,Q)=\sqrt{JSD(P||Q)}.
\label{spodnie}
\end{equation}
A probability distribution can be considered as a diagonal 
density matrix. Therefore,
Eq. (\ref{zmaxem}) in Theorem \ref{psyikoty} demonstrates a relation 
between functions of two distances in the set of diagonal density matrices.
Fig. \ref{fig:paraka} and Fig. \ref{fig:parapara} shows the comparison between these two distances 
for exemplary probability distributions.

A quantum counterpart of the Jensen--Shannon divergence,
in fact coinciding with the Holevo quantity,
was also considered \cite{topsoe}, \cite{briet}. 
Inequality (\ref{zmaxem}) provides thus an upper bound
on the quantum Jensen--Shannon divergence.

\begin{figure}[ht]
\centering
\scalebox{.7}{\includegraphics{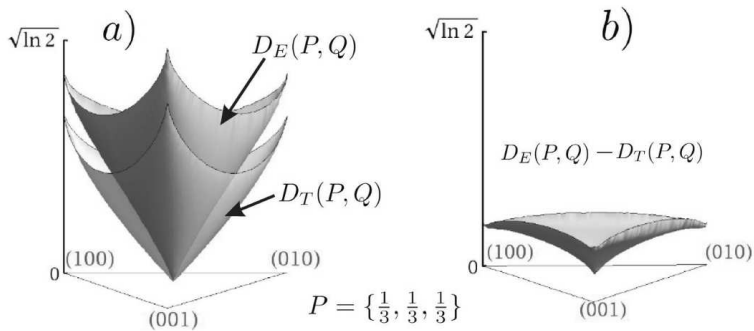}}
\caption{
$a)$ The entropic distance $D_E$ (\ref{myszy}) and the transmission distance $D_{T}$ (\ref{spodnie})
for two probability distributions, $P=(\frac{1}{3},\frac{1}{3},\frac{1}{3})$
and $Q=(q_1,q_2,q_3)$ which is arbitrary distribution of dimension $3$ represented by a point in the simplex -- the base of the figure. 
$b)$ The difference between the entropic distance $D_E$ and the transition distance $D_T$ for the same distributions  $P$ and $Q$.
}
\label{fig:parapara}
\end{figure}

\section{Conjecture on three--fidelity matrix}\label{dusiu}

The minimization problem for 
the entropy of the correlation matrix (\ref{mamin}) has been solved for
an ensemble consisting of $k=2$ quantum states.
In this case the solution is given by the square root fidelity matrix.
In the case of $k=3$ states in the ensemble the optimization over the set of three unitary 
matrices is more difficult.
Our numerical tests support the following conjecture, which is a generalization of
the bound found for the case of $k=2$. 
\begin{conjecture}\label{polny}
For an ensemble of $k=3$ quantum states, 
$\{p_i,\rho_i\}_{i=1}^3$ the entropy of the
square root fidelity matrix $G_{ij}=\sqrt{p_ip_j}\sqrt{F(\rho_i,\rho_j)}$ gives the upper bound 
on the Holevo quantity,
\begin{equation}
\chi(\{p_i,\rho_i\})\leq S\left(\begin{bmatrix}
 p_1 & \sqrt{p_1p_2}\sqrt{F_{12}} & \sqrt{p_1p_3}\sqrt{F_{13}} \\
 \sqrt{p_2p_1}\sqrt{F_{21}} & p_2 & \sqrt{p_2p_3}\sqrt{F_{23}} \\
 \sqrt{p_3p_1}\sqrt{F_{31}} & \sqrt{p_3p_2}\sqrt{F_{32}} & p_3 
\end{bmatrix}\right),
\label{zloto}
\end{equation}
where fidelity between two quantum states reads {\small $F_{ij}=F(\rho_i,\rho_j)=\big({\rm Tr}\sqrt{\sqrt{\rho_i}\rho_j\sqrt{\rho_i}}\big)^2$}.
\end{conjecture}
It has been shown \cite{fanesijegostudent}, \cite{roga5} that the matrix $G$ containing
square root fidelities is positively semi--defined for $k=3$.
However, the square root fidelity matrix is in general not positive for $k>3$. 
Numerical tests provide several counterexamples for 
positivity of $G$ for $k>3$, even in case of an  ensemble of pure states.
Note that the matrix $G$ is not a special case of the correlation matrix $\sigma$, which is positive by construction.

Theorem \ref{prop1} implies that Conjecture \ref{polny} holds for
ensembles containing three pure states.
Inequality (\ref{propeq1}) is in this case saturated as discussed 
in section \ref{sec:prl}. Square root fidelity matrix $G$ is obtained from
the Gram matrix of given pure states by taking modulus
of its matrix entries. Taking modulus of entries of a positive $3\times 3$ matrix
does not change neither the trace nor the determinant of the matrix. Only the
second symmetric polynomial of the eigenvalues is growing. 
Since the entropy is a monotonic increasing function of the second 
symmetric polynomial \cite{mitch}, the entropy of the square root fidelity 
matrix $G$ is larger than the entropy of the Gram matrix and therefore
it is also larger than the Holevo quantity.

\subsection{A strategy of searching for a proof of the conjecture}\label{skrzat}

The proof of Theorem \ref{prop1} consist of two 
steps. In the first step one has to find suitable multipartite state. 
In the second step the strong subadditivity relation of entropy 
has to be applied for the constructed multipartite state. 
 The same strategy will be 
used searching for the proof of Conjecture \ref{polny} or 
for proving other weaker inequalities.

For the purpose of obtaining the Holevo quantity 
from suitable terms of the strong subadditivity relation,
the multipartite state $\omega$ should have a few features:
\begin{itemize}
\item it is a block matrix which is positive, 
\item blocks on the diagonal should contain states $\rho_i$ multiplied by probabilities $p_i$,
\item traces of off-diagonal blocks should give square root fidelities, or some smaller 
numbers if one aims to obtain a weaker bound.
\end{itemize}
The following matrix satisfies above conditions,
\begin{equation}
X=\begin{bmatrix}
p_1\rho_1&0&0\mid&0&*&0\mid&0&0&*\\
0&0&0\mid&0&0&0\mid&0&0&0\\
\underline{0}&\underline{0}&\underline{0}\mid&\underline{0}&\underline{0}&\underline{0}\mid&\underline{0}&\underline{0}&\underline{0}\\
0&0&0\mid&0&0&0\mid&0&0&0\\
*&0&0\mid&0&p_2\rho_2&0\mid&0&0&*\\
\underline{0}&\underline{0}&\underline{0}\mid&\underline{0}&\underline{0}&\underline{0}\mid&\underline{0}&\underline{0}&\underline{0}\\
0&0&0\mid&0&0&0\mid&0&0&0\\
0&0&0\mid&0&0&0\mid&0&0&0\\
*&0&0\mid&0&*&0\mid&0&0&p_3\rho_3\\
\end{bmatrix},
\label{hiena}
\end{equation}
where in place of $*$ one can put any matrix, provided the matrix 
$X$ remains positive. If in place of $*$ one substitutes zeros, the strong subadditivity relation implies
the known formula that $\chi(\{p_i,\rho_i\})\leq S(\{p_i\})$.
Examples presented in the next section use described strategy
to prove some entropic inequalities for the Holevo quantity.

The main problem is to find a suitable positive block matrix.
In order to check positivity the Schur complement method \cite{positive}
is very useful. 
\begin{lemma}[Schur]\label{trolica}
Assume that $A$ is invertible and positive matrix, then
\begin{equation}
X=\begin{bmatrix}
 A & B \\
 B^{\dagger} & C
\end{bmatrix}
\end{equation}
is positive if and only if $\s S=C-B^{\dagger}A^{-1}B$ is positive semi--definite:
\begin{equation}
A>0\ =>\  (X>0\ <=>\ \s S\geq 0).
\end{equation}
\end{lemma}
The matrix $\s S$ is called the Schur complement.

\subsubsection{Three density matrices of an arbitrary dimension}

The strategy mentioned in the previous section 
will be used to prove the following
\begin{proposition}
\label{propo1}
 For a three states ensemble $\{p_i,\rho_i\}_{i=1,2,3}$ the following
bound for the Holevo quantity $\chi$ holds
\begin{equation}
\chi(p_i,\rho_i)\leq 
S\left(
\begin{bmatrix}
 p_1 & \sqrt{p_1p_2}\sqrt{F_{12}}/b & \sqrt{p_1p_3}\sqrt{F_{13}}/b \\
 \sqrt{p_2p_1}\sqrt{F_{21}}/b & p_2 & \sqrt{p_2p_3}\sqrt{F_{23}}/b \\
 \sqrt{p_3p_2}\sqrt{F_{31}}/b & \sqrt{p_3p_2}\sqrt{F_{32}}/b & p_3 
\end{bmatrix}
\right),
 \label{pro1}
\end{equation}
where $b\geq 2$.
\end{proposition}
\begin{proof}
It will be assumed that considered density matrices $\{\rho_i\}_{i=1}^3$ are invertible. 
After \cite{belavkin} 
the  square root of the product of two density matrices $\sqrt{\rho\sigma}$
will be defined as follows:
\begin{equation}
\sqrt{\rho\sigma}\equiv\rho^{1/2}\sqrt{\rho^{1/2}\sigma\rho^{1/2}}\rho^{-1/2}=\sigma^{-1/2}\sqrt{\sigma^{1/2}\rho\sigma^{1/2}}\sigma^{1/2}.
\label{kamyk}
\end{equation}
In this notation the fidelity between two states $\rho_i$ and $\rho_j$ 
can be written as:
\begin{equation}
F_{ij}=F(\rho_i,\rho_j)=\left(\tr\sqrt{\rho_i^{1/2}\rho_j\rho_i^{1/2}}\right)^2=(\tr\sqrt{\rho_i\rho_j})^2.
\label{pimpus}
\end{equation}
Formula (\ref{pimpus}) can be generalized for non-invertible matrices \cite{roga5}.

One can use the Schur complement Lemma \ref{trolica} 
to prove positivity of the block matrix:
\begin{equation}
X=\begin{bmatrix}
 \rho_1 & \sqrt{\rho_1\rho_2} \\
 \sqrt{\rho_2\rho_1} & \rho_2
\end{bmatrix}.
\label{mozg}
\end{equation}
In this case the matrices $A$ and $\s S$, which enter the Lemma \ref{trolica},
take the form:
 $A=\rho_1$, assume that it is invertible, and $\s S=\rho_2-\sqrt{\rho_2\rho_1}\rho_1^{-1}\sqrt{\rho_1\rho_2}$. 
Notice that
\begin{eqnarray}
\rho_2-\s S&=\sqrt{\rho_2\rho_1}\rho_1^{-1}\sqrt{\rho_1\rho_2}\rho_1\rho_1^{-1}\\ 
&=\sqrt{\rho_2\rho_1}\sqrt{\rho_2\rho_1}\rho_1^{-1}=\rho_2,
\end{eqnarray}
therefore in the case of matrix (\ref{mozg}),
 $\s S=0$ and $X>0$. Hence the following matrix $Y$ is also positive:
\begin{equation}
Y=
\begin{bmatrix}
 \frac{1}{2}\rho_1 & 0 & 0 & \frac{1}{2}\sqrt{\rho_1\rho_2} & 0 & 0 & 0 & 0 & \\
 0 & \frac{1}{2}\rho_1 & 0 & 0 & 0 & 0 & 0 & \frac{1}{2}\sqrt{\rho_1\rho_3} & 0 \\
 0 & 0 & 0 & 0 & 0 & 0 & 0 & 0 & 0 \\
 \frac{1}{2}\sqrt{\rho_2\rho_1} & 0 & 0 & \frac{1}{2}\rho_2 & 0 & 0 & 0 & 0 & 0 \\
 0 & 0 & 0 & 0 & 0 & 0 & 0 & 0 & 0 \\
 0 & 0 & 0 & 0 & 0 & \frac{1}{2}\rho_2 & 0 & 0 & \frac{1}{2}\sqrt{\rho_2\rho_3} \\
 0 & 0 & 0 & 0 & 0 & 0 & 0 & 0 & 0 \\
 0 & \frac{1}{2}\sqrt{\rho_3\rho_1} & 0 & 0 & 0 & 0 & 0 & \frac{1}{2}\rho_3 & 0 \\
 0 & 0 & 0 & 0 & 0 & \frac{1}{2}\sqrt{\rho_3\rho_2} & 0 & 0 & \frac{1}{2}\rho_3
\end{bmatrix}.
\label{blok1}
\end{equation}
Using strong subadditivity as described in section \ref{skrzat}
 to the multipartite state ${\rm Tr}_2Y$ extended by some rows and columns of zeros, one proves inequality (\ref{pro1}) for $b=2$. 
To prove relation (\ref{pro1}) for $b\ge 2$ a small modification of matrix (\ref{blok1}) is needed. 
The off--diagonal elements can be multiplied by the number
$0\leq r\leq 1$ without changing the positivity of the block matrix.
\end{proof}


\subsubsection{Three density matrices of dimension $2$}

Proposition \ref{propo1} can be amended for the case of $2\times 2$
by decreasing the parameter $b$ to the value at least $\sqrt{3}$.
\begin{proposition}\label{dwa}
 For an ensemble of three states of size two, 
$\{p_i,\rho_i\}_{i=1}^3$ one has
\begin{equation}
\chi(p_i,\rho_i)\leq 
S\left(
\begin{bmatrix}
 p_1 & \sqrt{p_1p_2}\sqrt{F_{12}}/b & \sqrt{p_1p_3}\sqrt{F_{13}}/b \\
 \sqrt{p_2p_1}\sqrt{F_{21}}/b & p_2 & \sqrt{p_2p_3}\sqrt{F_{23}}/b \\
 \sqrt{p_3p_2}\sqrt{F_{31}}/b & \sqrt{p_3p_2}\sqrt{F_{32}}/b & p_3 
\end{bmatrix}
\right)
\label{pro2}
\end{equation}
with $b\geq \sqrt{3}$.
\end{proposition}
\begin{proof}
The main task in the proof is to show that the block matrix
\begin{equation}
Y=\begin{bmatrix}
 p_1\rho_1 & \sqrt{p_1p_2}\sqrt{\rho_1\rho_2}/b & \sqrt{p_1p_3}\sqrt{\rho_1\rho_3}/b \\
 \sqrt{p_2p_1}\sqrt{\rho_2\rho_1}/b & p_2\rho_2 & \sqrt{p_2p_3}\sqrt{\rho_2\rho_3}/b \\
 \sqrt{p_3p_2}\sqrt{\rho_3\rho_1}/b & \sqrt{p_3p_2}\sqrt{\rho_3\rho_2}/b & p_3\rho_3 
\end{bmatrix}
\label{glowna}
\end{equation}
is positive for $b\geq\sqrt{3}$ as well as the analogous matrix
enlarged  by adding rows and columns  of zeros in order to have a matrix of the form (\ref{hiena}). 
The Schur complement method described in section \ref{skrzat} will be used, where:
\begin{equation}
A=\begin{bmatrix}
 \idty & 0 \\
 0 & p_1\rho_1 
\end{bmatrix}, \qquad
C=\begin{bmatrix}
 p_2\rho_2 & \sqrt{p_2p_3}\sqrt{\rho_2\rho_3}/b \\
 \sqrt{p_2p_3}\sqrt{\rho_3\rho_2}/b & p_3\rho_3 
\end{bmatrix},
\end{equation}
\begin{equation}
B=\begin{bmatrix}
 0 & 0 \\
 \sqrt{p_1p_2}\sqrt{\rho_1\rho_2}/b & \sqrt{p_1p_3}\sqrt{\rho_1\rho_3}/b 
\end{bmatrix}, \qquad
B^{\dagger}=\begin{bmatrix}
 0 & \sqrt{p_1p_2}\sqrt{\rho_2\rho_1}/b \\
 0 & \sqrt{p_1p_3}\sqrt{\rho_3\rho_1}/b 
\end{bmatrix}.
\end{equation}
Due to the fact that $A$ is positive one needs to prove the positivity of $\s S=C-B^{\dagger}A^{-1}B$:
\begin{equation}
{\footnotesize \s S=\begin{bmatrix}
 p_2\rho_2(1-\frac{1}{b^2}) & \sqrt{p_2p_3}(\sqrt{\rho_2\rho_3}/b-\sqrt{\rho_2\rho_1}\rho_1^{-1}\sqrt{\rho_1\rho_3}/b^2) \\
 \sqrt{p_2p_3}(\sqrt{\rho_3\rho_2}/b-\sqrt{\rho_3\rho_1}\rho_1^{-1}\sqrt{\rho_1\rho_2}/b^2) & p_3\rho_3(1-\frac{1}{b^2}) 
\end{bmatrix}}.
\end{equation}
To prove positivity of (\ref{glowna}) the Schur complement $\s S$ should be positive.
One can apply the Schur complement Lemma second time to the matrix $\s S$. 
Positivity condition required by Lemma \ref{trolica} enforces that
\begin{equation}
b^2(b^2-3)\rho_1+by\geq 0,
\label{dwarunek}
\end{equation}
where $y=\sqrt{\rho_1\rho_2}\rho_2^{-1}\sqrt{\rho_2\rho_3}\rho_3^{-1}\sqrt{\rho_3\rho_1}+h.c.$ 
For $2\times 2$ matrices one can assume without lost of generality 
that $\frac{1}{\sqrt{\rho_1}}y\frac{1}{\sqrt{\rho_1}}\geq 0$.
It is so because the matrix $\sqrt{\rho_1}^{-1}\sqrt{\rho_1\rho_2}\rho_2^{-1}\sqrt{\rho_2\rho_3}\rho_3^{-1}\sqrt{\rho_3\rho_1}\sqrt{\rho_1}^{-1}$ 
is a unitary matrix and its determinant is equal to $1$, therefore its eigenvalues are two conjugate numbers. The matrix
$\frac{1}{\sqrt{\rho_1}}y\frac{1}{\sqrt{\rho_1}}$, 
which consists of sum of the unitary matrix and its conjugation, is proportional to identity. If it is negative one can change 
$\sqrt{\rho_1\rho_2}$ into $-\sqrt{\rho_1\rho_2}$ and $\sqrt{\rho_2\rho_1}$ into $-\sqrt{\rho_2\rho_1}$ in (\ref{glowna}).
Transformation changing the sign does not act on the final result because off-diagonal blocks do not take part 
in forming the Holevo quantity
and in the case of $3\times 3$ matrices we can take modulus of each element of the matrix without changing its positivity.

Let us take $y=0$ in the positivity condition (\ref{dwarunek}). This condition implies $b\geq\sqrt{3}$.  
Knowing that (\ref{glowna}) is a positive matrix, the rest of the proof of (\ref{pro2}) goes like in section \ref{skrzat}.
\end{proof}

\subsubsection{Fidelity matrix for one--qubit states}

In previous section some bounds on 
the Holevo quantity were established.
These bounds are weaker than the bound
postulated by Conjecture \ref{polny}, since
decreasing the off--diagonal elements of a matrix one
 increases its entropy. In previous 
proposition the square root fidelities were divided 
by numbers greater than $1$. In the following section
the squares of the off--diagonal elements of the matrix $G$ in  (\ref{zloto})
will be taken. For such modified matrices the following
proposition holds for an arbitrary number of $k$ states in the ensemble.

\begin{proposition}\label{pro4}
Consider the ensemble $\{\rho_i,p_i\}_{i=1}^k$ of arbitrary number $k$ 
of \textbf{one-qubit} states and their probabilities. 
The Holevo information $\chi(\{p_i,\rho_i\})$ 
is bounded by the entropy of the auxiliary state 
$\varsigma$ which acts in the $k$ - dimensional Hilbert space,
\begin{equation}
\chi(\{p_i,\rho_i\})\leq S(\varsigma),
\label{ppp1}
\end{equation}
where $\varsigma_{ij}=\sqrt{p_ip_j}\,(\tr\sqrt{\rho_i\rho_j})^2=\sqrt{p_ip_j}\,F(\rho_i,\rho_j)$.
\end{proposition}

\begin{proof}
A positive block matrix $W$ is constructed in the following way:
\begin{equation}
W=
\begin{bmatrix}
 M_1&0& ... & 0 \\
 ... &...& ...& ... \\
 M_K &0& ... & 0  
\end{bmatrix}
\begin{bmatrix}
 M_1^{\dagger} & ... & M_K^{\dagger} \\
 0 & ... & 0 \\ 
... & ... & ... \\
 0 & ... & 0  
\end{bmatrix},
\label{gora}
\end{equation}
where $M_i=\sqrt{p_i}(A_i,  B_i)$ are block vectors of size $2\times 4$  
and $A_i=\rho_i$ and $B_i=\sqrt{{\rm det}\rho_i}\idty$ are sub--blocks of size $2\times 2$. 
The blocks of the block matrix $W$ read
\begin{equation}
W_{ij}=\sqrt{p_ip_j}(\rho_i\rho_j+\sqrt{{\rm det}\rho_i\rho_j}\idty).
\end{equation}
This formula can be compared with an expression for
the square root of any 
 $2\times 2$ positive matrix $X$
\begin{equation}
\sqrt{X}=\frac{(X+\sqrt{\det{X}}\idty)}{{\rm Tr}{\sqrt{X}}}.
\label{pierw}
\end{equation}
Therefore the block matrix (\ref{gora}) is given by
\begin{equation}
W_{ij}=\sqrt{p_ip_j}\sqrt{\rho_i\rho_j}\, {\rm Tr}\sqrt{\rho_i\rho_j}.
\end{equation}
The matrix $W$ is positive by construction.
Partial trace of this matrix gives matrix of  fidelities (without square root). 
The rest of the proof of Proposition \ref{pro4} goes in analogy to proofs analysed in Section \ref{skrzat}.
\end{proof}
This proposition holds for one-qubit states only 
since we applied relation (\ref{pierw}), 
which holds for matrices of dimension $d=2$. 

The fidelity matrix $\varsigma_{ij}=\sqrt{p_ip_j}\,(\tr\sqrt{\rho_i\rho_j})^2=\sqrt{p_ip_j}\,F_{ij}$
is not positive for a general $k$ and general dimensionality of $\rho_i$.
However the fidelity matrix is positive and bounds the Holevo quantity
in the case of an ensemble containing an arbitrary number
of pure quantum states of an arbitrary dimension. This is shown in the following proposition.



%

\begin{proposition}
\label{prop:fid:square}
Let $\{|\varphi_j\>\}$ be a set of vectors, then
\begin{equation}
\chi(\{p_i,\rho_i\})\leq S(\s F),
\end{equation}
where $\s F_{ij}=\sqrt{p_ip_j}\abs{\< \varphi_i | \varphi_j \>}^2$.
\end{proposition}

\begin{proof}
Introduce a complex conjugation $\varphi \mapsto \overline\varphi$
 by taking complex conjugations of all 
coordinates of the state in a given basis.  Hence for any choice of $\varphi,\ \psi$ one has
\begin{equation}
\< \varphi |\psi \> = \< \overline\psi | \overline\varphi \>.
\end{equation}
The matrix $F_2 := [F(\rho_i, \rho_j)^2]_{ij}$ can be rewritten as
\begin{equation}
\begin{split}
[\abs{\< \varphi_i | \varphi_j \>}^2]_{ij} 
&= [\< \varphi_i | \varphi_j \> \< \varphi_j | \varphi_i \>]_{ij} \\
&= [\< \varphi_i | \varphi_j \> \< \overline{\varphi_i}| \overline{\varphi_j} \>]_{ij} \\
&= [(\< \varphi_i| \otimes \<\overline{\varphi_i}|)(| \varphi_j \>\otimes |\overline{\varphi_j} \>)]_{ij}.
\end{split}
\end{equation}
This last matrix is the Gram matrix of the set of product states
 $\{ |\varphi_j\> \otimes |\overline{\varphi_j}\> \}_{j=1}^k$ and therefore is positively defined.

The next part of the proof continues according to the scheme presented in Section \ref{skrzat}. 
We use the multipartite state
\begin{equation}
\omega=\sum_{ij}\sqrt{p_ip_j}\<\varphi_i|\varphi_j\>|ii\>\<jj|\otimes|\varphi_i\>\<\varphi_j|.
\end{equation}
Its positivity is shown by taking the partial trace of the Gram matrix
\begin{equation}
\tilde{\omega}=\sum_{ij}|ii\>\<jj|\otimes|\varphi_i\>\otimes|\bar{\varphi}_i\>\<\varphi_j|\otimes\<\bar{\varphi}_j|.
\end{equation}
The proof is completed by considering partial traces  of the state $\omega$ and 
using the strong subadditivity relation.
\end{proof}





%

\subsubsection{Special case of the correlation matrix}

The previous propositions use
the strategy from the proof of 
Theorem \ref{prop1} and apply it
to positive block matrices which are
not necessary related to the correlation matrices.
Construction of multipartite states  
 allows one to 
obtain 
the matrices containing fidelities after a partial trace.
The following section deals again with the  
correlation matrices $\sigma_{ij}=\sqrt{p_ip_j}\tr\sqrt{\rho_i}\sqrt{\rho_j}U_j^{\dagger}U_i$.
Since the Holevo quantity does not depend on unitaries $U_i$,
these matrices can be chosen in such a way
that the three--diagonal of $\sigma$ consists of
the square fidelity matrices, $\sigma_{ij}=\sqrt{F(\rho_i,\rho_j)}$,
where $|i-j|\leq 1$.
This construction is used in the following proposition.

\begin{proposition}
\label{ppp2}
Consider an ensemble $\{\rho_i,p_i\}_{i=1}^k$ consisting 
of arbitrary number $k$ of invertible states of  an
\textbf{arbitrary dimension}.
The Holevo information $\chi(\{p_i,\rho_i\})$ 
is bounded by the exchange entropy $S(\sigma)$,
\begin{equation}
\chi(\{p_i,\rho_i\})\leq S(\sigma),
\end{equation}
where the correlation matrix $\sigma$ is given by:
\begin{eqnarray}
\sigma_{ii}&=&p_i,\\
\sigma_{ij}&=&\sqrt{p_ip_j}(\tr\sqrt{\rho_i\rho_j}),\quad {\rm iff}\quad |i-j|= 1,
\end{eqnarray}
and the upper off-diagonal matrix elements, where $(j-i)> 1$, read:
\begin{equation} 
\sigma_{ij}=\sqrt{p_ip_j}\,\tr\sqrt{\rho_j\rho_{j-1}}\ \frac{1}{\rho_{j-1}}\ \sqrt{\rho_{j-1}\rho_{j-2}}\ \frac{1}{\rho_{j-2}}\ \ ...\ \ \frac{1}{\rho_{i+1}}\ \sqrt{\rho_{i+1}\rho_{i}},
\label{rec2}
 \end{equation}
while lower off diagonal satisfy $\sigma_{ij}=\bar{\sigma}_{ji}$.
\end{proposition}
The matrix $\sigma$ has a layered structure presented here for $k=4$,
\begin{equation}
\sigma=
\begin{bmatrix}
p_1&0&0&0\\
0&p_2&0&0\\
0&0&p_3&0\\
0&0&0&p_4\\
\end{bmatrix}
+\begin{bmatrix}
0&f_{12}&0&0\\
f_{21}&0&f_{23}&0\\
0&f_{32}&0&f_{34}\\
0&0&f_{43}&0\\
\end{bmatrix}+
\begin{bmatrix}
0&0&f_{13}^{(2)}&f_{14}^{(3)}\\
0&0&0&f_{24}^{(2)}\\
f_{31}^{(2)}&0&0&0\\
f_{41}^{(3)}&f_{42}^{(2)}&0&0\\
\end{bmatrix}
\end{equation}  
with entries of this matrix equal to $\sigma_{ij}$ specified
in Proposition \ref{ppp2}.
\begin{proof}
Consider a correlation matrix: 
\begin{equation}
\sigma_{ij}=\sqrt{p_1p_3}\tr \sqrt{\rho_i}\sqrt{\rho_j}U^{\dagger}_jU_i
\label{correl}
\end{equation}
 where unitaries
$U_i$ are chosen in such a way that elements $\sigma_{i\ i\pm1}$ 
are square root fidelities: $\sqrt{F_{i\ i\pm1}}=\tr\sqrt{\rho_i\ \rho_{i\pm1}}$. Hence
\begin{equation}
 U^{\dagger}_j=V_{j-1,j}^{\dagger}\ U^{\dagger}_{j-1},
\label{recurrence}
\end{equation}
where $V_{j-1,j}^{\dagger}$ is the unitary matrix from the polar decomposition, 
\begin{equation}
\sqrt{\rho_i}\sqrt{\rho_j}=\abs{\sqrt{\rho_i}\sqrt{\rho_j}}V_{i,j}=\sqrt{\rho_i^{1/2}\rho_j\rho_i^{1/2}}V_{i,j}.
\end{equation}
Here the Hermitian conjugated unitary matrix $V_{i,j}^{\dagger}$ reads:
\begin{equation}
V_{i,j}^{\dagger}=\frac{1}{\sqrt{\rho_j}}\frac{1}{\sqrt{\rho_i}}\sqrt{\rho_i^{1/2}\rho_j\rho_i^{1/2}}.
\label{polar}
\end{equation}
The first unitary $U_{1}$ can be chosen arbitrarily. 
The recurrence relation (\ref{recurrence}) allows one
 to obtain formula (\ref{rec2}). 
\end{proof}
To analyse properties of the matrix $\sigma$ consider, for example, the matrix element $\sigma_{13}$.
\begin{eqnarray}
 \sigma_{13}&=&\sqrt{p_1p_3}\tr \sqrt{\rho_1}\sqrt{\rho_3}\,U^{\dagger}_3U_1\nonumber\\
&=&\sqrt{p_1p_3}\tr \sqrt{\rho_1}\sqrt{\rho_3}\,V^{\dagger}_{2,3}U^{\dagger}_{2}U_1\nonumber\\
&=&\sqrt{p_1p_3}\tr \sqrt{\rho_1}\sqrt{\rho_3}\,V^{\dagger}_{2,3}V^{\dagger}_{1,2}.
\end{eqnarray}
Using Eq. (\ref{polar}) one obtains
\begin{eqnarray}
\sigma_{13}&\!\!\!\!\!\!\!=&\!\!\!\!\!\!\!\!\!\!\sqrt{p_1p_3}\tr \sqrt{\rho_1}\sqrt{\rho_3}\frac{1}{\sqrt{\rho_3}}\frac{1}{\sqrt{\rho_2}}\sqrt{\rho_2^{1/2}\rho_3\rho_2^{1/2}}\frac{1}{\sqrt{\rho_2}}\frac{1}{\sqrt{\rho_1}}\sqrt{\rho_1^{1/2}\rho_2\rho_1^{1/2}}\nonumber\\
&=&\sqrt{p_1p_3}\tr \frac{1}{\sqrt{\rho_2}}\sqrt{\rho_2^{1/2}\rho_3\rho_2^{1/2}}\sqrt{\rho_2}\frac{1}{\rho_2}\frac{1}{\sqrt{\rho_1}}\sqrt{\rho_1^{1/2}\rho_2\rho_1^{1/2}}\sqrt{\rho_1}\nonumber\\
&=&\sqrt{p_1p_3}\tr \sqrt{\rho_3\rho_2}\frac{1}{\rho_2}\sqrt{\rho_2\rho_1},
\end{eqnarray}
that gives the matrix element $\sigma_{13}$ of (\ref{rec2}). 
The assumption that the matrices are invertible is used in (\ref{polar}) where the unitary
matrix of the polar decomposition of $\sqrt{\rho_i}\sqrt{\rho_j}$ is given explicitly.
However, the same strategy of the proof leads to analogous proposition 
involving non--invertible matrices. Only the equations (\ref{polar}) and (\ref{rec2})
are changed in this case.

\subsubsection{Hierarchy of estimations}

One can compare  average values of entropies 
from Conjecture \ref{polny} and Propositions \ref{propo1}, \ref{pro4} and \ref{ppp2}.
The average values are situated on the scale in which the 
Holevo quantity is set to $0$ and the entropy $S(P)$ of probability 
distribution is set to unity. 
The variable $\frac{x-\chi}{S(P)-\chi}$ is used,
where $x$ is replaced by the entropy of respective state.
The standard deviations are also computed.
The probability distributions
are generated according to the Dirichlet measure, while 
the set of $k=3$ density matrices 
is chosen randomly according to 
the Hilbert--Schmidt measure \cite{zyczkowskisommers} on the set of density matrices of size $2$. 
\begin{itemize}
\item $<\chi>=0$
\item $<S_{fid}>=0.176 \pm 0.065$, where $S_{fid}$ corresponds to the entropy from
 Conjecture \ref{polny}.
\item $<S_{layered}>=0.193 \pm 0.087$, where $S_{layered}$ corresponds to the entropy from
 Proposition \ref{ppp2} for $k=3$ states in the ensemble.
\item $<S_{fid^2}>=0.37 \pm 0.13$, where $S_{fid^2}$ corresponds to the entropy from
 Proposition \ref{pro4} for $k=3$ states in the ensemble.
\item $<S_{fid/b}>=0.750 \pm 0.015$, where $S_{fid/b}$ corresponds to the entropy from
 Proposition \ref{propo1}.
\item $<S(P)>=1$.
\end{itemize}

For an ensemble of $k=3$ one--qubit states Conjecture \ref{polny}
is the strongest, as it gives on average the lowest bound, while among
the statements proved in Propositions \ref{propo1}, \ref{pro4} and \ref{ppp2}
 the tightest bound (on average) is provided by \mbox{Proposition \ref{ppp2}.}






\subsection{Fidelity bound on the Holevo quantity for a special class of states}\label{sadelko}

Although, Conjecture \ref{polny} has been 
 confirmed in several numerical tests,
it has been proved so far for the set of pure states (Section \ref{dusiu}) only.
The aim of the following section is to prove that the square root fidelity 
matrix bounds the Holevo quantity for a restricted set of states.
It will be shown that for one--qubit states 
  among which
two are pure and one is mixed 
and for the uniform probability distribution, 
$\{\frac{1}{3},\frac{1}{3},\frac{1}{3}\}$,
Conjecture \ref{polny} holds.

\begin{proposition}\label{bungabunga}
Consider $k=3$ one--qubit states $\rho_i$ among which
 two  are pure $\rho_1=|\phi_1\>\<\phi_1|$, $\rho_2=|\phi_2\>\<\phi_2|$, and the state $\rho_3$ is mixed. 
The square root fidelity matrix $G$ for these states and the uniform distribution 
$P=\{ \frac{1}{3}, \frac{1}{3}, \frac{1}{3} \}$ bounds the Holevo quantity,
\begin{equation}
\chi(\{\rho_i, p_i\}) \leq S\left(
\frac{1}{3}
\begin{bmatrix}
1 & \sqrt{F_{12}} & \sqrt{F_{13}} \\
\sqrt{F_{21}} & 1 & \sqrt{F_{23}} \\
\sqrt{F_{31}} & \sqrt{F_{32}} & 1    
\end{bmatrix}
\right),
\label{ppmqubit}
\end{equation}
where $F_{ij}=(\tr\sqrt{\sqrt{\rho_i}\rho_j\sqrt{\rho_i}})^2$.
\end{proposition}

The proof goes as follows. 
First proper parameters characterizing three states will be chosen.
After that the formulas for the left and right side of 
 inequality 
(\ref{ppmqubit}), which are functions of two variables only, 
will be given. 
The fact that one of these functions is greater than the other 
is shown graphically. 

Notice that the left hand side of  
Eq. (\ref{ppmqubit}) depends only on the lengths of the Bloch vectors 
which represent the mixed state $\rho_3$ and the average state $\bar{\rho}=\frac{1}{3}(\rho_1+\rho_2+\rho_3)$
inside the 
Bloch ball.
The same average $\bar{\rho}$ can be realized by many triples $\{\rho_1,\rho_2,\rho_3\}$ 
where $\rho_1,\rho_2$ are pure and $\rho_3$ is mixed of 
given  length of the Bloch vector. 
The family of such triples is parametrized by two numbers $\alpha$ and $\beta$ as shown in Fig. \ref{fig:para}. 
The
points $B, D, E$ denote the following states: $B\rightarrow\rho_3$ which is mixed, $D\rightarrow\rho_1=|\phi_1\>\<\phi_1|$ 
and $E\rightarrow\rho_2=|\phi_2\>\<\phi_2|$, while $A\rightarrow\bar{\rho}$ represents the average state. 
The vector $\vec{OA}$ of length $a$ 
denotes the Bloch vector of the average state $\bar{\rho}$, the vector
$\vec{OB}$ of length $b$ characterizes the mixed state $\rho_3$. 
The position of the vector $\vec{OB}$ with respect to $\vec{OA}$ 
can be parametrized by an angle $\alpha$. 
These two vectors, $\vec{OB}$ and $\vec{OA}$, determine, but not uniquely, two pure states from the same triple characterized by
$\vec{OD}$ and $\vec{OE}$. 
Equivalently one can rotate the vectors $\vec{OD}$ and $\vec{OE}$ by an angle 
$\beta$ around the axis $\vec{OC}$ and obtain
pure states denoted by $F$ and $G$.
The ratio $|AB|:|AC|$ is equal to $2:1$ because in this case the average $A$ is the barycenter of three points $B$, $D$ 
and $E$ or a triple $B$, $F$ and $G$. The method of obtaining the points $C, D, E, F$ and $G$, when $a, b$ and $\alpha$ 
are given, is presented in Appendix 1.
Given a pair of parameters $(a,b)$ 
distinguishes the family of triples $\{|\phi_1\>\<\phi_1|,|\phi_2\>\<\phi_2|,\rho_3\}$
characterized by two angles $\alpha$ and $\beta$. The range of $\alpha$
 is given by condition $|OC|\leq 1$, it is
\begin{equation}
\begin{cases}
\frac{1}{2} \sqrt{9 a^2-6 b \cos (\alpha ) a+b^2}\leq 1\\
0\leq\alpha\leq \pi,
\end{cases}
\end{equation}
while the range of $\beta$ is $(0,\pi)$.
\begin{figure}[ht]
\centering
\scalebox{.6}{\includegraphics{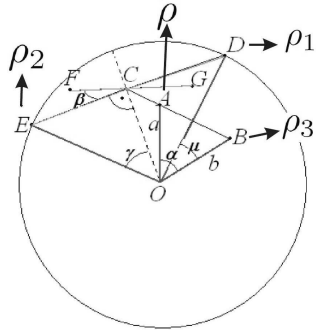}}
\caption{
The Bloch representation of the three states $\{\rho_1=|\phi_1\>\<\phi_1|,\rho_2=|\phi_2\>\<\phi_2|,\rho_3\}$ 
and the parameters used in the proof of Proposition \ref{bungabunga} are presented 
schematically in the Bloch ball. 
Two angles $(\alpha,\beta)$ characterize all possible triples $\{\rho_1,\rho_2,\rho_3\}$
if parameters $(a,b)$ are fixed.
}
\label{fig:para}
\end{figure}
Left hand side of Eq. 
(\ref{ppmqubit}) depends only on the lengths $a$ and $b$, and is independent of
 the concrete realization of the triple. 
Therefore to prove (\ref{ppmqubit}) for given $a$ and $b$ one has to find minimum 
of the entropy of the square root fidelity matrix
over all triples parametrized by the angles $\alpha$ and $\beta$. 

The entropy of the square--root fidelity matrix defined $G_{ij}=\sqrt{p_ip_j}\sqrt{F_{ij}}$
in Eq. (\ref{ppmqubit}) is a function of roots of the characteristic
polynomial:
\begin{equation}
(\frac{1}{3}-\lambda)^3+p(\frac{1}{3}-\lambda)+q=0,
\label{serdel}
\end{equation}
where
\begin{eqnarray}
p &=& -(F_{12} + F_{13} + F_{23})/9\label{ppmp}\\
q &=& 2 \sqrt{F_{12}F_{13}F_{23}}/27.\label{ppmq}
\end{eqnarray}
The parameter $p$ determines the second symmetric polynomial $s_2$ of 
eigenvalues of the square root fidelity matrix $G$
\begin{equation}
s_2=\frac{1}{9}\sum_{i<j}(1-F_{ij}).
\label{zwierz}
\end{equation}
The roots of equation (\ref{serdel}) are equal to:
\begin{eqnarray}
\lambda_k=\frac{1}{3}+2\sqrt{\frac{-p}{3}}\cos\Big[\Big(\frac{1}{3}\arccos{(\frac{3q}{2p}\sqrt{\frac{3}{-p}})}+k\frac{2\pi}{3}\Big)\Big],
\label{ppmeigenvalues}
\end{eqnarray}
where $k=1,...,3$. 

The entropy of the square root fidelity matrix is a function of $p$ and $q$,
which determine the second symmetric polynomial of eigenvalues (\ref{zwierz}) and 
the third symmetric polynomial is in this case equal to the determinant of the $3\times 3$ matrix 
$G_{ij}=\sqrt{p_ip_j}\sqrt{F_{ij}}$. 
The von Neumann entropy is a monotonically increasing function of
all symmetric polynomials of eigenvalues \cite{mitch}.
The parameter $q$ is 
a function of $(a,b,\alpha,\beta)$, while parameter $p$ depends only on $a$ and $b$ which is shown in following lemma:

\begin{lemma}\label{ppmlem1} 
For any triple of two pure and one mixed state of an arbitrary dimension 
the sum of fidelities depends only on the purity of the mixed state and
the barycenter of the ensemble.
\end{lemma}
\begin{proof}
Denote by $\bar{\rho}$ the barycenter of a mixed state $\rho$ and two pure states, $\left|\phi_1\right\rangle$, 
$\left|\phi_2\right\rangle$,
\begin{equation}
\bar{\rho}=\frac{1}{3}\rho+\frac{1}{3}\left|\phi_1\right\rangle\left\langle\phi_1\right|
+\frac{1}{3}\left|\phi_2\right\rangle\left\langle\phi_2\right|.
\end{equation}
The purity of $\bar{\rho}$ is given by
\begin{equation}
\tr{\bar{\rho}^2}=\frac{1}{9}\Big(\tr{\rho^2}+2+2\left\langle\phi_1\right|\rho\left|\phi_1\right\rangle
+2\left\langle\phi_2\right|\rho\left|\phi_2\right\rangle
+2\left|\left\langle\phi_1\right|\left.\phi_2\right\rangle\right|^2\Big).
\end{equation}
After reordering the terms one gets
\begin{equation}
F_{12}+F_{13}+F_{23}=\frac{1}{2}(9\tr{\bar{\rho}^2}-\tr{\rho^2}-2),\label{ppmsum}
\end{equation}
where
\begin{eqnarray}
F_{12}=\left|\left\langle\phi_1\right|\left.\phi_2\right\rangle\right|^2,\\
F_{23}=\left\langle\phi_2\right|\rho\left|\phi_2\right\rangle,\\
F_{13}=\left\langle\phi_1\right|\rho\left|\phi_1\right\rangle.
\end{eqnarray}
Since $\tr\bar{\rho}^2=\frac{1}{2}(1+a^2)$ and $\tr\rho^2=\frac{1}{2}(1+b^2)$, the parameter $p$ defined in (\ref{ppmp}) 
does not depend on the angles $\alpha$ and $\beta$.
This completes the proof of \mbox{Lemma \ref{ppmlem1}.}
\end{proof}

The parameter $p$ and the second symmetric polynomial (\ref{zwierz})
 does not depend on the angles $\alpha$ and $\beta$. 
Therefore, for given $a$ and $b$, the entropy of the square root fidelity matrix
 attains its
minimum over $\alpha$ and $\beta$  for minimal value of the determinant of $G$,
since the entropy is an increasing function of the determinant. The
determinant is given by
\begin{equation}
\det{\left(\frac{1}{3}
\begin{bmatrix}
1 & \sqrt{F_{12}} & \sqrt{F_{13}} \\
\sqrt{F_{21}} & 1 & \sqrt{F_{23}} \\
\sqrt{F_{31}} & \sqrt{F_{32}} & 1    
\end{bmatrix}\right)}=\frac{1}{27}\left(1+2\sqrt{F_{12}F_{13}F_{23}}-(F_{12}+F_{13}+F_{23})\right).
\end{equation}
It is the smallest for the smallest value of the parameter $q$
which is the function (\ref{ppmq}) of the off--diagonal elements of the matrix. 
During computations of the minimal value of $q$ another lemma will be useful:

\begin{lemma}\label{ppmlem2}
Among triples
 of one--qubit states which realize the same barycenter,
 where one state is mixed of a given purity and two others are pure, 
the product of three pairwise fidelities is
the smallest if three states and the average lie on the plane containing the great circle of the Bloch ball, 
i.e. $\beta=0$.
\end{lemma}
\begin{proof}
The function $f(a,b,\alpha,\beta)=F_{12}F_{13}F_{31}$
is given explicitly in
Appendix 2 based on Appendix 1. 
For given $a, b$ and $\alpha$ this function has minimum only at $\beta=0$ and equivalently for $\beta=\pi$.
\end{proof}
In consequence, searching for the minimum of 
the entropy of the square root fidelity matrix
we can restrict our attention to the case $\beta=0$. 
In fact, for our purpose it suffices to take the specific value of $\alpha$ which
is shown in the following lemma.
\begin{lemma}\label{zdzichzdzich}
Among triples of one--qubit states which realize the same barycenter,
in which one state is mixed of given purity and two others are pure, 
the product of three pairwise fidelities is
the smallest when 
two pure states are symmetric with respect to the mixed state
i.e. $\beta=0$ and $\alpha=0$ or $\alpha=\pi$.
\end{lemma}
\begin{proof}
The function $f_0(a,b,\alpha,\beta=0)=F_{12}F_{13}F_{31}$ 
is given directly in Appendix 1. 
It has only one minimum at $\alpha=0$ but in certain cases, 
depending on $a$ and $b$, the value on the edge of variable range, i.e. at $\alpha=0$ or $\alpha=\pi$ is smaller. 
\end{proof}

\subsubsection{Proof of the fidelity bound}
To prove inequality (\ref{ppmqubit}) the smallest entropy of the square root 
fidelity matrix for three states consistent with the left hand side of this inequality should be found. 
Entropy is a function of four parameters, $(a, b, \alpha, \beta)$. 
 The left hand side of (\ref{ppmqubit}), which is the Holevo quantity depends on two parameters $(a,b)$ as follows
\begin{equation}
\chi= S\left(\frac{1}{2}\begin{bmatrix}
                 1+a & 0 \\
		0 & 1-a
                \end{bmatrix}\right)-
\frac{1}{3} S\left(\frac{1}{2}\begin{bmatrix}
                 1+b & 0 \\
		0 & 1-b
                \end{bmatrix}\right).
\end{equation}
For given parameters $a$ and $b$ lemmas 1, 2 and 3 allows us to find specific $\alpha$ and $\beta$ for which
 minimization of right hand side of (\ref{ppmqubit}) is obtained. 
One can fix $\alpha=0$ or $\alpha=\pi$ and $\beta=0$.
That means, that minimal entropy of the 
square root fidelity $G$ over the angles is 
obtained if the three states $\{\rho_1=|\phi_1\>\<\phi_1|,\rho_2=|\phi_2\>\<\phi_2|,\rho_3\}$
 are lying on
the great circle and the two pure states are symmetric
with respect to the mixed state.
In this case the matrix $G$ is characterized by two parameters, $F=F_{12}=F_{23}$ and $b$. 
Here $F$ is the fidelity between the pure state $\rho_1$  and the mixed
state $\rho_3$ whereas $b$ characterize the length of the Bloch
vector of the mixed state $\rho_3$. The matrix $G$ reads
\begin{equation}
G=\frac{1}{3}
\begin{pmatrix}
1 & \sqrt{F} & |\frac{2F-1}{b}| \\
\sqrt{F} & 1 & \sqrt{F} \\
|\frac{2F-1}{b}| & \sqrt{F} & 1    
\end{pmatrix},
\label{fid21}
\end{equation}
where $F$ is a function of $b$, such that $F(b)=\frac{1}{2}(1-bc)$,
and $c$ is the length of the Bloch vector representing the barycenter of two pure states $\rho_1$ and $\rho_2$.
The fidelity $F$
 is equal to $1/2$ if $b$ tends to $0$. 
The parameter $c$ determines also the projection of the Bloch vector
of the pure state $\rho_1$ on the Bloch vector of the mixed state $\rho_3$.
The absolute value $|c|$ is equal to the square root fidelity
between the two pure states. The range of variables
are $0\leq b\leq1$ and $\frac{1}{2}(1-b)\leq F\leq \frac{1}{2}(1+b)$.

Considered case is shown in Fig. \ref{fig:zplakatu}.
There are two surfaces -- functions of two parameters
$F$ and $b$. The lower surface represents the Holevo quantity $\chi$, 
and the upper surface denotes the entropy of the square root
fidelity matrix (\ref{fid21}). The surface $S(G)$
lies always above $\chi$ and is composed of  
two smooth functions characterizing cases in which 
 all vectors lay on the same semicircle or
pure states and the mixed state belong to the opposite semicircles.

Fig. \ref{fig:zplakatu} suggests that in the case of three pure states,
$b=1$, laying on the same semicircle the 
inequality is saturated, 
$\chi=S(G)$. In this case, $F\geq 1/2$, the rank
of the square root fidelity matrix is equal to $2$,
and the nonzero eigenvalues are $(1\pm a)/2$,
where $a=(4F-1)/3$ is the length 
of the Bloch vector of the average state $\bar{\rho}$.  
In general we have $a=\frac{1}{3}(b+2\frac{2F-1}{b})$.
In case of $\chi=S(G)$ the Holevo quantity is equal to
the entropy of the average state $\bar{\rho}$. This finishes the proof of Proposition \ref{bungabunga}.
 \begin{figure}[ht]
 \centering
 \scalebox{.9}{\includegraphics{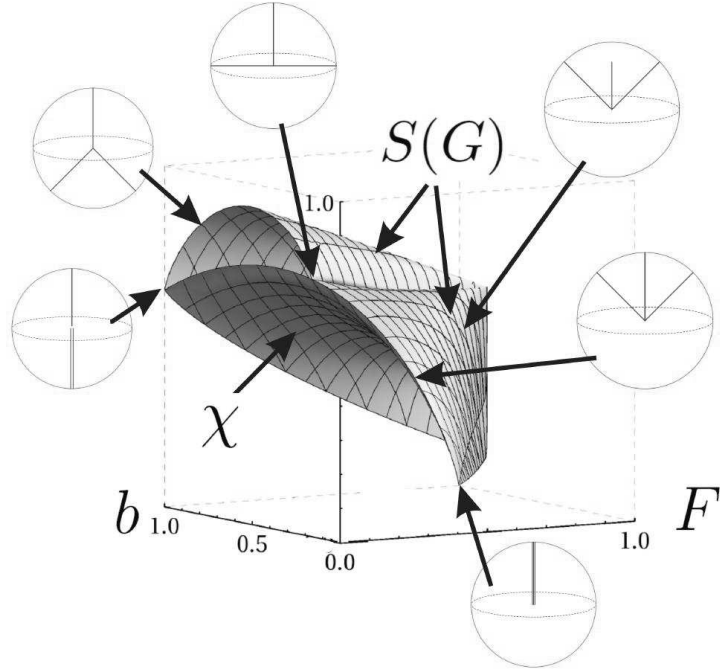}}
 \caption{ 
Evidence in favour of Proposition \ref{bungabunga}.
The Holevo quantity as function of two variables:  fidelity $F$
between the pure state $\rho_1$ and the mixed state $\rho_3$, 
and the length $b$ of the Bloch vector characterizing the state $\rho_3$.
The upper surface representing the square root fidelity matrix $G$
is composed of two smooths parts. Every circle represents schematically
the Bloch ball with exemplary positions of 
Bloch vectors characterizing
 three states $\{\rho_1=|\phi_1\>\<\phi_1|,\rho_2=|\phi_2\>\<\phi_2|,\rho_3\}$
of the ensemble.
}
 \label{fig:zplakatu}
\end{figure}

\part{Minimal output entropy and map entropy}\label{partthree}

\section{Entropies for one-qubit channels}\label{sec:ent}
The question on additivity of the channel capacity
is one of the most interesting problems
in quantum information theory \cite{amosov}.
Shor showed \cite{shor} that this problem has several 
equivalent formulations.
One of them concerns the minimal output entropy,
\begin{equation}
S^{\min}(\Phi)=\min_{\rho}S(\Phi(\rho)).
\end{equation}
In the case of one--qubit channel the
minimal output entropy is 
the entropy of a state 
characterized by point on the ellipsoid, which is the image of the Bloch sphere,
 the closest to this sphere. 
The pure state which is transformed into a state
of the minimal entropy is called minimizer. 

For any setup in which minimal 
output entropy is additive the quantum channel
capacity is additive as well. 
Additivity implies that an
entangled state cannot increase
capacity of two channels with respect
to the sum of their capacities taken separately.
The additivity conjecture can also be formulated as a
statement that capacity of two channels
is minimized for a product state.

The conjecture was confirmed in many
special cases. For instance, additivity holds, 
if one of the channels is arbitrary and the second one is:
bistochastic one--qubit map \cite{king1},
a unitary transformation \cite{amosov},
generalized depolarizing channel \cite{king2},
entanglement breaking channel \cite{shor2},
very noisy map \cite{diVincenzo}
and others. A useful review 
on this subject was written by Holevo \cite{holrev}.
Different strategies of proving the
additivity conjecture are analyzed there.
For a recent relation on the additivity conjecture see also \cite{kingostatni}.

Also counterexamples to the additivity conjecture 
have been found.
One of them was presented
by Hastings \cite{hastings}. 
He found the lower bound for
the output entropy of some channels when
the input was a product state.
Next he estimated the output entropy
for a maximally entangled input. Due to such
estimations it was shown that the entangled
state decreases channel capacity below
the value achievable for product states.

The proof of Hastings used pairs of complementary channels.
His argument was not constructive and works in
high dimensional spaces.
Counterexamples for the additivity hypothesis are also studied in \cite{horo}.
 
It is still an open question,
whether the additivity holds for an arbitrary one--qubit channel.
Originally, the hypothesis on additivity of
minimal output entropy was formulated
for the von Neumann entropy.
One of the approaches to the problem 
uses a one--parameter
family of entropies, called R\'{e}nyi entropies
characterized by a parameter $q$,
\begin{equation}
S_q(\rho):=\frac{1}{1-q}\log\tr\rho^q.
\end{equation}
Calculations are sometimes easier when the
R\'{e}nyi entropies are considered. The quantity $S_q$ tends to the
von Neumann entropy in the limit $q\rightarrow1$.
Additivity of the minimal output R\'{e}nyi entropy 
has been proved only in some range of the parameter $q$
depending on the channels considered \cite{king1,king2,kingostatni}.

Although the R\'{e}nyi entropy is sometimes 
computationally more feasible,
finding minimum over entire set
of quantum states is still a hard problem.
One of the ideas how to omit this difficulty 
tries to use some relations between
minimal output entropy and other quantities 
which are
easier to calculate. In the following chapter
the R\'{e}nyi entropy of a map (the map entropy) 
 will be used to estimate 
the minimal output entropy.
Map entropy (entropy of a map) is defined by the entropy of the
Choi-Jamio{\l}kowski state (\ref{zyrafa}) corresponding to the map.
This  quantity is easy to obtain.
Numerical tests presented
in Fig. \ref{fig:reny1}, \ref{fig:reny2}, \ref{fig:reny3} show
that there is no simple functional 
relation between the map entropy and the 
minimal output entropy. 
Nevertheless being aware of the structure
of the set of quantum maps projected on 
the plane $(S^{\map}_q,S^{\min}_q)$
can be useful.
Knowledge of entropies of  maps at the boundaries of the allowed set 
can be used to estimate
the minimal output entropy by the entropy
of the map.

\subsection{Structure of the set of Pauli channels}\label{eryk}

Quantum channels which preserve the maximally mixed
state are called bistochastic.
All bistochastic one--qubit channels can
be represented as a convex combination of the
identity matrix $\sigma_0=\idty$ and 
three Pauli matrices $\sigma_{i=1,2,3}$ (\ref{malpa})
\begin{equation}
\Phi_{\vec{p}}(\rho)=\sum_{i=0}^{3} p_i\sigma_i\rho\sigma_i, \qquad \sum_{i=0}^3p_i=1, \qquad \forall_i p_i\geq 0.
\label{pauli}
\end{equation}
Bistochastic one--qubit quantum operations are thus
called {\it Pauli channels}.
The structure of the set of all Pauli channels
 forms a regular tetrahedron $\Delta_3$ as shown in 
Fig. \ref{fig:tetra}{\it a}.
There are many channels characterized by the 
points of tetrahedron which can be obtained from
other channels following a unitary transformation.
Our considerations are often 
restricted to the asymmetric tetrahedron $K$ (see Fig. \ref{fig:tetra}{\it b})
which is a subset  of $\Delta_3$.
All maps in $\Delta_3$ can be obtained 
from channels of $K$ by concatenation these channels with unitary transformations.
The set $K$ is formed  by
the convex combination of four vectors $\vec{p}$ from (\ref{pauli}),
$A=(0,0,0,0),\  B=(1/2,1/2,0,0),\ C=(1/3,1/3,1/3,0)$, and $D=(1/4,1/4,1/4,1/4)$.
\begin{figure}[htbp]
	\centering
	\scalebox{0.9}{
		\includegraphics{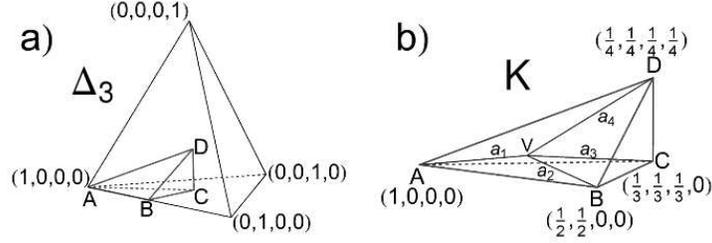}}
\caption{The structure of  one--qubit bistochastic quantum operations 
corresponds to the regular tetrahedron $\Delta_3$.
This figure is spanned by four extremal vectors $\vec{p}$
from formula (\ref{pauli}).
Symmetries of the tetrahedron allow us to distinguish the
asymmetric set $K$ inside $\Delta_3$. 
Any vector $\vec{p}$ characterizing a Pauli channel 
can be obtained by permutation of elements 
of vectors from $K$.
}
	\label{fig:tetra}
\end{figure}
\ 
\begin{figure}[htbp]
	\centering
	\scalebox{1}{
		\includegraphics{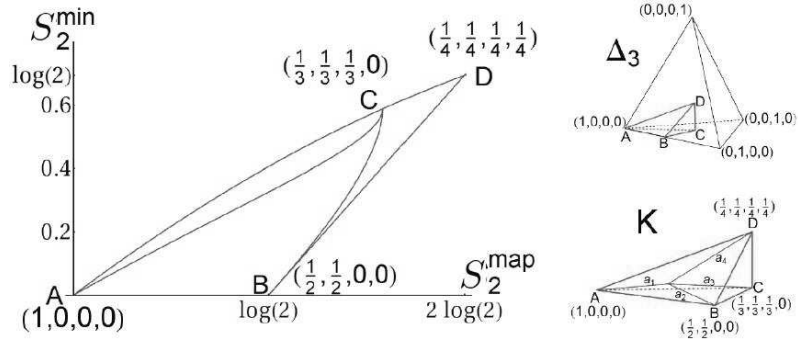}}
\caption{
Lines $AB$, $BD$ and $AD$, which correspond 
to the edges of asymmetric tetrahedron $K$
form the boundaries of the entire set
of Pauli matrices projected on the plane $(S^{\map}_2,S^{\min}_2)$.
}
	\label{fig:reny1}
\end{figure}
Extremal lines of the asymmetric tetrahedron
correspond to the following families of maps:
$AB$ - dephasing channels, $BD$ - classical bistochastic maps, 
$AD$ and $CD$ - depolarizing channels.
The families mentioned above are also shown in
Fig. \ref{fig:reny1} which presents boundaries of the set of all one--qubit
bistochastic channels projected onto the plane 
$(S_2^{\map},S_2^{\min})$. 
A following proposition proved in \cite{roga3} 
characterizes this projection.
\begin{proposition}
Extremal lines of asymmetric tetrahedron
correspond to boundaries of the set of all bistochastic
one--qubit maps on the plot $(S_2^{\map},S_2^{\min})$.
\end{proposition}

\subsection{Depolarizing channels}\label{depol}
 
Fig. \ref{fig:reny2} and Fig. \ref{fig:reny3}  
show the projection of the Pauli channels on
the plane $(S_q^{\map},S_q^{\min})$
with parameter $q$ different than $2$.
Comparison of these figures with Fig. \ref{fig:reny1}
shows that the structure of the set
of channels on the plane $(S_q^{\map},S_q^{\min})$
is the simplest in case of the R\'{e}nyi entropy of order $q=2$.
In this case, the depolarizing channels form 
one of the 
edges of the set of all quantum one--qubit maps projected onto the plane.
Indeed the following theorem proved in \cite{roga3} 
confirms the special role
of depolarizing channels in the set of all quantum channels
acting on states of arbitrary dimension $N$. 

\begin{theorem}
\label{glodu}
Depolarizing channels have the smallest map R\'{e}nyi 
entropy $S_2^{\map}$ among all channels with the same
minimal output 
R\'{e}nyi entropy $S_2^{\min}$.
\end{theorem}

The 
family of depolarizing channels is represented
in the plane $(S_2^{\map},S_2^{\min})$ by
the continuous line on the entire range
of $S_2^{\map}$. 
The minimal output entropy of a depolarizing channel $\Lambda_N$ acting
on $\c M_N$ is the following function of the map entropy
\begin{equation}
 S_2^{\r{min}} \Big(  S_2^{\r{map}}(\Lambda_N) \Big) = -\log\Big( \frac{1+N\r e^{- S_2^{\r{map}}(\Lambda_N)}}{N+1} \Big).
\label{zabamuminek}
\end{equation}
This is a monotonously increasing
function from $0$ to $\log{N}$. Therefore the following
theorem holds.

\begin{theorem}
Depolarizing channels have the greatest minimal output 
R\'{e}nyi entropy $S_2^{\min}$ among all maps of the same 
R\'{e}nyi entropy of a map $S_2^{\map}$. 
\end{theorem}

One can try to use the extremal position 
of depolarizing channels to estimate the 
minimal output entropy of some channels.
In the case of Hastings' counterexample for the additivity
conjecture the author showed that
 due to a maximally entangled
input state 
one can obtain smaller output entropy of the product 
of two channels than in the case of any product state
taken as an input. Let us estimate the R\'{e}nyi $q=2$ output
entropy for a product channel when the input
is maximally entangled. 
Following proposition proved in \cite{roga3} presents one of estimations.

\begin{proposition}\label{esest}
For any entropy $\c S$ which is subadditive  
the following inequality holds
\begin{equation}
|\c S^{\map}(\Phi_1)- \c S^{\map}(\Phi_2)|\leq  \c S\Big([\Phi_1\otimes\Phi_2](|\phi_+\>\<\phi_+|)\Big)\leq  \c S^{\map}(\Phi_1)+\c S^{\map}(\Phi_2),
\label{estmax}
\end{equation}
where $|\phi_+\>\<\phi_+|$ is a maximally entangled state.
\end{proposition}
\begin{proof}
The proof starts form the Lindblad inequality \cite{lindblad}, which is based on the subadditivity of the von Neumann entropy,
 \begin{equation}
\bigl| \c S(\rho) - \c S\big(\varsigma(\Phi,\rho)\big) \bigr| \le \c S(\Phi(\rho)) \le \c S(\rho) + \c S\big(\varsigma(\Phi,\rho)\big),
\label{lili}
\end{equation}
where $\varsigma(\Phi,\rho)=[\id\otimes\Phi]\Big(|\phi\>\<\phi|\Big)$ and $|\phi\>$ is a purification of $\rho$ as in (\ref{smark}).
The entropy of this state, $S(\varsigma(\Phi,\rho))$, is the exchange entropy which
does not depend on the choice of purification \cite{szuma}. The state $\varsigma$ 
defined for a channel $\Phi$ and the maximally mixed state $\rho_*=\idty/N$ 
is equal to the normalized dynamical matrix of $\Phi$ (\ref{miopio}),
\begin{equation}
\varsigma(\Phi,\rho_*)=\sigma_{\Phi}=\frac{1}{N}D_{\Phi},
\end{equation}
The entropy of this state defines $\c S^{\map}(\Phi)$.
Since the map $\Phi$ is trace preserving, the condition ${\rm Tr}_2\sigma_{\Phi}=\frac{1}{N}\idty$ holds, see (\ref{partialtrace}).
Apply Lindblad formula (\ref{lili}) to the state
\begin{equation}
[\Phi_1 \otimes \Phi_2] (|\psi_+\> \<\psi_+|)  = [\Phi_1 \otimes \id] \left( [\id \otimes \Phi_2] (|\psi_+\> \<\psi_+|) \right),
\end{equation}
where $|\psi_+\> =\frac{1}{\sqrt{N}}\sum_{i=1}^N|i\>\otimes|i\>$ is the maximally mixed state which is a purification of $\rho_*$.
Expression (\ref{lili}) applied to this state gives
\begin{eqnarray}
\bigl| \c S^{\r{map}}(\Phi_2) - \c S\bigl( \varsigma(\Phi_1\otimes\id,\, \sigma_{\Phi_2}) \bigr) \bigr| &\le&
\c S\bigl( (\Phi_1 \otimes \Phi_2) (|\psi_+\> \<\psi_+|) \bigr)\\ &\le& \c S^{\r{map}}(\Phi_2) + \c S\bigl( \varsigma(\Phi_1\otimes\id,\, \sigma_{\Phi_2}) \bigr).\nonumber
\end{eqnarray}
The exchange entropy $\c S\bigl( \varsigma(\Phi_1\otimes\id,\, \sigma_{\Phi_2}) \bigr)$ 
is the same as $\c S\bigl( \varsigma(\Phi_1,\, \tr_2 \sigma_{\Phi_2}) \bigr)$ 
since a purification of $\sigma_{\Phi_2}$ is as well the purification of $\tr_2\sigma_{\Phi_2}$ 
and the exchange entropy does not depend on a purification.
Due to the trace preservation formula, $\tr_2\sigma_{\Phi_2}=\rho_*$,
the state $\varsigma(\Phi_1,\, \tr_2 \sigma_{\Phi_2})=\varsigma(\Phi_1,\, \rho_*) = \sigma_{\Phi_1}$ which completes the proof.
\end{proof}

Proposition \ref{esest} is applicable for any 
entropy which is subadditive.
The R\'{e}nyi  entropy of order $q=2$ is not 
subadditive, however, it is a function of the
Tsallis entropy of order $2$ for which
the subadditivity holds. Therefore
Proposition \ref{esest} can be used to estimate the
output R\'{e}nyi $q=2$ entropy of a product channel
if the input state is maximally
entangled.
The following inequality corresponds to
R\'{e}nyi $q=2$ version of the lower bound in (\ref{estmax}),

\begin{equation}
-\log\Big(1-|e^{-S_2^{\map}(\Phi_1)}-e^{-S_2^{\map}(\Phi_2)}|\Big)\leq S_2\Big((\Phi_1\otimes\Phi_2)(|\psi_+\>\<\psi_+|)\Big).
\label{estmax2}
\end{equation}

It is possible to find channels $\Phi_1$ and $\Phi_2$ such that
the left hand side of (\ref{estmax2}) is 
greater than $S^{\min}_2$ of
depolarizing channel $\Lambda$, which has the same map entropy
as $S_2^{\map}(\Phi_1\otimes\Phi_2)$. 
Notice that
for any two channels the map entropy of
their tensor product is characterized by the following result.

\begin{proposition}
The R\'{e}nyi map entropy $S_q^{\map}$ is additive with respect to
tensor product of quantum maps for any parameter $q\geq 0$:
\begin{equation}
S_q^{\map}(\Phi_1\otimes\Phi_2)=S_q^{\map}(\Phi_1)+S_q^{\map}(\Phi_2).
\label{ludwiczek}
\end{equation}
\end{proposition}

\begin{proof}
The map entropy $S_q^{\map}(\Phi)$ is defined as the entropy of normalized dynamical matrix $D_{\Phi}$.
The matrix representation of $D_{\Phi_1 \otimes \Phi_2}$
is related to superoperator matrix of the quantum operation $\Phi_1 \otimes \Phi_2$, 
due to formula (\ref{zyrafa}).
Using explicit calculations on matrix elements one can show that $D_{\Phi_1 \otimes \Phi_2}$ 
is unitarily equivalent with $D_{\Phi_1} \otimes D_{\Phi_2}$. 
That implies the additivity of the map entropies, since the quantum 
R\'{e}nyi entropy of any order of a given state is a function of its spectrum.
 
Consider a set of $N$-dimensional matrices equipped with the Hilbert-Schmidt inner product
\begin{equation}
\< A | B \>_{\b h} := \tr A^\dagger B.
\end{equation}
In this space the matrix units $\bigl\{ |i\> \<j| \,\bigm|\, i,j = 1,2,\ldots,N\bigr\}$ 
form an orthonormal basis. The elements of this basis are denoted by $|i\> \<j| := |ij\>_{\b h}$. 
A quantum operation $\Phi$ is represented by a matrix $\hat\Phi$:
\begin{equation}
\<ij|\hat\Phi| k\ell\>_{\b h} = \tr  \Bigl( |j\> \<i|\, \Phi(|k\> \<\ell|) \Bigr),
\end{equation}
hence
\begin{equation}
\Phi(|k\> \<\ell|) = \sum_{i,j} \<ij|\hat\Phi| k\ell\>_{\b h}\, |i\> \<j|.
\label{rozklad}
\end{equation}
Due to the reshuffling procedure (\ref{zyrafa}),
the entries of the dynamical matrix $D_{\Phi}$ read
\begin{equation}
\<ab |D_\Phi|cd\>_{\b h} = \<ac|\hat\Phi| bd\>_{\b h}.
\label{tas}
\end{equation}
The entries of $D_{\Phi_1 \otimes \Phi_2}$ are 
obtained by using unnormalized maximally entangled state $|\Psi_+\> := \sum_{i,\ell} |i\ell\> \otimes |i\ell\>$
according to definition (\ref{miopio}) as follows 
\begin{align}
\, \<abcd|D_{\Phi_1 \otimes \Phi_2}|efgh \>
&= \<abcd |\bigl[ (\Phi_1 \otimes \Phi_2) \otimes \id \bigr] \bigl( |\Psi_+\> \<\Psi_+| \bigr)|efgh\>
\nonumber \\
&= \sum_{i,\ell,j,m} \<abcd |\bigl[ (\Phi_1 \otimes \Phi_2)(|i\ell\> \<jm|) \otimes |i\ell\> \<jm| \bigr]|efgh\>.
\label{rez}
\end{align}
Now expression (\ref{rozklad}) is used and the matrix elements of $D_{\Phi_1 \otimes \Phi_2}$ read
\begin{equation}
\<abcd | D_{\Phi_1 \otimes \Phi_2}|efgh\> = \sum_{\alpha,\beta,\gamma,\delta} \<\alpha\beta | \widehat{\Phi_1}| ij\>_{\b h}\, \<\gamma\delta |\widehat{\Phi_2}| ij\>_{\b h}\, \<abcd |\alpha\gamma i\ell\>\, \<\beta\delta jm |efgh\>.
\end{equation}
Since $\<abcd |\alpha\gamma i\ell\>$ is expressed in terms of Kronecker deltas $\delta_{a\alpha}\delta_{b\gamma}\delta_{ci}\delta_{d\ell}$ 
and $\<\beta\delta jm |efgh\>$ analogously, 
the summation over the Greek indexes gives,
\begin{align}
\<abcd| D_{\Phi_1 \otimes \Phi_2}|efgh\>
&= \<ac | D_{\Phi_1}| eg\>\, \<bd | D_{\Phi_2}| fh\>
\nonumber \\
&= \<acbd | D_{\Phi_1} \otimes D_{\Phi_2}| egfh\>.
\end{align}
The matrix $D_{\Phi_1 \otimes\Phi_2}$ is related to $D_{\Phi_1} \otimes D_{\Phi_2}$ 
by a unitary matrix \newline $U = \sum_{a,b,c,d} |abcd\> \<acbd|$. 
Therefore both matrices have the same eigenvalues and the same entropies.
\end{proof}

Since minimal output entropy of a depolarizing 
channel $\Lambda$ is a function of its map entropy (\ref{zabamuminek}),
the estimation on the left hand side of (\ref{estmax2})
can be made in terms of such $\Lambda$ for which $S_2^{\map}(\Lambda)=S_2^{\map}(\Phi_1)+S_2^{\map}(\Phi_2)$.
As a result of this estimation one obtains condition on
the pair of channels, for which a
maximally mixed input state 
does not decrease the output entropy
below the smallest value obtained by the product input
state,
\begin{equation}
1-\frac{MN+1}{MN}|e^{-S_2^{\map}(\Phi_1)}-e^{-S_2^{\map}(\Phi_2)}|\leq e^{-S_2^{\map}(\Phi_1\otimes\Phi_2)}=e^{-\big[S_2^{\map}(\Phi_1)+S_2^{\map}(\Phi_2)\big]},
\label{estmax3}
\end{equation}
where $\Phi_1$ acts on $\c M_N$ and $\Phi_1$ on $\c M_M$.
\begin{figure}[htbp]
	\centering
	\scalebox{0.9}{
		\includegraphics{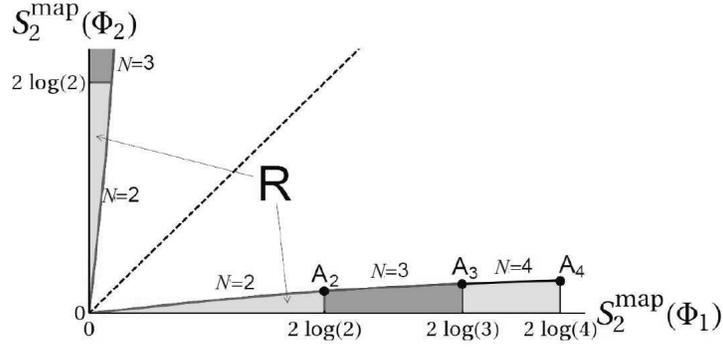}}
\caption{
Colored parts of the figure denote the
region described by inequality (\ref{estmax3}).
This region contains pairs of maps characterized
by their map entropy for which the additivity
is conjectured. The dotted line contains the
pairs of complementary channels. The region
is enlarged if a larger dimensions are considered.
}
	\label{fig:reg}
\end{figure}
\begin{figure}[htbp]
	\centering
	\scalebox{0.9}{
		\includegraphics{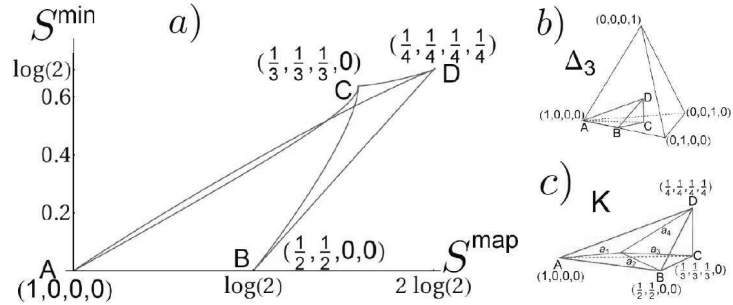}}
\caption{
The set of the Pauli channels projected on the plane spanned by the
map entropy $S^{\map}$ and the minimal output entropy $S^{\min}$. 
The von Neumann entropies are considered.
Solid curves correspond to the edges of the asymmetric
tetrahedron K. The curve AD characterizing the family of
depolarizing channels does not belong to the boundary of the set.
}
	\label{fig:reny2}
\end{figure}
Fig. \ref{fig:reg} presents the region defined 
by (\ref{estmax3}). Such a set 
is not empty and contains maps, for which
$S^{\map}_2(\Phi_1)<<S^{\map}_2(\Phi_2)$
or $S^{\map}_2(\Phi_2)<<S^{\map}_2(\Phi_1)$.
The dotted line represents the set
of complementary channels for which both map entropies are equal. This set contains
the channels breaking the conjecture of 
additivity of minimal output entropy according to the proof of Hastings.
The region defined by (\ref{estmax3}) does not 
intersect the set. 
It was also shown \cite{amosov}, \cite{diVincenzo} 
that additivity holds if one of the channels
is unitary or if one of the channels is
very noisy. These both cases are covered
by condition (\ref{estmax3}). These
examples support formulation of 

\begin{conjecture}[\cite{roga3}]
The additivity of minimal output R\'{e}nyi $q=2$
entropy holds for pair of channels satisfying inequality (\ref{estmax3}).                                            
\label{conjadd}
\end{conjecture}

Recent literature does not answer the question, whether the
additivity conjecture
is broken for low dimensional channels and the R\'{e}nyi
entropy of order $q=2$. Our Conjecture \ref{conjadd}
suggests for which pairs of channels finding a counterexample of additivity
is unlikely. 
\begin{figure}[htbp]
	\centering
	\scalebox{1}{
		\includegraphics{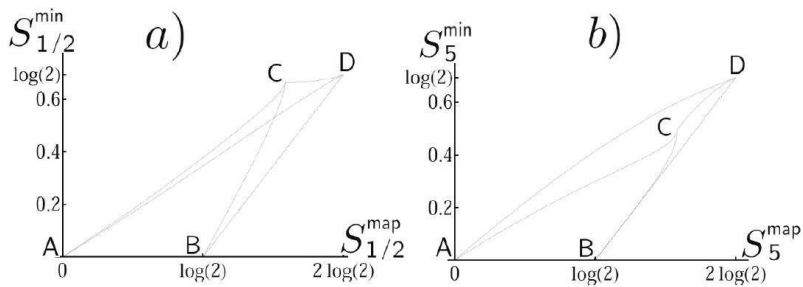}}
\caption{
As in Fig. \ref{fig:reny2}{\it a}: projection of the set
of Pauli channels onto the plane spanned
by the R\'{e}nyi entropy of a map $S_q^{\map}$
and the minimal R\'{e}nyi output entropy $S_q^{\min}$
obtained for $a)$ $q=1/2$ and $b)$ $q=5$. 
}
	\label{fig:reny3}
\end{figure}
Conjecture \ref{conjadd} uses the map entropy and is formulated for the
R\'{e}nyi entropy of order $2$, for which 
the theorem about extremal position of the
depolarizing channels was proved. This
is the key theorem which allows us to derive
estimations (\ref{estmax2}) and (\ref{estmax3}).
Numerical tests (Fig. \ref{fig:reny2} and \ref{fig:reny3})
suggest that the depolarizing channels are not 
situated at the boundary of the set of all channels
in the plane $(S_q^{\map},S_q^{\min})$ for
$q\leq2$, while their extremal position could
be confirmed in case $q\geq 2$.
Nevertheless, the R\'{e}nyi entropy is a smooth 
function of $q$. Therefore, a conjecture 
similar to Conjecture \ref{conjadd} may hold also
for other values of the R\'{e}nyi parameter $q$.

\subsection{Transformations preserving minimal output entropy}\label{mlotek}

In previous chapter the set of one--qubit quantum 
operations was considered in context of the plot $(S_q^{\map},S_q^{\min})$.
One could ask, whether the family of
maps lying at the same vertical or horizontal line can be 
characterized.
The following section gives a partial answer to this question.
Transformations of 
one qubit maps which preserve the
minimal output entropy will be considered. 
Such a transformation changes the quantum 
channel
and moves the corresponding point in the plane $(S_q^{\map},S_q^{\min})$
along a given
 horizontal line.
In the following section we consider
the geometrical picture 
of one--qubit maps acting on the set of pure states.
One--qubit quantum operation transforms the Bloch ball
into an ellipsoid inside the ball.
A transformation of quantum operation
 which changes the lengths of the axes
of the ellipsoid and their orientation 
and leaves the minimal output entropy unchanged
will be studied. 

Consider the superoperator matrix 
of a one--qubit quantum operation:
\begin{equation}
\Phi=
\begin{pmatrix}
\Phi_{11} &\Phi_{12} &\overline{\Phi_{12}} &\Phi_{14} \\
\Phi_{21} &\Phi_{22} &\overline{\Phi_{32}} &\Phi_{24} \\
\overline{\Phi_{21}} &\Phi_{32} &\overline{\Phi_{22}} &\overline{\Phi_{34}} \\
1-\Phi_{11} &-\Phi_{12} &-\overline{\Phi_{12}} & 1-\Phi_{14}
\end{pmatrix}.
\label{operacja}
\end{equation}
Parameters $\Phi_{11}$ and $\Phi_{14}$ are real, the complex conjugation of $\Phi_{ij}$ is denoted by $\overline{\Phi_{ij}}$. 
The form (\ref{operacja}) guarantees 
that the dynamical matrix of $\Phi$ is Hermitian and the trace preserving condition (\ref{partialtrace}) is satisfied.

Assume that the quantum operation $\Phi_1$ has the output entropy minimizer at the point 
\begin{equation}
\rho_p=\begin{pmatrix}
p &\sqrt{p(1-p)}  \\
\sqrt{p(1-p)} &1-p  
\end{pmatrix}.
\label{czystybezfazy}
\end{equation}
Such an assumption is not restrictive since
one can always treat the operation $\Phi_1$ as a
 concatenation of a given operation with a
unitary rotation which does not change the minimal output entropy.
The
quantum operation (\ref{operacja}) acting on 
a pure state 
\begin{equation}
\rho_{in}=\begin{pmatrix}
a &\sqrt{a(1-a)}  \\
\sqrt{a(1-a)} &1-a  
\end{pmatrix}
\end{equation}
gives an output state
\begin{equation}
\begin{aligned}\rho_{out}\ \ =\ \ 
&a\begin{pmatrix}
\Phi_{11} &\Phi_{21}  \\
\overline{\Phi_{21}} &1-\Phi_{11}  
\end{pmatrix}+
(1-a)\begin{pmatrix}
\Phi_{14} &\Phi_{24}  \\
\overline{\Phi_{24}} &1-\Phi_{14}  
\end{pmatrix}+\\
&+\sqrt{a(1-a)}\begin{pmatrix}
2\Re(\Phi_{12}e^{i\phi}) &\Phi_{22}e^{i\phi}+\overline{\Phi_{32}}e^{-i\phi}  \\
\overline{\Phi_{22}}e^{-i\phi}+\Phi_{32}e^{i\phi} &-2\Re(\Phi_{12}e^{i\phi})  
\end{pmatrix},
\end{aligned}
\label{output}
\end{equation}
which attains the minimum entropy if $a=p$.
\begin{itemize}
\item Transformation changing the lengths of the axes of the ellipsoid.

Consider a quantum operation $\Phi_{ellipsoid}$, which transforms
the Bloch ball into such an ellipsoid that the end of its longest axis touches
the Bloch sphere in the "North Pole",
\begin{equation}
\Phi_{ellipsoid}=\left(
\begin{array}{cccc}
 1 & 0 & 0 & 1-\text{$\eta $3} \\
 0 & \frac{\text{$\eta_1 $}+\text{$\eta_2 $}}{2} & \frac{\text{$\eta_1 $}-\text{$\eta_2 $}}{2} & 0 \\
 0 & \frac{\text{$\eta_1 $}-\text{$\eta_2 $}}{2} & \frac{\text{$\eta_1 $}+\text{$\eta_2 $}}{2} & 0 \\
 0 & 0 & 0 & \text{$\eta_3 $}
\end{array}
\right).
\end{equation}
Suitable rotations of the Bloch ball before and after the action of
 $\Phi_{ellipsoid}$ 
guarantees that the point of contact with the Bloch sphere is the minimizer of $\Phi_1$. 
Therefore the concatenation of $\Phi_1\cdot\Phi_{rotation}\cdot\Phi_{ellipsoid}\cdot\Phi_{rotation}$ 
has the same minimal output entropy and the same minimizer that $\Phi_1$. The rotation operation is given by
\begin{equation}
\Phi_{rotation}=\left(
\begin{array}{cccc}
 p & -\sqrt{(1-p) p} & -\sqrt{(1-p) p} & 1-p \\
 \sqrt{(1-p) p} & p & p-1 & -\sqrt{(1-p) p} \\
 \sqrt{(1-p) p} & p-1 & p & -\sqrt{(1-p) p} \\
 1-p & \sqrt{(1-p) p} & \sqrt{(1-p) p} & p
\end{array}
\right),
\end{equation}
where $p$ is defined by the minimizer of output entropy for $\Phi_1$.
This transformation changes the lengths of axes of the ellipsoid but it does not change the point at the ellipsoid which is
 the closest to the Bloch sphere. In other words, this transformation does not change the
 directions of the axes of the image of $\Phi_1$ into the Bloch ball, but only their lengths.

\item Transformation changing directions of the axis. 

The next transformation changes directions of axes of an ellipsoid but preserves the entropy minimizer. 
In particular, if the image of the minimizer is on the longest axis of an ellipsoid, after the transformation the point 
which is the closest to the Bloch sphere 
is no longer on the main axis of the ellipsoid.

Entropy of an output state (\ref{output}) is a function of its determinant. The
 minimum of the determinant determines the minimum of the entropy. 
Consider a transformation which preserves the value of the determinant and compute its derivative in a minimizer. 
It is useful to introduce the compact notation of Eq. (\ref{output}):
\begin{equation}
\rho_{out}=aA+(1-a)B+\sqrt{a(1-a)}C,
\label{output2}
\end{equation}
where matrices $A, B$ and $C$ correspond to the
matrices (\ref{output}).
Consider  a transformation $\Phi_1\rightarrow\Phi_1+\Phi_{direction}$.
The output of  $\Phi_1+\Phi_{direction}$ is given by
\begin{equation}
\rho'=a(A+\frac{1}{2}\frac{\sqrt{1-p}}{\sqrt{p}}X)+(1-a)(B+\frac{1}{2}\frac{\sqrt{p}}{\sqrt{1-p}}X)+\sqrt{a(1-a)}(C-X),
\label{output3}
\end{equation}
where $X$ is a matrix, which is hermitian and has trace equal to zero.
Moreover, the matrix $X$ satisfies the condition guaranteeing that $\Phi_1+\Phi_{direction}$ is completely
positive.
The state $\rho'$ coincides with (\ref{output2}) if $a=p$. 
Moreover, the derivative of formula (\ref{output2}) 
with respect to $a$ is the same as 
the derivative of Eq. (\ref{output3}) at the point $a=p$. 
Therefore, the determinants of (\ref{output2}) and (\ref{output3})
 are the same and
the derivative at $a=p$ is equal to zero.
A proper choice of parameters in $X$ guarantees that 
 there is a minimum at point $a=p$.
Hence both maps, $\Phi_1$ and $\Phi_1+\Phi_{direction}$ have the same minimal output entropy. 

The part $\Phi_{direction}$ can be characterized by two parameters $(t,n)$,
\begin{equation}
\Phi_{direction}=\frac{1}{2}\left(
\begin{array}{cccc}
  \sqrt{\frac{1-p}{p}}\ t & -t & -t & \sqrt{\frac{p}{1-p}}\ t \\
  i \sqrt{\frac{1-p}{p}}\ n & -i\ n & -i\ n & i \sqrt{\frac{p}{1-p}}\ n \\
 -i \sqrt{\frac{1-p}{p}}\ n & i\ n & i\ n & -i \sqrt{\frac{p}{1-p}}\ n \\
 - \sqrt{\frac{1-p}{p}}\ t & t & t & - \sqrt{\frac{p}{1-p}}\ t
\end{array}
\right).
\end{equation}
\end{itemize}
Such a form guarantees that the output state of $\Phi_1+\Phi_{direction}$ is given by Eq. (\ref{output3}).

The map $\Phi_2$ of the same minimal output entropy as $\Phi_1$
obtained by joint action of three transformations, 
$\Phi_{rotation}$, $\Phi_{ellipsoid}$ and $\Phi_{direction}$, on $\Phi_1$  can be given by:
\begin{equation}
\Phi_2=\Phi_1\Phi_{rotation}\cdot\Phi_{ellipsoid}\cdot\Phi^{T}_{rotation}+\Phi_{direction}.
\end{equation}
We are not able to prove that this transformation contains
all possibilities of obtaining maps with the same 
minimal output entropy as a given one,
however, the
transformation is characterized by $5$ parameters and also $5$ 
parameters are needed to have all different (up to one rotation) 
ellipsoids tangent to the sphere on its inner side in a given point. 
Three parameters are associated with the lengths of axes $|\eta_1|, |\eta_2|, |\eta_3|$, while
 two parameters define the direction of the longest axis $n, t$.

Above considerations introduce a $5$-parameter transformation
of a quantum map $\Phi_1\rightarrow\Phi_2$. 
The transformation preserves the minimal output entropy. 
Therefore, it determines the family of maps which are situated
at the same horizontal line of the plot $(S^{\map},S^{\min})$.
Characterization of the family of quantum maps parametrized
by the minimal output entropy can be useful to
further investigations of relations between 
$S^{\min}$ and $S^{\map}$ and their consequences. 
  

\section{Davies maps for qubits and qutrits}\label{bida}

Explicit description of general continuous dynamics 
of an open quantum system
is difficult in practice. 
Exact formulas describing the time
evolution are known in some special cases only.
One of the cases in which the problem can be
solved uses the assumption of a week coupling \cite{alickii} 
of a low dimensional quantum system interacting with
much bigger reservoir in the thermal equilibrium.
Such an interaction changes only the state of the system 
whereas the state of the environment remains unchanged.
By analogy to the classical process, 
in which the evolution of a state does not depend on the history,
such an evolution is called a {\it Markov process}.

However, while analysing the continuous evolution of the input state,
sometimes there is no need to know the entire time evolution
since only the output state is relevant.
The "black box" description is useful in such cases.
A "black box" acts like an evolution discrete in time and can be 
 described using completely positive maps, 
represented as matrices of superoperators. 


\begin{figure}[htbp]
	\centering
	\scalebox{0.9}{
		\includegraphics{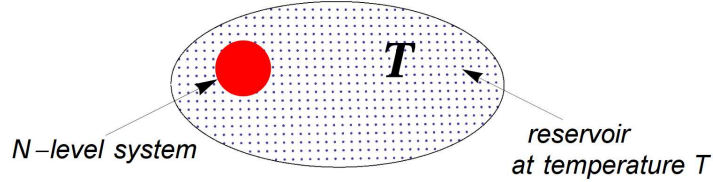}}
\caption{
Model of a quantum $N$--level system
characterized by Hamiltonian $H$
interacting with a much larger environment
in a thermal equilibrium at temperature $T$.
}
	\label{fig:ter}
\end{figure}
The following chapter distinguishes a concrete class of physical processes 
described by a {\it Davies map} \cite{davies}.
Such a process is compatible with the interaction of a quantum state 
with an environment in a given temperature, see Fig. \ref{fig:ter}.
Due to a suitable choice of the entries of a superoperator matrix $\Phi$ and relations between them 
one can say whether some continuous time evolution is described by a given discrete quantum map. 
The solution concerns the maps acting on one--qubit, $N=2$,
and one--qutrit, $N=3$.
In the case of one-qubit maps we determine the state which is the most resistant 
on Davies channels. 
It will be shown that the maximal output $2$--norm of Davies maps
is additive with respect to the tensor product of two such maps.

\subsection{Quantum Markov process}

The quantum Markov process is characterized 
by quantum maps belonging to the one-parameter 
completely positive semigroup, $\Phi_t=\exp{\c Gt}$, where 
$\c G$ denotes a generator and  positive parameter $t$ is
associated with time. 

The most general form of the generator of a 
completely positive semigroup was given by
Gorini, Kossakowski, Sudarshan \cite{gorini} and Lindblad \cite{lindbl}.
It can be written as
\begin{equation}
\c G=i\delta+\c L,
\label{gorini}
\end{equation}
where $\delta$, given by the commutator with
the effective Hamiltonian of the system $\delta:\rho\rightarrow[\rho,H]$, describes the unitary part of the evolution. The dissipative part $\c L$
has the Lindblad form
\begin{equation}
\c L:\rho\rightarrow\sum_{\alpha}\Big(K^{\alpha}\rho K^{\alpha\dagger}-\frac{1}{2}\{K^{\alpha\dagger}K^{\alpha},\rho\}\Big),
\label{lifor}
\end{equation}
where $\{A,B\}=AB+BA$ is anticommutator, while operators $K^{\alpha}$
can be associated with the Kraus representation of the quantum operation.


Deciding whether a given superoperator matrix belongs
 to the completely positive semigroup 
was shown \cite{wolfa2} to be a
problem 'NP' hard with respect to the dimension $N$.
Nevertheless, some additional assumptions
allow one to characterize matrices
from completely positive semigroups at least
for a few low dimensions. In following chapter,
such a solution will be given for $N=2$, and $N=3$,  under additional conditions:
independence of unitary and dissipative parts of the
evolution and the detailed balance condition. 
These three conditions define the
so--called Davies maps \cite{davies}. Sometimes
the uniqueness of the invariant state is 
also added to the definition.

\subsection{Characterization of the model}

Consider a quantum $N$ - level system characterized
by the Hamiltonian in its eigenbasis,
\begin{equation}
H=\sum_{i=1}^N\epsilon_i|i\>\<i|.
\end{equation}
Assume that such a system is weekly coupled
to the environment of a given temperature $T$, see Fig. \ref{fig:ter}.
An interaction with the environment
preserves one invariant state, which is
the Gibbs state
\begin{equation}
\rho_{\beta}=\frac{1}{\c Z}\exp(-\beta H),
\label{knur}
\end{equation}
where $\c Z=\sum_{i=1}^N\exp(-\beta \epsilon_i)$ is a partition function and $\beta=\frac{1}{kT}$. Here
$k$ represents the Boltzmann constant.
A quantum map $\Phi$ satisfies
the {\it detailed balance} condition if it is 
Hermitian with respect to the scalar product 
defined by the Gibbs state
\begin{equation}
\tr \rho_{\beta}A\Phi^*(B)=\tr\rho_{\beta}\Phi^*(A)B,
\label{detbal}
\end{equation}
where $A$ and $B$ are arbitrary observables and $\Phi^*$ 
the quantum operation in the Heisenberg picture.
Detailed description of this condition
can be found in \cite{db1}.

The name "detailed balance" was taken from
the theory of stochastic processes.
Detailed balance means that in an equilibrium state any two levels
of the evolving system remain in an equilibrium: 
the rate of transition from the level $i$ to $j$ 
and the transition rate from $j$ to $i$ are equal. 
Mathematical formula describing this fact reads
\begin{equation}
\c F_{ij}p_{i}=\c F_{ji}p_j,
\label{cdetbal}
\end{equation}
where $\c F_{ij}$ are entries of a stochastic transition matrix
and $p_i$ represent the components of the invariant probability vector.


\subsection{Matrix representation of Davies maps}\label{sec:mat}

One qubit map in the "black box" description is represented 
by a superoperator matrix. It is a matrix 
acting on the vector formed by the
entries of a density matrix
ordered in a single column.
A superoperator $\Phi$ represents a Davies map,
if the following conditions are satisfied.
\begin{itemize}
\item The map $\Phi$ is completely positive. 

This condition 
is guaranteed if the Choi--Jamio{\l}kowski matrix $D_{\Phi}$ (\ref{miopio})
of  the map is positive. One has to reshuffle the elements of the matrix $\Phi$
according to (\ref{zyrafa}) and check positivity of the resulting dynamical matrix $D_{\Phi}=\Phi^{R}$.

\item Superoperator $\Phi$ belongs to 
the semigroup of completely positive maps.

This is equivalent to existence 
of a generator $\c G$ of the Lindblad form (\ref{lifor})
and the parameter
 $t\geq 0$ such that $\c Gt=\log\Phi$.
Knowing the logarithm of $\Phi$ one has to 
determine whether it is of the Lindblad form.
It was shown in \cite{wolfa1} that if the
Choi-Jamio{\l}kowski matrix of a given generator is positive
in the subspace orthogonal to the maximally entangled state, then
the generator can be written in the Lindblad form.

It is not a trivial task to write an
analytical expression for the logarithm of a
given matrix if its dimension is greater than two. 
Such a 
problem for $3\times 3$ stochastic matrices is discussed 
in the last section of the following chapter.

\item Since the rotational part of
the evolution is independent 
of the dissipative (contractive) part, 
the structure of the superoperator 
is restricted to the block diagonal form. 
Off--diagonal elements of the density matrix
are just multiplied by numbers, while
the diagonal elements can be mixed between themselves.
More detailed discussion on this property
is given in Section \ref{qut}.

\item The detailed balance condition
introduces further restrictions
on the elements of 
the block acting on the diagonal
part of the density matrix.
This block is a stochastic matrix, the entries of which
 satisfy Eq. (\ref{cdetbal}).
\end{itemize}

Since now, only the dissipative part of the 
evolution will be considered.
Due to the above conditions 
the dissipative part of the generator of the one--qubit Davies maps
can be written as
\begin{equation}
 \c L_{\alpha,\lambda,p}=\begin{pmatrix}
-\alpha &0 &0&\alpha\frac{p}{1-p}  \\
0 &\lambda &0&0  \\
0 &0 &\lambda&0  \\
\alpha &0 &0&-\alpha\frac{p}{1-p}  \end{pmatrix},
\end{equation}
while
the corresponding superoperator acting on two-dimensional 
states (in the Hamiltonian basis) has the form
\begin{equation}
\Phi_{a,c,p}=\begin{pmatrix}
1-a &0 &0&a\frac{p}{1-p}  \\
0 &c &0&0  \\
0 &0 &c&0  \\
a &0 &0&1-a\frac{p}{1-p}  
\label{dav}
\end{pmatrix}.
\end{equation}
Here, $p$ is a function of temperature,
$p=\big(1+\exp{(-\frac{\epsilon}{(kT}})\big)^{-1}$, which determines 
the invariant state
\begin{equation}
\Phi_{a,c,p}(\rho_*)=\rho_*=\begin{pmatrix}p&0\\0&1-p\end{pmatrix}.
\end{equation}
Notice that (\ref{dav}) has a block diagonal form
which is a consequence of independence of 
rotational and contractive evolution.
This is also equivalent to
independence of changes in diagonal and off--diagonal
entries of a density matrix. The detailed balance
condition $(\ref{cdetbal})$ implies the form of the outer block in Eq. (\ref{dav}).
One--qubit Davies maps form a three-parameter family characterized
by $(a,c,p)$, where $p$ is a function of the temperature.
Conditions that
such a matrix is an element of the semigroup
of completely positive maps introduce the following
restrictions on the parameters $(a,c,p)$:
\begin{equation}
a+p<1,\qquad 0<c<\sqrt{1-\frac{a}{1-p}}.
\label{davies3}
\end{equation}

Equality $\Phi=\exp{\c Lt}$ allows one to write
explicit formulas for time dependence
of parameters $a$ and $c$,
\begin{equation}
a=(1-p)\Big(1-\exp(-At)\Big), \qquad c=\exp(-\Gamma t),
\end{equation}
where $A$ and $\Gamma$ are parameters such that $A\geq\frac{1}{2}\Gamma\geq0$.
The entire paths of the semigroup are
showed in Fig. \ref{fig:pat}
\begin{figure}[htbp]
	\centering
	\scalebox{1}{
		\includegraphics{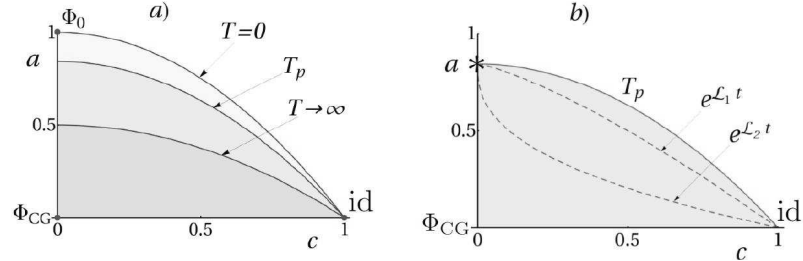}}
\caption{
Panel $a)$ contains the region 
of parameters $(a,c)$ allowed by
 relation (\ref{davies3}) and describing the one--qubit Davies maps. 
The upper border lines
are also drawn for different temperature $T$.
Panel $b)$ shows the region allowed for a given
temperature $T$. The lines describe two
semigroup corresponding to two different
randomly chosen generators $\c L_1$ and $\c L_2$.
The extremal lines corresponding to 
the solid line on panel $a$ describe
the semigroup with the smallest ratio 
of decoherence to the damping rate.
Maps $\id, \Phi_0, \Phi_{CG}$ are
the identity channel, completely depolarizing
and coarse graining channel respectively.
}
	\label{fig:pat}
\end{figure}

One--qubit Davies map can be
written
using the Bloch parametrization (\ref{bloch})
\begin{equation}
\Phi=\bordermatrix{ 
&{}\cr
& 1  & 0 & 0 & 0\cr 
& 0  & \eta_1 & 0 & 0\cr
& 0  & 0 & \eta_1 & 0\cr
& \kappa_3 & 0 & 0 & \eta_3 \cr}.
\label{dt}
\end{equation}
where $|\eta_i|$ denote the lengths
of   axes of the ellipsoid and $\vec{\kappa}$ is the translation vector. 
These parameters are related to the parameters $(a,c,p)$
\begin{eqnarray}
\eta_1=c\geq0,\qquad \eta_3=1-\frac{a}{1-p}\geq0,\nonumber\\
\kappa_1=\kappa_2=0,\qquad \kappa_3=a\frac{2p-1}{1-p}\geq0. \label{kolor}
\end{eqnarray}
The image of the set of pure states under an action of one--qubit Davies map
is shown in Fig. \ref{fig:eli}.
The image of the Bloch ball forms an ellipsoid with rotational symmetry.
Fig. \ref{fig:eli} presents the image of an exemplary one--qubit
Davies map for which $\eta_1\geq\eta_3$,
however conditions (\ref{davies3}) admits also
the case $\eta_3\geq\eta_1$.

\begin{figure}[htbp]
	\centering
	\scalebox{0.6}{
		\includegraphics{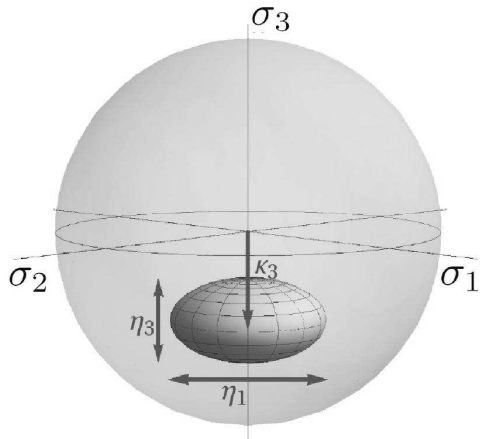}}
\caption{
Ellipsoid obtained by an action of a 
one--qubit Davies channel
on the Bloch sphere. The channel is
characterized by parameters $(\eta_1,\eta_3,\kappa_3)$
defined in (\ref{kolor}).
}
	\label{fig:eli}
\end{figure}

\subsection{Physical examples}

Qubit maps of the structure similar to (\ref{dt})
were analysed before in context
of quantum optics. 
The unitary evolution is induced by the laser field,
while the dissipative dynamics is caused by an interaction with
the environment. The state of a two level atom is 
characterized by the Bloch vector $(x,y,z)$, where $z$ represents
the difference between the diagonal entries of a density matrix 
equal to the inversion of populations of the atomic levels.
Variables $x$ and $y$ are associated with the atomic
dipole operators. 
The evolution in this set
has been defined by means of variables
describing the  decay  rate  $\tau_1$ of the coherences 
and the rate $\tau_2$ of attaining the equilibrium state. 
These parameters
correspond to the variables considered in 
the Section \ref{sec:mat}, $\eta_1=\exp{(-t/\tau_1)}$ and $\eta_3=\exp{(-t/\tau_3)}$ 
which are related
to squeezing of the axes of the ellipsoid.
Formula (\ref{davies3}) corresponds
to the relation between the decay rates:
\begin{equation}
\tau_1\leq 2\tau_3.
\end{equation}
This relation was obtained 
by analysing a concrete physical model
of the evolution of the two level system by means 
of Bloch equations \cite{kimura}. 
The one--qubit operations (\ref{dt})
were also studied by \cite{wodkiewicz}.

\subsection{Minimal output entropy of Davies maps}

In context of transmission of quantum information,
it is natural to ask, which pure states are the most
resistant with respect to the changes caused by the Davies maps.
The answer depends on the selected measure of decoherence.
Such a measure can be described, for example, by means of some 
matrix norm of the output state maximized over
the input states.
Among quantities measuring the decoherence,
the minimal output entropy is of special importance
because 
some questions concerning the channel capacity,
such as additivity problem, can be related with similar problem written
 in terms of minimal output entropy. 
The minimal output entropy is related to the maximal norm of 
the output state if the input is pure. 

Since a Davies map has rotational symmetry,
the minimizer can be chosen to be a real state:
\begin{equation}
\rho=\begin{pmatrix}
  \mu&\nu\\
\nu&1-\mu    
     \end{pmatrix},
\end{equation}
where $\nu^2=(1-\mu)\mu$ since the state is pure.
After an action of the operation (\ref{dav}) this state is transformed into
\begin{equation}
\rho'=\begin{pmatrix}
  (1-a)\mu+b(1-\mu)&c\nu\\
c\nu& a\mu+(1-b)(1-\mu)    
     \end{pmatrix},
\end{equation}
where $b=ap/(1-p)$.
Computing the eigenvalues and minimizing the entropy over $\mu$
one can characterize the minimizer in two cases:
\begin{itemize}
 \item If $c^2 \leq (1-a-b)(1-2b)$ the minimizer is characterized by $\mu=0$ and it 
forms an eigenstate of the Hamiltonian $H$.
\item If $c^2 \geq (1-a-b)(1-2b)$ the minimizer is characterized by
\begin{equation}
 \mu=\frac{(a+b-1)(2b-1)-c^2}{2(a+b-1)^2-2c^2}.
\end{equation}
It is no longer the eigenvalue of the Hamiltonian,
however, after some time of the evolution $t\gg 0$ the second
case changes into the first one and the 
minimizer is a state ${\rm diag} (0,1)$. 
This is an eigenstate of the Hamiltonian.
The situation that the minimizer is in the vector ${\rm diag} (0,1)$ reminds the
classical evolution of two--dimensional vector governed 
by the stochastic matrix. In this case the extremal vector like $(0,1)$ is 
the minimizer of the Shannon entropy of the output.
\end{itemize}

\subsection{Multiplicativity of maximal output norm of one--qubit Davies maps}\label{multipy}

As discussed in the introduction 
to Chapter \ref{sec:ent}, 
the question of additivity of 
minimal output von Neumann entropy with respect
to the tensor product of quantum operations
is one of the most interesting problem
in quantum information theory.
This problem can be equivalently stated 
in terms of channel capacity.
In general, the conjecture on additivity
of channel capacity is false,
however, there is still an interesting 
problem, for which class
of maps the conjecture can be confirmed.

Recent studies of the problem use the notion 
of the R\'{e}nyi entropy of order $q$.
This entropy tends to the von Neumann version as $q\rightarrow 1$.
The problem of additivity of minimal output R\'{e}nyi $q$ entropy is
directly related to multiplicativity of the maximal output Schatten $q$--norm. This norm is defined as
\begin{equation}
\|X\|^S_q=(\tr |X|^q)^{1/q},
\end{equation}
where $|X|=\sqrt{X^{\dagger}X}$.
Maximal Schatten $q$ norm of a quantum map $\Phi$ is:
\begin{equation}
 \|\Phi\|^S_q:=\max_{\rho}(\tr|\Phi(\rho)|^q)^{1/q},
\label{normoperat2}
\end{equation}
where maximization is taken over the entire set of density matrices $\rho$.
The R\'{e}nyi entropy of order $q$ of a state $\rho$ 
can be defined  as follows \cite{alic}
\begin{equation}
 S_q(\rho)=\frac{q}{1-q}\log\|\rho\|^S_q.
\end{equation}
Due to logarithm in this formula
the multiplicativity of maximal $q$--norm
is equivalent to the additivity of minimal output entropy $S^{\min}_q$.

In this section, multiplicativity of operator 2--norm
induced by the Euclidean vector norm will be proved for
the quantum one--qubit Davies maps. 
This vector induced norm 
is not related to the R\'{e}nyi entropy by such an elegant formula 
 like it is in the case for Schatten norm, however, it is a bit easier to calculate than  the Schatten counterpart.
These particular results support the general solution for multiplicativity problem 
 for Schatten 2--norm
which implies the additivity property for minimal output R\'{e}nyi entropy of order 2 and which
has been already proved for general one--qubit quantum operations \cite{koldman} (see also \cite{kingostatni}).



\subsubsection{Outline of the proof of  multiplicativity}\label{puszak}

The Euclidean norm (2--norm) of a vector $x=(x_{1},...,x_{n})$ is defined as:
\begin{equation}
\left\|x\right\|_2=\sqrt{\sum_{i=1}^n\ \left|x_{i}\right|^{2}}.
\end{equation}
This vector norm induces the 2--norm of an operator $A$:
\begin{equation}
\left\|A\right\|_2={\rm max}_{x\neq0}\frac{\left\|Ax\right\|_2}{\left\|x\right\|_2}.
\end{equation}
One of the property of this norm (see \cite{ORT}) is
that $\|A\|_2$ is equal to square root of the spectral radius of $A^{\dagger}A$ 
or equivalently to the greatest singular value 
of the matrix $A$,
\begin{equation}
\left\|A\right\|_2=\left[r(A^{\dagger}A)\right]^{1/2},
\label{specrad}
\end{equation} 
where a spectral radius $r(A^{\dagger}A)={\rm max}_{i}\left|\xi_{i}\right|$
and $\xi_{i}$ are eigenvalues of  $A^{\dagger}A$. 
In this section the maximal two norm of the output of a quantum map $\Phi:\c M_N\rightarrow\c M_N$
will be considered
\begin{equation}
M_{\Phi}=\max_{\rho\in\c M_N}\|\Phi(\rho)\|_2=\max_{A\geq0}\frac{\|\Phi(A)\|_2}{\tr A}.
\label{defztr}
\end{equation}
One can ask, whether the maximal two--norm is multiplicative in a sense:
\begin{equation}
M_{\Phi\otimes\Omega}=M_{\Phi}M_{\Omega}.
\end{equation}
It will be shown that if $\Phi$ is one--qubit Davies map and $\Omega$ 
is an arbitrary quantum map acting on $N$--dimensional state the multiplicativity holds. 

The idea of the proof of the theorem given below
is borrowed from the paper of King and Ruskai \cite{king}. 
These authors prove an analogical theorem about a 
bistochastic quantum
map $\Phi$. 
They noted that the same proof holds as well for 
stochastic one--qubit maps.
Here we will present an explicit calculations for the case 
of Davies maps with $|\eta_3|\leq|\eta_1|$.
%

\begin{theorem}
Let $\Phi:\c M_2\rightarrow\c M_2$ be an one--qubit Davies map and 
$\Omega:\c M_N\rightarrow\c M_N$ be an arbitrary quantum map. 
The maximal two norm of the output is multiplicative:
\begin{equation}
M_{\Phi\otimes\Omega}= M_{\Phi}M_{\Omega}.
\label{puchatek}
\end{equation}
\label{thor}
\end{theorem}

In this section the sketch of the proof will be given,
while some details of the
calculation will be presented in the next
section.
 In order to present the proof we need
to introduce the following set. 
An arbitrary density matrix on $\c H_{2}\otimes \c H_{N}$ 
can be written as a block matrix
\begin{equation}
\rho=\bordermatrix{         
&{}\cr
& \rho_{1}  & \gamma\cr 
& \gamma^{\dagger}  & \rho_{2}\cr},
\label{blockro}
\end{equation}
where $\rho_{1},\ \rho_{2},\ \gamma$ are $N\times N$ matrices and 
the trace condition ${\rm Tr} (\rho_{1} + \rho_{2})=1$
is satisfied. 
The output  state of the product of two quantum operations
 $\Phi\otimes\Omega$,
can be described by: 
\begin{equation}
\big(\Phi\otimes\Omega\big)(\rho)=\bordermatrix{
&{}\cr
& P  & L\cr 
& L^{\dagger}  & Q\cr}.
\label{peelku}  
\end{equation}
Here $\Phi$ denotes an one--qubit operation, while the map $\Omega$ 
acts on $\c M_{N}$.
Also other block matrices will occur and
their positivity will play 
an important role during the proof.
The Schur complement lemma \cite{HJ} ensures 
positivity of block matrices, see Lemma \ref{trolica}, Section \ref{skrzat}.
 
To demonstrate additivity (\ref{puchatek}) we shall analyse the inequality
$M_{\Phi\otimes\Omega}\geq M_{\Phi}M_{\Omega}$ which
is almost immediate since
the equality is attained by  
a product of states which 
maximize output norm of each map. 
Because the entire set of states is larger, it
contains product and entangled states,
the result of maximizing over the entire set 
can give only a better result. 
Therefore to prove multiplicativity 
of maximal output $2$--norm with respect to
the tensor product of two maps it is enough to show that
\begin{equation}
z\geq M_{\Phi}M_{\Omega} \Rightarrow z\idty-\big(\Phi\otimes\Omega\big)(\rho)\geq 0.
\label{nier35}
\end{equation}
Insert the block matrix form (\ref{peelku}) to (\ref{nier35}).
Due to the Schur complement lemma the right hand side of (\ref{nier35})
is positive if and only if
\begin{equation}
L(z\idty-P)^{-1}L^{\dagger} \leq\  z\idty-Q.
\end{equation}
Notice that this inequality holds if 
\begin{equation}
\left\|LL^{\dagger}\right\|_2\leq(z-\left\|P\right\|_2)(z-\left\|Q\right\|_2),
\label{nierefekt}
\end{equation}
since using the general property $P\leq\left\|P\right\|$ one gets:
\begin{eqnarray}
L(z\mathbbm{1}-P)^{-1}L^{\dagger} & \leq & L(z-\left\|P\right\|_2)^{-1}L^{\dagger}\ \leq\ \left\|LL^{\dagger}\right\|_2(z-\left\|P\right\|_2)^{-1}\\ & \leq & (z-\left\|Q\right\|_2)\ \leq\  z\mathbbm{1}-Q.
\end{eqnarray}
Therefore the positivity of $(z-\left\|P\right\|_2)$ 
and $(z-\left\|Q\right\|_2)$ and inequality (\ref{nierefekt})
are the only relations needed to prove Theorem \ref{thor}. 
These relations will be proved in the next section for the case 
of $\Phi$ being an arbitrary one--qubit Davies map with $|\eta_3|\leq|\eta_1|$.

\subsubsection{Details of the proof of multiplicativity}

\begin{proof} of Theorem \ref{thor}.
It is necessary to find the specific form of $M_{\Phi},\ P$ an $Q$ in (\ref{peelku}),
then to check positivity of $(M_{\Phi}M_{\Omega}-\left\|P\right\|_2)$ 
and $(M_{\Phi}M_{\Omega}-\left\|Q\right\|_2)$, and finally to prove (\ref{nierefekt}).
Let us restrict our considerations to the case of  Davies maps $\Phi$, for which 
$\eta_3^{2}\leq\eta_1^{2}$ in (\ref{dt}) as discussed in Section \ref{puszak}.

\begin{itemize}

\item{Maximal 2--norm of the output, $M_{\Phi}$}.

Use the Bloch parametrization of $\Phi$ as in (\ref{dt}), 
let it act on the Bloch vector $(1,x,y,z)^{\dagger}$, where $x,y,z$ are real.
Moreover $x^{2}+y^{2}+z^{2}=1$ guarantees restriction to  pure states. 
It is enough to take pure input state because the 2--norm 
is convex on the set of density matrices and it
attains maximum at the boundary of the set. 
The spectral radius of the square of the output state reads according to (\ref{specrad}):
\begin{equation}
\sqrt{[r(\Phi(\rho)^{\dagger}\Phi(\rho))]}=\frac{1}{2}\big(1+\sqrt{(\kappa_3+z\eta_3)^{2}+(1-z^{2})\eta_1^{2}}\big).
\label{specrad12}
\end{equation}
Since the image of the Davies map has rotational symmetry, 
there are no parameters $x$ and $y$ 
in this formula.
Second derivative of the function (\ref{specrad12})  with respect to
 $z$ is negative under the
condition: $\eta_3^{2}\ \leq\ \eta_1^{2}$. 
Therefore function (\ref{specrad12}) has a maximum:
\begin{equation}
M_{\Phi}=\frac{1}{2}\left(1+\sqrt{\eta_1^{2}+\frac{\kappa_3^{2}\eta_1^{2}}{\eta_1^{2}-\eta_3^{2}}}\right).
\label{rex}
\end{equation}

\item{Output of a product map.}
 
Now the explicit form of matrices $P,\ Q$ and $L$ 
of the output state (\ref{peelku}) will be given. 
Consider an one--qubit input state. 
A vector $(1,x,y,z)^{\dagger}$ corresponds to the density matrix:
\begin{equation}
\rho=\frac{1}{2}\bordermatrix{
&{}\cr
& 1+z & x+iy\cr 
& x-iy & 1-z\cr}.
\label{blochro}
\end{equation}
Its image with respect to a Davies map (\ref{dt}) reads:
\begin{equation}
\Phi(\rho)=\frac{1}{2}\bordermatrix{
&{}\cr
& 1+z\eta_3+\kappa_3 & \eta_1(x+iy)\cr 
& \eta_1(x-iy) & 1-z\eta_3-\kappa_3\cr}.
\label{poopro}  
\end{equation} 
In the analogous way the initial state 
in a space $\c M_{2N}$ can be given by
(\ref{blockro}) 
\begin{equation}
\rho=\frac{1}{2}\bordermatrix{
&{}\cr
& \rho_{1}+\rho_{2}+\hat{z} & \hat{x}+i\hat{y}\cr 
& \hat{x}-i\hat{y} & \rho_{1}+\rho_{2}-\hat{z}\cr},
\end{equation}
where $\hat{z}=\rho_{1}-\rho_{2}$ and $\hat{x}-i\hat{y}=2\gamma$ are $N\times N$ matrices. 
The output state of a map $\Phi\otimes\idty$ is:
\!\!\!\!\!\!{\footnotesize\begin{eqnarray}
 \!\!\!\!\!&\ &\big(\Phi\otimes\id\big)(\rho)=\label{lola}\\
\!\!\!\!\!\!\!&\ &\bordermatrix{
&{}\cr
& \frac{1}{2}\big(\rho_{1}+\rho_{2}+\eta_3(\rho_{1}-\rho_{2})+\kappa_1(\rho_{1}+\rho_{2})\big) & \eta_1\gamma\cr 
& \eta_1\gamma^{\dagger} & \frac{1}{2}\big(\rho_{1}+\rho_{2}-\eta_3(\rho_{1}-\rho_{2})-\kappa_3(\rho_{1}+\rho_{2})\big)\cr}.\nonumber
\end{eqnarray}}
Finally the matrices $P, Q$ and $L$ are defined by
comparison of suitable blocks of two block matrices:
{\footnotesize \begin{eqnarray}
&\ &\big(\Phi\otimes\Omega\big)(\rho)=\bordermatrix{
&{}\cr
& P  & L\cr 
& L^{\dagger}  & Q\cr}\\
\!\!\!\!\!&=&\!\!\!\!\!\!\!\!\!\!\! \bordermatrix{
&{}\cr
& \frac{1}{2}\Omega\big(\rho_{1}+\rho_{2}+\eta_3(\rho_{1}-\rho_{2})+\kappa_3(\rho_{1}+\rho_{2})\big) & \eta_1\Omega\big(\gamma\big)\cr 
& \eta_1\Omega\big(\gamma\big)^{\dagger} & \frac{1}{2}\Omega\big(\rho_{1}+\rho_{2}-\eta_3(\rho_{1}-\rho_{2})-\kappa_3(\rho_{1}+\rho_{2})\big)\cr}.\nonumber
\end{eqnarray}}

\item{Multiplicativity}.

One can use the property $\left\|\Omega(\rho)\right\|_2\leq{\rm Tr}(\rho)M_{\Omega}$  (\ref{defztr}) to show that $(M_{\Phi}M_{\Omega}-\left\|P\right\|_2)$ is positive. 
It is so if
\begin{equation}
\frac{1}{2}M_{\Omega}\Big(1+\sqrt{\eta_1^{2}+\frac{\kappa_3^{2}\eta_1^{2}}{\eta_1^{2}-\eta_3^{2}}}\Big)>\frac{1}{2}M_{\Omega}\Big({\rm Tr}(\rho_{1}+\rho_{2})+\eta_3(\rho_{1}-\rho_{2})+\kappa_3(\rho_{1}+\rho_{2})\Big).
\end{equation}
Notice that ${\rm Tr}(\rho_{1}+\rho_{2})=1$. 
To show that the above inequality is true, it is sufficient to prove:
\begin{equation}
\sqrt{\eta_1^{2}+\frac{t^{2}\eta_1^{2}}{\eta_1^{2}-\eta_3^{2}}}>\eta_3+\kappa_3.
\label{koko}
\end{equation}
Taking the square of both sides one gets the expression:
\begin{equation}
\left(\kappa_3\eta_3-(\eta_1^{2}-\eta_3^{2})\right)^{2}>0. 
\end{equation}
This implies that $(M_{\Phi}M_{\Omega}-\left\|P\right\|_2)>0$.
In a similar way we prove the positivity of $(M_{\Phi}M_{\Omega}-\left\|Q\right\|_2)$.
The last step is to prove inequality (\ref{nierefekt}).
Consider a positive block matrix
$\big(\idty\otimes\Omega\big)(\rho)=\bordermatrix{         &{}\cr
& \Omega(\rho_{1})  & \Omega(\gamma)\cr 
& \Omega(\gamma^{\dagger})  & \Omega(\rho_{2})\cr}$.
Assume that $\Omega(\rho_1)>0$ (if $\Omega(\rho_1)\geq 0$
one can add $\epsilon\idty$ to $\rho_1$ and eventually 
take the limit $\epsilon\rightarrow 0$). 
Due to the inequality $\Omega(\rho_1)\leq \|\Omega(\rho_1)\|_2$ one can write
\begin{equation}
\<v|\Omega(\gamma)\Omega(\gamma)^{\dagger}|v\>\leq
\|\Omega(\rho_1)\|_2
\<v|\Omega(\gamma)\Omega(\rho_1)^{-1}\Omega(\gamma)^{\dagger}|v\>.
\end{equation}
Due to the Schur complement lemma we have $\Omega(\rho_2)\geq \Omega(\gamma)\Omega(\rho_1)^{-1}\Omega(\gamma)^{\dagger}$ and therefore,
\begin{equation}
\|\Omega(\rho_1)\|_2\<v|\Omega(\gamma)\Omega(\rho_1)^{-1}\Omega(\gamma)^{\dagger}|v\>\leq
\|\Omega(\rho_1)\|_2\<v|\Omega(\rho_2)|v\>\leq
\|\Omega(\rho_1)\|_2\|\Omega(\rho_2)\|_2.
\end{equation}
Hence the inequality $
\|\Omega(\gamma)\Omega(\gamma)^{\dagger}\|_2\leq
\|\Omega(\rho_1)\|_2\|\Omega(\rho_2)\|_2$ holds. This inequality together with definition (\ref{defztr}) implies 
\begin{equation}
\|\Omega(\gamma)\Omega(\gamma)^{\dagger}\|_2\leq M_\Omega^2\tr \rho_1\tr \rho_2.
\end{equation}
Denote ${\rm Tr}(\rho_{1})$ by $x$. 
To prove inequality (\ref{nierefekt}) it is enough to show that the second inequality
holds in the expression below
\begin{equation}
\left\|LL^{\dagger}\right\|_2=\eta_1^{2}\left\|\Omega(\gamma)\right\|_2^2
\leq\eta_1^{2}x(1-x)M_{\Omega}^{2}
\leq(M_{\Phi}M_{\Omega}-\left\|P\right\|_2)(M_{\Phi}M_{\Omega}-\left\|Q\right\|_2),
\end{equation}
and this is true if
\begin{equation}
\eta_1^{2}x(1-x)M_{\Omega}^{2}\leq\frac{1}{4}M_{\Omega}^{2}\Big[\eta_1^{2}+\frac{\kappa_3^{2}\eta_1^{2}}{\eta_1^{2}-\eta_3^{2}}-(\eta_3(2x-1)+\kappa_3)^{2}\Big].
\end{equation}
This inequality can be shown by taking the function which is 
the difference between the right hand side and the left hand side. 
The second derivative of this function is equal 
to $2(\eta_1^{2}-\eta_3^{2})$. 
Therefore whenever $(\eta_1^{2}>\eta_3^{2})$ 
the difference is a convex function which has minimum at $0$. 
That finishes the proof of the last inequality. 
Therefore inequality (\ref{nierefekt}) 
holds and it proves Theorem \ref{thor}.
\end{itemize}
\end{proof}

In the case  $|\eta_1|\leq|\eta_3|$ the proof
goes analogously. The maximal output norm (\ref{rex})
has in this case a simpler form, since 
the maximizer is a pure state described by the Bloch vector $(x=0,y=0,z=1)$. 
The specific form of the Davies map was used in this 
proof in (\ref{rex}) when the formula of the maximal output norm
was computed and in formula (\ref{lola}). 
Moreover, positivity of $\kappa_3$ is used in (\ref{koko}).




\subsection{Davies maps acting on qutrits}\label{qut}

In this chapter a characterization of the Davies maps
for qutrits, $N=3$, will be given.
Going to higher dimensions demands more abstract
and systematic approach than in the case of one--qubit maps.
The entire evolution consists of the unitary part and the
dissipative part and such is the structure of the
generator $\c G=i\delta+\c L$. The unitary evolution $\delta$ is governed
by the Hamiltonian which in its eigenbasis has a form
$H=\sum_{i=1}^3\epsilon_i|i\>\<i|$,
where $\epsilon_1>\epsilon_2>\epsilon_3$.  
Differences of energies $\{\omega_{ij}=\epsilon_i-\epsilon_j\}$ are called Bohr frequencies.
They are eigenvalues of the unitary part of the evolution,
$\delta$ given by $\rho\rightarrow[H,\rho]$, while the eigenvectors 
of $\delta$ are $|i\>\<j|$ for $i,j=1,2,3$.
Assume that the set of Bohr frequencies is not degenerated 
beside the zero frequency case, $\omega_{ii}$. The
subspace related to the zero frequency is $3$--dimensional.
Since the dissipative part $\c L$ of the evolution commutes with
the unitary part, it has the same eigenvectors
and therefore it does not couple the non-degenerated subspaces. Thus the
off diagonal entries of a density matrix are not mixed with
the diagonal ones, if the matrix is written in the eigenbasis of
the Hamiltonian.

Like in the case of one--qubit maps only the dissipative part
of the evolution will be analysed. 
An one--qutrit Davies map has a structure
{\footnotesize \begin{equation}
\Phi=\bordermatrix{         &{}\cr
& 1-\c F_{21}-\c F_{31}  & 0 & 0 & 0 & \c F_{12} & 0 & 0 & 0 & \c F_{13}\cr 
 & 0  & \mu_1 & 0 & 0 & 0 & 0 & 0 & 0 & 0\cr 
& 0  & 0 & \mu_2 & 0 & 0 & 0 & 0 & 0 & 0\cr 
 & 0  & 0 & 0 & \mu_1 & 0 & 0 & 0 & 0 & 0\cr 
& \c F_{21}  & 0 & 0 & 0 & 1-\c F_{12}-\c F_{32} & 0 & 0 & 0 & \c F_{23}\cr 
 & 0  & 0 & 0 & 0 & 0 & \mu_3 & 0 & 0 & 0\cr 
& 0  & 0 & 0 & 0 & 0 & 0 & \mu_2 & 0 & 0\cr 
 & 0  & 0 & 0 & 0 & 0 & 0 & 0 & \mu_3 & 0\cr 
& \c F_{31}  & 0 & 0 & 0 & \c F_{32} & 0 & 0 & 0 & 1-\c F_{13}-\c F_{23}\cr},
\label{cyc}
\end{equation}}
where $\c F_{21}, \c F_{31}, \c F_{32}$ and $\mu_1, \mu_2, \mu_3$
parametrize the map. The off--diagonal elements are related by the detailed balance
formula
\begin{equation}
\c F_{ij}p_j=\c F_{ji}p_i. 
\end{equation}
Here $p_i$ determine the invariant Gibbs state (\ref{knur}).
The Choi-Jamio{\l}kowski matrix of (\ref{cyc}) preserves the same structure:
  {\footnotesize  \begin{equation}
D_{\Phi}=\frac{1}{3}\bordermatrix{         &{}\cr
&1-\c F_{31}-\c F_{21}  & 0 & 0 & 0 & \mu_1 & 0 & 0 & 0 & \mu_2\cr 
& 0  & \c F_{12} & 0 & 0 & 0 & 0 & 0 & 0 & 0\cr 
& 0  & 0 & \c F_{13} & 0 & 0 & 0 & 0 & 0 & 0\cr 
& 0  & 0 & 0 & \c F_{21} & 0 & 0 & 0 & 0 & 0\cr 
& \mu_1  & 0 & 0 & 0 & 1-\c F_{32}-\c F_{21} & 0 & 0 & 0 & \mu_3\cr 
& 0  & 0 & 0 & 0 & 0 & \c F_{23} & 0 & 0 & 0\cr 
& 0  & 0 & 0 & 0 & 0 & 0 & \c F_{31} & 0 & 0\cr 
& 0  & 0 & 0 & 0 & 0 & 0 & 0 & \c F_{32} & 0\cr 
& \mu_2 & 0 & 0 & 0 & \mu_3 & 0 & 0 & 0 & 1-\c F_{32}-\c F_{31}\cr}.
\end{equation}}
The generator and its Choi-Jamio{\l}kowski matrix have also the same structure.

Block of the superoperator $\Phi$ of the Davies quantum operation
which is related to zero frequency space is a $3\times 3$ stochastic
matrix
\begin{equation}
\c F=\bordermatrix{         &{}\cr
& 1-\c F_{31}-\c F_{21}  & \c F_{12} &\c F_{13}\cr 
& \c F_{21}  & 1-\c F_{21}-\c F_{32} & \c F_{23}\cr
& \c F_{31}  & \c F_{32} & 1-\c F_{13}-\c F_{23}\cr},	
\label{gil}
\end{equation}
where 
$\c F_{32}, \c F_{31}, \c F_{21} \geq 0.
$ 
Due to the definition of the quantum detailed balance condition (\ref{detbal})
the Davies map is Hermitian with respect to scalar
product $\<X,Y\>_{\beta}:=\tr \rho_{\beta}^{-1}X^{\dagger}Y$
and therefore it has a real spectrum. Moreover, the spectrum is positive,
 since there is real logarithm of the matrix $\Phi$ represented the Davies map.
The positivity of the zero frequency block implies that 
\begin{equation}
\begin{array}{c}
\c F_{32}+\c F_{31}+\c F_{21}\leq 1, \\
3-4(\c F_{32}+\c F_{31}+\c F_{21})+3(\c F_{32}\c F_{31}+\c F_{31}\c F_{21}+\c F_{21}\c F_{32})\geq 0.
\end{array}
\label{co1}
\end{equation}
The question considered in this chapter 
concerns explicit analytical relations for entries of the
superoperator (\ref{cyc}), which imply that
the superoperator represents a Davies map.
One of the condition for $\Phi$ is that there exists an 
exponential form
\begin{equation}
\c F= e^{Lt}.
\end{equation}
Operator $L$ is the zero frequency part of the contractive part of the
 generator of completely positive Davis semigroup. 
It is parameterized as follows:
\begin{equation}
L=\bordermatrix{&{}\cr
 & -L_{21}-L_{31}  & L_{12} & L_{13}\cr 
& L_{21}  & -L_{12}-L_{32} & L_{23}\cr
& L_{31}  & L_{32} & -L_{13}-L_{23}\cr}. 
\label{co2}
\end{equation}
This is only the zero frequency block which satisfies the detailed balance condition. 
The entire dissipative part of the generator 
is represented by a $9\times 9$ matrix.
Its Choi matrix has on diagonal elements $L_{32}, L_{31}, L_{21}$.
Since the Choi state of the generator has to be 
positive on the subspace perpendicular to the maximally entangled state
we need to require that
$L_{32}, L_{31}, L_{21}\geq 0$.
 
In the next section an explicit calculation of 
the logarithm of a stochastic matrix of order three (\ref{gil})
is presented.

\subsubsection{Logarithm of a stochastic matrix of size three}

To compute analytically the logarithm of a positive matrix (\ref{gil})
one may relay on the following construction. As matrix $\c F$
has the eigenvalues $\left\{1,x+y,x-y\right\}$, where 
\begin{equation}
\begin{array}{l}
x=\frac{1}{2}(\tr \c F-1),\\ y=\frac{1}{2}\sqrt{2\tr \c F^{2} -(\tr \c F)^2 +2\tr\c F-3},
\end{array}
\label{eig}
\end{equation}
the logarithm has the form
\begin{equation}
\log (\c F)=\log\Big[U\bordermatrix{         &{}\cr
& 1  & 0 & 0\cr 
& 0  & x+y & 0\cr
& 0  & 0 & x-y\cr}
U^{-1}\Big],
\end{equation}
where $U$ is a unitary matrix which transforms $\c F$ into its diagonal form. 
Let us evaluate $\log(\c F)$ without computing the matrix $U$ explicitly. 
One can write that
\begin{equation}
\log(\c F)=\frac{1}{2}\Big[\log(x^{2}-y^{2})Z^{2}+\log(\frac{x+y}{x-y})Z\Big],
\label{log}
\end{equation}
where
\begin{equation}
Z=U\bordermatrix{         &{}\cr
& 0  & 0 & 0\cr 
& 0  & 1 & 0\cr
 & 0  & 0 & -1\cr}U^{-1}.
\end{equation}
The matrix $\c F$ can be given in terms of $Z$:
\begin{equation}
\c F=1-Z^{2}+xZ^{2}+yZ.
\end{equation}
This relation allows one to compute $yZ=(\c F-1)-(x-1)Z^{2}.$
Formula for $Z^2$ can be calculated by
taking the square of this equation and using the fact that $Z^{4}=Z^{2}$ and that $Z^{2}(\c F-1)=(\c F-1)Z^{2}=(\c F-1)$.
The last formula holds since operator $(\c F-1)$ is defined 
in the subspace for which $Z^{2}$ is the identity,
\begin{equation}
Z^{2}=\frac{(\c F-1)\left[(\c F-1)-2(x-1)\right]}{y^{2}-(x-1)^{2}}.
\end{equation}
Therefore the logarithm of the matrix $\c F$ can be expressed according to Eq. (\ref{log}). 
By comparing a suitable entries of $\log(\c F)$ with the parameters of $L$,
 one gets the parameter $L_{21}$ as a function of $(\c F_{32},\c F_{31},\c F_{21})$,
\begin{equation}
 L_{21}= y_2 (1 - x - y) \log(x - y)-(y_1 (1 - x + y) \log(x + y)),\label{cece}
%
\end{equation}
where $x$ and $y$ are given by Eq. (\ref{eig}) and
\begin{eqnarray}
 y_1 &:=& 2y - \c F_{12} - \c F_{21} + \c F_{13} - \c F_{31} + \c F_{23} - \c F_{32} + 2\frac{\c F_{23}\c F_{31}}{\c F_{21}},\\
y_2 &:=& 4 - 2y - \c F_{12} - \c F_{21} + \c F_{13} - \c F_{31} + \c F_{23} - \c F_{32}+2 \frac{\c F_{23}\c F_{31}}{\c F_{21}}.
\end{eqnarray}

The set of points $\left\{\c F_{23},\c F_{13},\c F_{12}\right\}$, 
which defines the set of symmetric bistochastic matrices from the dynamical semigroup, 
is shown in Fig. \ref{fig:mata} and Fig. \ref{fig:mata23} denoted by $\s E$. This set is 
inside the set of all bistochastic $3\times 3$ matrices which is denoted by $\s D$. 
The boundaries of the set are stated by the constraints $L_{21},L_{31},L_{32}\geq 0$. 

\begin{figure}[htbp]
	\centering
	\scalebox{.7}{
		\includegraphics{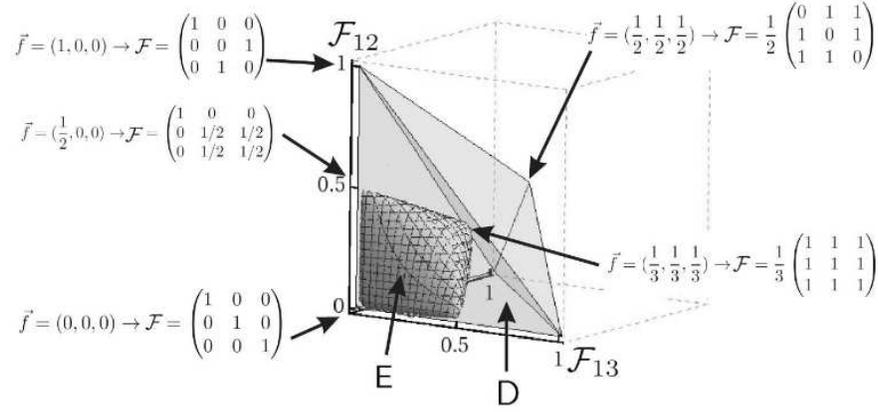}}
	\caption{The set $\s E$ of $3\times3$ bistochastic matrices $\c F$ (\ref{gil}),
which form the zero frequency block 
of the Davies channel $\Phi$ (\ref{cyc}) under condition $T\rightarrow \infty$,
is represented by the vector of the off--diagonal elements $\vec{f}=\{\c F_{12}, \c F_{13}, \c F_{23}\}$.
The set $\s E$ is inside the set of all bistochastic $3\times 3$ matrices $\s D$. 
The characteristic points are denoted by  $\vec{f}$ and the corresponding matrix $\c F$.
}
\label{fig:mata}
\end{figure}

Expression (\ref{cece}) allows one to check that the set
of stochastic matrices belonging to the semigroup of completely positive maps
 with the detailed balance condition
 is not convex. 
Consider two exemplary points which lie near the border of the 
cross-section and belong to the set ($L_{21}\geq0$): $\left\{0.5,0,0\right\}$ and $\left\{0.22744,0.22744,0.04512\right\}$. 
Their convex combination does not belong to the set.
Therefore the set of Davies map is not convex.
Fig. \ref{fig:mata23} presents the cross-section of the set of bistochastic matrices which
form the zero frequency part of the Davies map represented in the space of
parameters
$\c F_{21}, \c F_{31}, \c F_{32}$.
Fig. \ref{fig:mata23}{\it b} plots a non--convex cross--section of the set $\s E$ by the plane $\s M$.

\begin{figure}[htbp]
	\centering
		\scalebox{0.7}{
		\includegraphics{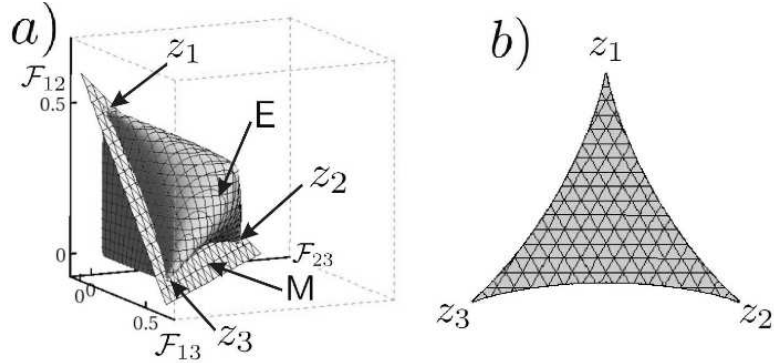}}
\caption{ $a)$ The set $\s E$ of $3\times3$ bistochastic matrices $\c F$ (\ref{gil}),
which form the zero frequency block 
of the Davies channel $\Phi$ (\ref{cyc}) under condition $T\rightarrow \infty$,
is represented by the vector of the off--diagonal elements $\{\c F_{12}, \c F_{13}, \c F_{23}\}$.
The set $\s E$ is cut by the plane $\s M: \c F_{12}+\c F_{13}+\c F_{23}=\frac{1}{2}$.
The cross--section is presented in Panel $b)$ and shows that the set $\s M$ is not convex.
}
\label{fig:mata23}
\end{figure}



In order to obtain a
full characterization of 
the Davies map for qutrits, not only its zero frequency par have to be analysed.
One needs to take into consideration also the
complete positivity condition and the condition
on the semigroup related to the Choi--Jamio{\l}kowski matrices of the superoperator
and its generator. These conditions allow us to specify 
the matrix entries $\mu_i$ from Eq. (\ref{cyc}).

In this way the full characterization of
the Davies channels for one--qubit and one--qutrits is provided.

\section{Concluding remarks and open problems}

The aim of this thesis was to investigate quantum channels 
on different levels of generality and using different approaches. 
For instance, general properties of quantum channels were considered in Chapter \ref{sec:prl},
while some particular classes of 
one--qubit and one--qutrit quantum channels were analysed in Part \ref{partthree} of these thesis. The
Davies maps motivated by a specific physical model were studied in Chapter \ref{bida}.
Some useful characteristics of a quantum channel are provided by different kinds of 
entropies. Among them we used the minimal output entropy, the entropy of a map, 
the entropy of an environment which takes part in an evolution described by a channel. 
Apart of the standard von Neumann entropy which is the quantum counterpart of 
the Shannon entropy, the quantum R\'{e}nyi and Tsallis entropies
were also applied.

In Part \ref{parttwo} 
the universal entropic inequality for an arbitrary ensemble of quantum states is proved for the von Neumann entropy.
This part of the thesis treats a quantum channel as a device  
preparing a quantum ensemble. 
The Holevo quantity
of this ensemble is shown to be bounded by the entropy of an  
environment, used in the preparation process.
The state of the environment after a quantum operation $\Phi$
is equivalent to the output of the complementary channel $\tilde{\Phi}$.

One can define {\it selfcomplementary} channels 
for which $\Phi(\rho)=\tilde{\Phi}(\rho)$ for any $\rho$. 
Relation (\ref{pantalon}) between the Kraus operators of $\Phi$
and the Kraus operators of $\tilde{\Phi}$ is useful to specify selfcomplementary channels.
Since the coherent information (\ref{zwierzyna}) of such channels is
equal to $0$, the same holds also for the quantum channel capacity (\ref{katemidlton}).
Identification of selfcomplementary channels, as well as 
investigation of their properties are worth to
be studied in future.

Chapter \ref{dusiu} contains the conjecture which 
establishes a relation between the Holevo quantity,
and the matrix of fidelities. This leads to a
geometric characterization of the states in the ensemble.
The bound on the Holevo quantity proved in Chapter \ref{sec:prl}
can be also related to other notions 
of quantum information theory, such as the {\it quantum discord} \cite{coles}, \cite{ollivier} 
which measures the quantum correlations in a two--partite system.

The study of quantum channels is an important task of the modern
theory of quantum information. 
For example, the problem of additivity of the channel capacity, 
or equivalently, additivity of the minimal output entropy
remains open even for channels acting on a single qubit.
Results presented in this thesis could be 
further developed to investigate the additivity conjecture
for different classes of quantum channels.

Some results of Chapter \ref{sec:ent} concern general properties
of quantum channels.
For instance we proved the additivity of the map entropy (\ref{ludwiczek}), and
Theorem \ref{glodu} establishing the
extremal position of depolarizing channels
in the set of all channels characterized by the R\'{e}nyi entropies $S_2^{\min}(\Phi)$ and $S_2^{\map}(\Phi)$. 
These results allow us to pose Conjecture \ref{conjadd} specifying 
pairs of maps for which the additivity of channel capacity may hold.

In Part \ref{partthree}, some specific types of channels
are investigated. Properties of one--qubit channels are
analysed in Chapter \ref{sec:ent}.
Some transformations on one--qubit quantum channels
 defined in Section \ref{mlotek} lead to new results 
on the characterization of the set of quantum channels in the plane
$(S^{\min}(\Phi), S^{\map}(\Phi))$.
The aim of this analysis is to find some conditions
that enable one to estimate the minimal output entropy,
which is difficult to compute, by the entropy of the map easy to calculate.

The Davies channels, which correspond to 
a concrete physical model, are studied in Chapter \ref{bida}. 
Superoperators of the Davies maps are specified in the case of 
one--qubit maps and one--qutrit maps.
The question whether the channel capacity of the Davies maps
is additive is still open, although, Davies
maps acting on $N$--level system 
compose the set of only $d=N^2-1$ dimensions, while the set
of all quantum operations acting on system of the same $N$
has $N^2(N^2-1)$ dimensions.

The quantum information theory is a modern
field of science which  
creates an environment for new future applications and
opens new paths for 
development of technology.
Quantum channels, which describe any possible evolution
of a quantum state, play an important role in possible applications.
Quantum channels describe decoherence caused by the interaction with an environment.
Knowledge of their properties allows one to choose 
the most efficient quantum protocols for a given purpose. 
Theoretical investigations
uncover new possibilities, new laws and fundamental restrictions on
processing of quantum information. 

The classical theory of information began with investigations on communication
in a given language through given technological tools. 
However, very fast, the laws of information became treated
as fundamental properties of nature.  
Therefore, studies in the field of quantum information are so exciting.

%
\newpage
{\Large\textbf{Appendix 1}}

\medskip\medskip\medskip
In Appendix we analyze ensembles of three
one--qubit states $\{\rho_1=|\phi_1\>\<\phi_1|,\rho_2=|\phi_2\>\<\phi_2|,\rho_3\}$
and provide calculations related to Fig. \ref{fig:para} necessary to 
prove Lemma \ref{zdzichzdzich} in Section \ref{sadelko}.

The Bloch vector characterizing the average states can be given by
\begin{equation}
\vec{OA}=a(0,0,1).
\end{equation}
The Bloch vector representing the mixed state $\rho_3$ is
 parameterized by an angle $\alpha$ 
\begin{equation}
\vec{OB}=b(0,\sin{\alpha},\cos{\alpha}).
\end{equation}
The vector $\vec{OC}$ is chosen in such a way that the ratio $|CA|:|AB|$ is $1:2$. Therefore one has
\begin{equation}
\vec{OC}=\frac{1}{2}(3\vec{OA}-\vec{OB}).
\end{equation}
The point $C$ is in the center of the interval $DE$, between two pure states $|\phi_1\>$ and $|\phi_2\>$  
characterized 
by the points $D$ and $E$. Both vectors $\vec{OD}$ 
and $\vec{OE}$ form with vector $\vec{OC}$ the angle $\gamma$ so that 
\begin{equation}
 \cos{\gamma}=|OC|.
\end{equation}
This is in turn the square root of the fidelity $|\<\phi_1|\phi_2\>|^2$,
because the angle $\gamma$ is half of the angle between two pure states,
\begin{equation}
F_{23}=\cos^2{\gamma}.
\end{equation}
The fidelity between two one--qubit states represented by Bloch vectors $\vec{x}$ and $\vec{y}$ reads
\begin{equation}
F=\frac{1}{2}(1+\vec{x}\cdot\vec{y}).
\label{ppmqufid}
\end{equation}
The scalar product of $\vec{OB}$ and $\vec{OD}$ is equal to:
\begin{equation}
\vec{OB}\cdot\vec{OD}=b\cos{(\mu+\gamma-\gamma)}=b\Big[\cos{(\mu+\gamma)}\cos{\gamma}+\sin{(\mu+\gamma)}\sin{\gamma}\Big].
\end{equation}
Hence
\begin{equation}
F_{12}=\frac{1}{2}(1+\vec{OC}\cdot\vec{OB}+b\sqrt{1-\frac{(\vec{OC}\cdot\vec{OB})^2}{b^2\vec{OC}\cdot\vec{OC}}}\sqrt{1-\vec{OC}\cdot\vec{OC}}).
\end{equation}
The third fidelity $F_{13}$ can by obtained using Lemma \ref{ppmlem1}. 
For $\beta=0$ the product of three fidelities used in Lemma \ref{zdzichzdzich} 
is a function $f_0(a,b,\alpha,\beta=0)$ given by
\begin{eqnarray}
&f_0(a,b,\alpha,\beta=0)=F_{12}F_{13}F_{23}=\\
&\frac{1}{64} \left(9 a^2-6 b \cos\alpha
   a+b^2\right) \left(b^2-3 a \cos\alpha
   b-2\right)^2\\
&+\frac{1}{64} \left(9 a^2 b^2 \left(9 a^2-6 b \cos\alpha a+b^2-4\right) \sin^2\alpha\right).
\label{ppmpro}
\end{eqnarray}

\quad \\



\newpage
{\Large\textbf{Appendix 2}}

\medskip\medskip\medskip

In this appendix we present computations necessary to prove
Lemma \ref{ppmlem2}. 
It is convenient to change the basis such that the vector $\vec{OB}$ (see Fig. \ref{fig:para}) is transformed into
\begin{equation}
\vec{OB}'=b(0,0,1).
\end{equation}
Denote the angle $\nu:=\mu+\gamma$, where $\mu$ is the angle between $\vec{OB}$ and $\vec{OD}$. 
The vectors $\vec{OD}$ and $\vec{OE}$ 
in the new basis can be obtained by rotating the state $(0,0,1)$ 
around the axis $x$ by angles:
\begin{eqnarray}
\label{ppmvec}
\vec{OD}'&=&U_x(\mu)(0,0,1),\\
\vec{OE}'&=&U_x(\mu+2\gamma)(0,0,1).
\end{eqnarray}
The vectors $\vec{OG}'$ and $\vec{OF}'$ are obtained by rotating 
the above vectors around the axis $U_x(\nu)(0,0,1)$ by angle $\beta$. 
Such a rotation can be defined as an action of a unitary matrix $U$ on vectors (\ref{ppmvec}). 
The unitary matrix is given by
\begin{equation}
U=U_z(-\frac{\pi}{2})\,U_y(\nu)\,U_z(\beta)\,U^{\dagger}_y(\nu)\,U^{\dagger}_z(-\frac{\pi}{2}),
\end{equation}
where the rotation matrices read
\begin{eqnarray*}
&\ &U_x(\alpha)=
\begin{pmatrix}
1 & 0 & 0\\
0 & \cos{\alpha} & -\sin{\alpha}\\
0 & \sin{\alpha} & \cos{\alpha}
\end{pmatrix},\quad
U_y(\alpha)=
\begin{pmatrix}
\cos{\alpha} & 0 & \sin{\alpha}\\
0 & 1 & 0\\
 -\sin{\alpha} & 0 & \cos{\alpha}
\end{pmatrix},\\
&\ &U_z(\alpha)=
\begin{pmatrix}
\cos{\alpha} & -\sin{\alpha} & 0\\
\sin{\alpha} & \cos{\alpha} & 0\\
0 & 0 & 1
\end{pmatrix}.
\end{eqnarray*}

One can use formula (\ref{ppmqufid}) to calculate the product of three fidelities for 
three considered states, $\vec{OB}', \vec{OG}'$ and $\vec{OF}'$ as a function of the angle $\beta$. 
\begin{eqnarray}
 f&=&F_{12}F_{13}F_{23}\\
&=&\!\!\!\!\!\!\frac{1}{16} \cos^2\gamma \left((\cos\mu+\cos
   (2 \gamma +\mu )+2)^2-\cos ^2\beta (\cos\mu-\cos (2 \gamma +\mu ))^2\right).\nonumber
\end{eqnarray}
The product of three pairwise fidelities attains its minimum at $\beta=0$ 
as stated in Lemma \ref{ppmlem2}.


\newpage 









\begin{thebibliography}{99}  

\bibitem{Shannon} 
C. Shannon,
{\it A Mathematical Theory of Communication},
The Bell System Technical Journal, \textbf{27}  379--423, 623--656 (1948).

\bibitem{Renyi102}
A. R\'{e}nyi, 
{\it On measures of information and entropy}, 
Proceedings of the 4th Berkeley Symposium on Mathematics, Statistics and Probability, 547--561 (1960). 


\bibitem{tsallis102}
C. Tsallis, 
{\it Possible generalization of Boltzmann-Gibbs statistics}, 
J. Stat. Phys., \textbf{52} 479-487 (1988).

\bibitem{tsallis103}
A. Plastino, A. R. Plastino,
{\it Tsallis Entropy and Jaynes' Information Theory Formalism},
Brazilian Journal of Physics, \textbf{29} 50-60 (1999).

\bibitem{daviesstoch}
E. Davies,
{\it Quantum stochastic processes}, 
Commun. Math. Phys.,    \textbf{15} 277--306 (1970).

\bibitem{kossakowski}
A. Kossakowski, 
{\it On quantum statistical mechanics of non-Hamiltonian systems}, 
Rep. Math. Phys., \textbf{3} 247--274 (1972).

\bibitem{holevo}
A. Holevo,
{\it Bounds for the quantity of information transmitted by a quantum communication channel},
\emph{Prob.\ Inf.\ Transm.\ } \textbf{9} 177--83 (1973).

\bibitem{nielsen}
M. Nielsen, I. Chuang,
\emph{Quantum Computation and Quantum Information},
Cambridge University Press, Cambridge (2000).



\bibitem{bennett}
C.H. Bennett, G. Brassard, C. Cr\'{e}peau, R. Jozsa, A. Peres, and W.K. Wootters, 
{\it Teleporting an unknown quantum state via dual classical and Einstein-Podolsky-Rosen channels}, 
Phys. Rev. Lett.,  \textbf{70} 1895--1899  (1993).


\bibitem{deutsh}
D. Deutsch, R. Jozsa, 
{\it Rapid solutions of problems by quantum computation}, 
Proceedings of the Royal Society of London A, \textbf{439} 553--558 (1992).

\bibitem{grover}
Grover L.K.
{\it A fast quantum mechanical algorithm for database search}, 
Proceedings, 28th Annual ACM Symposium on the Theory of Computing, 212--219 (1996).

\bibitem{shorfactor}
P. Shor,
{\it Polynomial-Time Algorithms for Prime Factorization and Discrete Logarithms on a Quantum Computer},
SIAM J.Sci.Statist.Comput.,  \textbf{26} 1484--1509 (1997).

\bibitem{hortele}
M. Horodecki, P. Horodecki, and R. Horodecki,
{\it General teleportation channel, singlet fraction, and quasidistillation}, 
Phys. Rev. A,  \textbf{60} 1888--1898 (1999).

\bibitem{horo4}
R. Horodecki, P. Horodecki, M. Horodecki, K. Horodecki,
{\it Quantum entanglement},
Rev. Mod. Phys.,  \textbf{81} 865--942 (2009).

\bibitem{KZ} 
I.~Bengtsson and K.~{\.Z}yczkowski,
\emph{Geometry of Quantum States: An Introduction to Quantum Entanglement}, 
Cambridge University Press, Cambridge (2006)


\bibitem{hastings}
M. Hastings, 
{\it Superadditivity of communication capacity using entangled inputs}, 
Nature Physics, \textbf{5}  255--257   (2009).

\bibitem{horo}
M. Horodecki, 
{\it On Hastings counterexamples to the minimum output entropy additivity conjecture}, 
Open Systems $\&$ Information Dynamics,    \textbf{17} 31--52 (2010) .


\bibitem{kingostatni}
C. King,
{\it Remarks on the Additivity Conjectures for Quantum Channels},
in: Entropy and the Quantum, eds. R. Sims, D. Ueltschi, Contemporary Mathematics, \textbf{529} 177--188,
University of Arizona (2010).


\bibitem{bennettbrassard}
C. Bennett, G. Brassard,
{\it Quantum Cryptography: Public Key Distribution and Coin Tossing}, 
Proceedings of IEEE International Conference on Computers Systems and Signal Processing, Bangalore India, 175--179 (1984).


\bibitem{wehrl}
A. Wehrl, 
{\it General properties of entropy},
Rev. Mod. Phys.,   \textbf{50} 221--260 (1978). 

\bibitem{ruskaijakies}
M. B. Ruskai,
{\it Inequalities for Quantum Entropy: A Review with Conditions for Equality},
J. Math. Phys.,  \textbf{43} 4358--4375 (2002).

\bibitem{Renyitsallis}
F. Nielsen, R. Nock,
{\it On R\'{e}nyi and Tsallis entropies and divergences for exponential families}, 
arXiv:1105.3259, (2011). 

\bibitem{shannon2}
C. Shannon, 
{\it Communication in the presence of noise}, 
Proc. Institute of Radio Engineers, \textbf{37} 10--21. (1949). 

\bibitem{hartley}
R. Hartley,
{\it Transmission of Information}, 
Bell System Technical Journal, \textbf{7} 535--563 (1928).


\bibitem{schumacher}
B. Schumacher, 
{\it Quantum coding}, 
Phys. Rev. A, \textbf{51} 2738--2747 (1995).

\bibitem{desurvire}
E. Desurvire, 
{\it Classical and Quantum Information Theory. An Introduction for the Telecommunication Scientists},
Cambridge University Press, Cambridge (2009).


\bibitem{holevocapa}
A. Holevo, 
{\it The capacity of quantum channel with general signal states},
IEEE Trans. Info. Theory, \textbf{44} 269--273 (1998).


\bibitem{schumwest}
B. W. Schumacher and M. Westmoreland, 
{\it Sending classical information via noisy quantum channels},
Physical Review A, \textbf{56} 131--138 (1997).

\bibitem{szuma}
B. Schumacher, 
{\it Sending entanglement through noisy quantum channels},
Phys. Rev. A, \textbf{54} 2614--2628 (1996).


\bibitem{fuchs}
C. Fuchs, 
{\it Nonorthogonal quantum states maximize classical information capacity},
Phys. Rev. Lett.,  \textbf{79} 1162--1165 (1997).


\bibitem{hayashi}
M. Hayashi, H. Imai, K. Matsumoto, M. B. Ruskai, T. Shimono,
{\it Qubit channels which require four inputs to achieve capacity: Implication for
additivity conjectures}, 
Quantum Inf. Comput., \textbf{5} 13--31 (2005).



\bibitem{knill}
E. Knill, R. Laflamme, A. Ashikhmin, H. Barnum, L. Viola and W. H. Zurek,
{\it Introduction to Quantum Error Correction}, 
arXiv:quant-ph/0207170 (2002).


\bibitem{choi}
M.-D. Choi. 
{\it Completely positive linear maps on complex matrices}, 
Linear Algebra and Its Applications,  \textbf{10} 285--290 (1975).

\bibitem{jam} 
A. Jamio{\l}kowski, 
{\it Linear transformations which preserve trace and positive semidefiniteness of operators}, 
\emph{Rep.\ Math.\ Phys.}\ \textbf{3} 275 (1972).


\bibitem{kraus2}
K. Kraus. 
{\it States, Effects and Operations: Fundamental Notions of Quantum Theory}, 
Springer-Verlag, Berlin (1983).


\bibitem{depillis}
J. De Pillis, 
{\it Linear transformations which preserve hermitian and positive semidefinite operators},
Pacific J. Math.,  \textbf{23} 129--137 (1967).





\bibitem{fujivara}
A. Fujiwara and P. Algoet, 
{\it Affine parametrization of quantum channels}. 
Phys. Rev. A,  \textbf{59} 3290--3294  (1999).

\bibitem{ruskaiszarek}
M. B. Ruskai, S. Szarek, E. Werner,
{\it An analysis of completely-positive trace-preserving maps on $\c M_2$},
Linear Algebra and its Applications, \textbf{347} 159--187 (2002). 

\bibitem{shor}
P. Shor, 
{\it Equivalence of additivity questions in quantum information theory}, 
Commun. Math. Phys., \textbf{246}  453--472 (2004).


\bibitem{amosov}
G. Amosov, A. Holevo, and R. Werner, 
{\it Additivity/multiplicativity problems for quantum communication channels}, 
Quantum Communication, Computing, and Measurement,   \textbf{3} 3--10 (2001).


\bibitem{king2}
C. King, 
{\it The capacity of the quantum depolarizing channel}, 
IEEE Transactions on Information Theory,  \textbf{49}  221--229 (2003).

\bibitem{ruskai} 
T. S. Cubitt, M B. Ruskai, G. Smith,
{\it The structure of degradable quantum channels},
\emph{J. Math. Phys.} \textbf{49} 102104 (27 pp) (2008).


\bibitem{lindbl}
G. Lindblad,
{\it On the generators of quantum dynamical semigroups}, 
Commun. Math. Phys.,    \textbf{48} 119--130 (1976).


\bibitem{kraus}
K.~Kraus,
{\it General state changes in quantum theory},
\emph{Ann.\ Phys.\ }, \textbf{64} 311--35 (1971).

\bibitem{breuer} 
H.-P. Breuer, F. Petruccione, 
\emph{The theory of open quantum systems}, 
Clarendon Press, Oxford (2006).

\bibitem{roga0}
W. Roga, M. Fannes, K. \.Zyczkowski, 
{\it Composition of quantum states and dynamical subadditivity},
J. Phys. A. Math. Theor.,   \textbf{41} 035305 (15 pp) (2008).


\bibitem{verstaete}
F. Verstraete, H. Verschelde, 
{\it On quantum channels}, 
arXiv:quant-ph/0202124 (2002).


\bibitem{ziman}
M. Ziman, 
{\it Incomplete quantum process tomography and principle of maximal entropy},
Phys. Rev. A, \textbf{78} 032118 (8 pp) (2008)


\bibitem{roga}
W. Roga, M. Fannes, K. \.Zyczkowski, 
{\it Universal Bounds for the Holevo Quantity, Coherent Information, and the Jensen-Shannon Divergence}, 
\emph{Phys. Rev. Lett.}, \textbf{105} 040505 (4 pp) (2010).  

\bibitem{roga2}
W. Roga, M. Fannes, K. \.Zyczkowski,
{\it Davies maps for qubit and qutrits},
Rep. Math. Phys.,   \textbf{66} 311--329  (2010).

\bibitem{roga3}
W. Roga, M. Fannes, K. \.Zyczkowski,
{\it Entropic characterization of quantum operations},
International Journal of Quantum Information,    \textbf{9} 1031--1045 (2011).

\bibitem{roga5}
M. Fannes, F. de Melo, W. Roga, K. \.Zyczkowski
{\it Matrices of fidelities for ensembles of quantum states and the Holevo quantity},
arXiv/quant-ph:1104.2271 (2011).

\bibitem{roga4}
W. Roga, M. Smaczy\'nski, K. \.Zyczkowski,
{\it Composition of Quantum Operations and Products of Random Matrices},
Acta Physica Polonica B,  \textbf{42} 1123 (18 pp) (2011).

\bibitem{king}
 C. King and M. B. Ruskai, 
{\it Minimal Entropy of States Emerging from Noisy Quantum Channels}, 
IEEE Trans. Info. Theory.,    \textbf{47} 192--209 (2001).


\bibitem{agarval}
M. Agrawal, 
{\it Axiomatic / Postulatory Quantum Mechanics, in Fundamental Physics in Nano-Structured}, 
Stanford University (2008).


\bibitem{hong}
C. Hong, Z. Ou, L. Mandel,
{\it Measurement of subpicosecond time intervals between 2 photons by interference}, 
Phys. Rev. Lett.,   \textbf{59} 2044--2046 (1987).

\bibitem{daviesmierzy}
E. B. Davies, 
{\it Quantum Theory of Open Systems}, 
Academic Press, London (1976).

\bibitem{schmidt}
E. Schmidt, 
{\it Zur Theorie der linearen und nicht linearen Integralgleichungen},
Math. Ann,   \textbf{63} 433--466 (1907).

\bibitem{vonneumann}
J. von Neumann,  
{\it Mathematische Grundlagen der Quantenmechanik (Mathematical Foundations of Quantum Mechanics)}, 
 Springer, Berlin (1955).



\bibitem{liebruskai}
E. H. Lieb and M.B. Ruskai, 
{\it Proof of the strong subadditivity of quantum mechanical entropy}, 
J. Math. Phys.,   \textbf{14} 1938--1941 (1973).

\bibitem{sting}
W. Stinespring, 
{\it Positive Functions on C*-algebras}, 
Proc. Amer. Math. Soc.,   \textbf{6} 211--216, (1955).


\bibitem{uhlmann76}
A. Uhlmann, 
{\it The ``transition probability`` in the state space of a *-algebra}, 
Rep. Math. Phys.,      \textbf{9} 273--279 (1976).


\bibitem{kokk}
S. Kokkendorf,
{\it Gram matrix analysis of finite distance spaces in constant curvature},
Discrete Comput. Geom.,   \textbf{31} 515--543 (2004).


\bibitem{bahr}
B. Bahr, B. Dittrich,
{\it Regge calculus from a new angle},
New Journal of Physics,   \textbf{12} 033010, (10 pp) (2010).

\bibitem{AlickiFannes} 
R. Alicki, M. Fannes,
{\it Quantum dynamical systems},
Oxford University Press (2001).

\bibitem{szli}
R. Jozsa, J. Schlienz,
{\it Distinguishability of states and von Neumann entropy}, 
Phys. Rev. A,  \textbf{62} 012301  (11 pp) (2000).

\bibitem{mitch}
G. Mitchison, R. Jozsa,
{\it Towards a geometrical interpretation of quantum--information compression},
Phys. Rev. A,  \textbf{69} 032304 (6 pp) (2004). 


\bibitem{wooters}
W. Wootters,
{\it Statistical distance and Hilbert space},
Phys. Rev. D.,  \textbf{23} 357--362 (1981).

\bibitem{topsoe}
B. Fuglede, F. Tops\o{}e,
{\it Jensen-Shannon divergence and Hilbert space embedding},
IEEE International Symposium on Iinformation Theory, Proceedings,  31--31 (2004). 

\bibitem{topsss}
F. Tops\o{}e, 
{\it Some inequalities for information divergence and related measures of discrimination}, 
IEEE Trans. Inform. Theory,   \textbf{46} 1602--1609 (2000).


\bibitem{holsir}
A. Holevo, M. Sirokov, 
{\it  Mutual and coherent information for infinite-dimensional quantum channels},
Problems of information transmission,  \textbf{46} 201--218 (2010). 

\bibitem{private}
K. Horodecki, M. Horodecki, P. Horodecki, J. Oppenheim,
{\it General paradigm for distilling classical key from quantum states},
IEEE Transactions on Information Theory,  \textbf{55} 1898--1929 (2009). 


\bibitem{szum}
B. W. Schumacher, M. A. Nielsen, 
{\it Quantum data processing and error correction},
Phys. Rev. A,   \textbf{54} 2629--2635 (1996).

\bibitem{lloyd}
S. Lloyd, 
{\it Capacity of the noisy quantum channel}, 
Phys. Rev. A,  \textbf{55} 1613--1622 (1997).

\bibitem{lindblad}
G. Lindblad, 
\emph{Quantum entropy and quantum measurements}, 
in Quantum Aspects of Optical Communication, eds. C. Bendjaballah et al., 
Lecture Notes in Physics,  \textbf{378} 71--80, Springer-Verlag, Berlin (1991).

\bibitem{arakilieb}
H. Araki, E. Lieb,
{\it Entropy inequalities},
Comm. Math. Phys. \textbf{18} 160--170 (1970).

\bibitem{blady}
G. Lindblad, 
{\it An entropy inequality for quantum measurements}, 
Commun. Math. Phys.,  \textbf{28} 245--249 (1972).


\bibitem{uhlma}
A. Uhlmann, 
{\it Relative entropy and the Wigner-Yanase-Dyson-Lieb concavity in an interpolation theory}, 
Commun. Math. Phys.,  \textbf{54} 21--32 (1977).


\bibitem{monotsallis}
S. Furuichi, K. Yanagi, and K. Kuriyama
{\it Fundamental properties of Tsallis relative entropy},
J. Math. Phys.,  \textbf{45} 4868 (10 pp) (2004).


\bibitem{monoRenyi}
F. Hiai, M. Mosonyi, D. Petz,
{\it Monotonicity of f-divergences: A review with new results},
arXiv/math-phys: 1008.2529 (2008).


\bibitem{stephanie}
R. Konig, S. Wehner,
{\it A Strong Converse for Classical Channel Coding Using Entangled Inputs},
Phys. Rev. Lett.,  \textbf{103} 070504 (4 pp) (2009).


\bibitem{briet}
J. Bri\"et, P. Herremo{\"e}s
{\it Properties of classical and quantum Jensen-Shannon divergence},
\emph{Phys.\ Rev.\ A}, \textbf{79} 052311 (11 pp) (2009).


\bibitem{endres}
D. M Endres, J. E. Schindelin, 
{\it A new metric for probability distributions}, 
IEEE Trans. Inf. Theory,  \textbf{49} 1858--1860 (2003).


\bibitem{fanesijegostudent}
M. Fannes and D. Vanpeteghem, 
{\it A three state invariant},
arXiv:quant-ph/0402045 (2002).



\bibitem{positive}  
R. Bhatia, 
{\it Positive Definite Matrices},
Princeton University Press, Princeton (2007).

\bibitem{zyczkowskisommers}
K. \.Zyczkowski, H.-J. Sommers,
{\it Hilbert--Schmidt volume of the set of mixed quantum states},
J. Phys. A, \textbf{36} 10115--10130 (2003).

\bibitem{king1} 
C. King,
{\it Additivity for unital qubit channels},
J. Math. Phys.,  \textbf{43} 4641 (13 pages) (2002).

\bibitem{shor2}
P. Shor, 
{\it Additivity of the classical capacity of entanglement-breaking quantum channels},
J. Math. Phys.,  \textbf{43} 4334--4340, (2003).

\bibitem{diVincenzo}
D. DiVincenzo, P. Shor, J. Smolin, 
{\it Quantum-channel capacity of very noisy channels},
Phys. Rev. A, \textbf{57} 830--839 (1998).

\bibitem{holrev}
A. Holevo,
{\it The additivity problem in quantum information theory},
Russian Mathematical Surveys,   \textbf{61} 301--339 (2006).

\bibitem{alickii}
R. Alicki, K.  Lendi, 
{\it Quantum dynamical semigroups and applications},
 Springer-Verlag, Berlin  (1987).

\bibitem{davies}
E. B. Davies, 
{\it Markovian master equations}, 
Commun. Math. Phys.,   \textbf{39} 91--110 (1974).

\bibitem{gorini}
V. Gorrini, A. Kossakowski and E. Sudarshan, 
{\it Completely positive dynamical semigroups of n-level systems}, 
J. Math. Phys.,  \textbf{17} 821--825 (1976).


\bibitem{wolfa2}
T. Cubitt, J. Eisert, M. Wolf, 
{\it Deciding whether a Quantum Channel is Markovian is NP-hard}, 
arXiv/math-phys:0908.2128v1 (2009).


\bibitem{wolfa1}
M. Wolf, J. Eisert, T. S. Cubitt and J. I. Cirac, 
{\it Assessing non-Markovian dynamics}, 
Phys. Rev. Lett.,   \textbf{101} 150402 (4 pp)  (2008).

\bibitem{db1}
G. S. Agarwal,
 {\it Open quantum Markovian systems and the microreversibility}, 
Z. Phys.,  \textbf{258} 409--422 (1973).

 \bibitem{kimura} 
G. Kimura, 
{\it Restriction on relaxation times derived from the Lindblad-type master equations for two-level    systems}, 
Phys. Rev. A,  \textbf{66} 062113 (4 pp)  (2002).

\bibitem{wodkiewicz}
S. Daffer, K. W\'{o}dkiewicz and J. K. McIver: 
{\it Quantum Markov channels for qubits}, 
Phys. Rev.  A,  \textbf{67} 062312 (13 pp) (2003).

\bibitem{koldman}
C. King, N. Koldan
{\it New multiplicativity results for qubit maps},
J. Math. Phys.,  \textbf{47} 042106 (9 pp) (2006).

\bibitem{ORT} 
J. M.~Ortega,
\emph{Matrix Theory, A Second Course}, 
Plenum Press, New Yourk (1987).

\bibitem{HJ} 
R. A.~Horn, C. R. Johnson,
\emph{Matrix Analysis}, 
Cambridge University Press, Cambridge (1985).



\bibitem{coles}
P. Coles,
{\it Non-negative discord strengthens the subadditivity of quantum entropy functions},
arXiv:1101.1717 (2011).

\bibitem{ollivier}
H. Ollivier, W. Zurek,
{\it Quantum Discord: A Measure of the Quantumness of Correlations},
Phys. Rev. Lett.,  \textbf{88} 017901 (4 pp) (2002).






\bibitem{Coles10} 
P. Coles, L. Yu, V. Gheorghiu, R. Griffiths, 
{\it Information-theoretic treatment of tripartite systems and quantum channels}, 
Phys. Rev. A,  \textbf{83} 062338 (2011).




 


\bibitem{lamberti}
P. W. Lamberti, M. Portesi, J. Sparacino,
{\it Natural metric for quantum information theory},
International Journal of Quantum  Information, \textbf{7} 1009--1019 (2009). 



\bibitem{belavkin} 
V. Belavkin, 
{\it Contravariant densities, complete distances and relative fidelities for quantum channels}, 
Rep. Math. Phys.,  \textbf{55} 61--77 (2005).



\bibitem{alic}
R. Alicki, M. Fannes, 
{\it Note on Multiple Additivity of Minimal R\'{e}nyi Entropy Output of the Werner-Holevo Channels},
Open Sys. $\&$ Information Dyn.,  \textbf{11} 339--342 (2004).



















%
%
%
%
%
%
%
%
%
%
%
%
%
%
%
%
%
%
%


\end{thebibliography}
\end{document}